\documentclass[11pt]{article}
\usepackage{amsfonts}
\usepackage{amssymb}
\usepackage{amsthm}
\usepackage{amsmath}
\usepackage{graphicx}
\usepackage{fixltx2e}
\usepackage{multirow}
\usepackage[numbers]{natbib}

\usepackage[table]{xcolor}
\usepackage{indentfirst}
\usepackage{mathrsfs}
\usepackage{tikz}
\usepackage{graphics}
\usepackage{braket}
\theoremstyle{plain}
\usepackage{enumitem}
\usepackage{upgreek}
\usepackage{siunitx}
\newtheorem{theorem}{\textbf{Theorem}}[section]
\newtheorem{lemma}{\textbf{Lemma}}[section]
\newtheorem{proposition}{\textbf{Proposition}}[section]
\newtheorem{corollary}{\textbf{Corollary}}[section]
\newtheorem{remark}{\textbf{Remark}}[section]
\newtheorem{definition}{\textbf{Definition}}[section]

\allowdisplaybreaks[4]

\def\R3{\mathbb{R}^3} 
\def\R{\mathbb{R}}
\def\F2o{\overline{F_2}}

\def \eps{\varepsilon}

\newcommand{\beq}[0]{\begin{equation}}
\newcommand{\eeq}[0]{\end{equation}}

\newcommand{\Zn}[0]{Zn\textsubscript{45}Au\textsubscript{30}Cu\textsubscript{25}}

\newcommand{\vc}[1]{\mathbf{#1}}
\newcommand{\mt}[1]{\mathsf{#1}}

\DeclareMathOperator{\diag}{diag}

\DeclareMathOperator{\tr}{tr}

\DeclareMathOperator{\cof}{\mt{cof}}

\DeclareMathOperator{\rank}{rank}

\usepackage[hang,flushmargin]{footmisc}
\makeatletter
\def\blfootnote{\gdef\@thefnmark{}\@footnotetext}
\makeatother

\usepackage{arydshln}
\usetikzlibrary{positioning,shapes,shadows,arrows}

\begin{document}

\title{On the cofactor conditions and further conditions of supercompatibility between phases}

\author{
{\sc Francesco Della Porta}\footnote{Mathematical Institute, University of Oxford, Oxford OX2 6GG, UK \textit{dellaporta@maths.ox.ac.uk}}
}
%
%
\date{\vspace{-5ex}}
\maketitle

\blfootnote{\noindent\rule{6.8cm}{0.4pt}\\
\textbf{Acknowledgements:} This work was supported by the Engineering and Physical Sciences Research Council [EP/L015811/1]. The author would like to thank Xian Chen for the useful discussions and the invitation to spend a period in her group at the Hong Kong University of Science and Technology where this work began. The author would also like to thank John Ball for his valuable suggestions and discussions. Thanks also to Xian Chen, Yintao Song and Richard James for allowing me to include Figure \ref{microstrutture star}. The author would
like to acknowledge the two anonymous reviewers for improving this paper with their comments.\\
\textbf{Declarations of interest:} none.
}

\abstract{In this paper we improve the understanding of the cofactor conditions, which are particular conditions of geometric compatibility between austenite and martensite, that are believed to influence reversibility of martensitic transformations. We also introduce a physically motivated metric to measure how closely a material satisfies the cofactor conditions, as the two currently used in the literature can give contradictory results. We introduce a new condition of super-compatibility between martensitic laminates, which potentially reduces hysteresis and enhances reversibility. Finally, we show that this new condition of super-compatibility is very closely satisfied by
\Zn, the first of a class of recently discovered materials, fabricated to closely satisfy the cofactor conditions, and undergoing ultra-reversible martensitic transformation.}

\section{Introduction}
The cofactor conditions are particular conditions of supercompatibility between phases in martensitic transformations. These include, among other conditions, that the middle eigenvalue $\lambda_2$ of the martensitic transformation matrices is equal to one, which has formerly been shown to influence the hysteresis of martensitic transformations (see e.g., \cite{JamesMuller}). The cofactor conditions allow finely twinned martensitic variants to be compatible with austenite, independently of the volume fraction, across a plane. Due to this special compatibility, the cofactor conditions have been conjectured to influence reversibility of the phase transitions, first in \cite{JZ} and later in \cite{JamesHyst}. The fabrication of Zn\textsubscript{45}Au\textsubscript{30}Cu\textsubscript{25}, the first material closely satisfying the cofactor conditions, partially confirms this conjecture (see \cite{JamesNew}). Indeed, both the latent heat of the transformation and the critical temperature in Zn\textsubscript{45}Au\textsubscript{30}Cu\textsubscript{25} do not change significantly over $16000$ thermal cycles (see \cite{JamesNew}). Furthermore, the hysteresis loop in this new material seems to be only very slightly affected after $10^5$ cycles of uniaxial compressive loading (see \cite{ChenLoading}). After Zn\textsubscript{45}Au\textsubscript{30}Cu\textsubscript{25}, other alloys closely satisfying the cofactor conditions have been fabricated (see \cite{Chluba} and \cite{ChlubaJames}), whose hysteresis curve does not significantly change after $10^7$ cycles of uniaxial tension. We refer the reader also to \cite{JamesSur} and \cite{JamesSur2} for two reviews on the topic.

From a theoretical point of view, the cofactor conditions were first introduced in \cite{BallJames1}, as conditions of degeneracy for the equations of the crystallographic theory of martensite (see Theorem \ref{thm cof cond} below). Much later, the cofactor conditions were further investigated from a theoretical point of view in \cite{JamesHyst}, where the authors prove that if a martensitic type I/II twin satisfies the cofactor conditions, then it can form exact phase interfaces, that is with a stress-free transition layer, between austenite and a martensitic laminate, independently of the volume fractions within the laminate (see Theorem \ref{TypeI} and Theorem \ref{TypeII} below). Therefore, from a theoretical point of view, these materials can convert extremely easily a laminate of martensite into austenite and vice versa, without any elastic stress, and hence without the need to incur plastic deformation. The condition $\lambda_2=1$ guarantees already the possibility of having exact stress-free austenite-martensite phase interfaces. However, during martensitic transformations, nucleations often occur at different points of the sample \cite{Inamura}, and, without further compatibility, the single variants of martensite cannot grow further and stay stress-free after they have met, unless plastic deformations take place.

The aim of this paper is to study further the cofactor conditions, with a particular focus on cubic to monoclinic transformations, such as for Zn\textsubscript{45}Au\textsubscript{30}Cu\textsubscript{25}. The case of cubic to orthorhombic transformations, which is relevant for the new materials in \cite{Chluba,ChlubaJames}, is a special case of cubic to monoclinic, and hence all our results apply. {Here and below, by cubic to monoclinic transformation we mean cubic to monoclinc II (see e.g., \cite[Table 4.4]{Batt}), where the axis of monoclinic symmetry corresponds to a $\langle1 0 0\rangle_{cubic}$ direction. In the same way, for cubic to orthorhombic transformations we implicitly assume that the axis of orthorhombic symmetry corresponds to a $\langle1 0 0\rangle_{cubic}$ direction (see e.g., \cite[Table 4.2]{Batt}).}

{In Section \ref{cofac cond sec} we recall some useful results from twinning theory and the definition of cofactor conditions.}
In Section \ref{cofac cond secTH} we prove that once the cofactor conditions are satisfied by a type I/II twin, they are satisfied by all symmetry related twins. As a consequence, in the cubic to monoclinic case, if a twin satisfies the cofactor conditions, then eleven other different twins enjoy the same property (see Table \ref{Table 1}). This generalises previous work in \cite{JamesHyst} (and Table 2 therein) where only three such twins were noted. The presence of multiple twins satisfying the cofactor conditions gives to a material many possibilities to change phase without inducing any elastic stress. We prove that the same pair of martensitic variants cannot satisfy at the same time the cofactor condition as a type I and a type II twin. This result is puzzling because, as discussed below and in Section \ref{metricclose}, in Zn\textsubscript{45}Au\textsubscript{30}Cu\textsubscript{25} the same pair of martensitic variants seems to satisfy very closely the cofactor conditions with both the type I and the type II twin generated by the same pair of martensitic variants. However, being close is a matter of metric, that is, of how we measure the cofactor conditions. Indeed, as the cofactor conditions can never be satisfied exactly by real materials, in Section \ref{metricclose} we discuss how these are measured in the literature, and introduce a new metric which we believe to be related to reversibility. We find that in Zn\textsubscript{45}Au\textsubscript{30}Cu\textsubscript{25}, the stress required to deform austenite in such a way that is exactly compatible (that is with no interface layer) with a laminate of martensite is very small for type II twins, and almost ten times bigger for type I twins. Therefore, according to our new metric, it seems that Zn\textsubscript{45}Au\textsubscript{30}Cu\textsubscript{25} satisfies the cofactor conditions with type II twins much better than with type I twins.

{In Section \ref{qcsec} we study the possible homogeneous average deformations (also called constant macroscopic  deformation gradients in the literature \cite{BallJames2}) that can be obtained by finely mixing two unstressed martensitic variants, and we study which are compatible with austenite across a plane. Surprisingly, if the cofactor conditions are satisfied just by the type I (or just by the type II) twinning system the set of average deformation gradients which are compatible with austenite across a plane is of the same dimension as in standard shape-memory alloys.}

{In Section \ref{StarSec} we introduce a new condition of super-compatibility between phases which supplements the cofactor conditions. We call the twins satisfying these conditions {\it star twins}. {Let $\mt V_1,\mt V_2\in\R^{3\times3}_{Sym^+}$ be two deformation gradients related to two different martensitic variants. Let $\vc b_{12},\vc m_{12}\in\R^3$, and $\mt R_{12}\in SO(3)$ be a solution to the twinning equation for $\mt V_1,\mt V_2$ (that is such that \eqref{intror3} below is satisfied, but see also Section \ref{cofac cond sec} for further details). Here, $SO(3)$ denotes the group of rotations, while $\vc b_{12},\vc m_{12}$ characterise the twinning elements: the twinning shear is given by $s=|\vc b_{12}||\mt V_1^{-1}\vc m_{12}|$, $\vc b^*_{12}=\frac{\vc b_{12}}{|\vc b_{12}|}$ is the direction of shear, and $\vc m^*_{12}=\frac{\mt V_1^{-1}\vc m_{12}}{|\mt V_1^{-1}\vc m_{12}|}$ gives the normal to the twin plane, so that $\mt R_{12}\mt V_2 = (\mt 1+s\vc b_{12}^*\otimes \vc m_{12}^*)\mt V_1$ (see \cite{Batt}). 

If $(\mt V_1,\vc b_{12},\vc m_{12})$ satisfies the cofactor conditions (see (CC1)--(CC3) and Theorem \ref{thm cof cond} in Section \ref{cofac cond sec} below), then Theorem 7 and Theorem 8 in \cite{JamesHyst} (reported here as Theorem \ref{TypeI} and Theorem \ref{TypeII}) imply the existence of $\vc a_1,\vc a_2,\vc n_1,\vc n_2\in\R^3,$ and $\hat{\mt R}_{12}\in SO(3)$ such that}
\begin{align}
\label{intror1}
\hat{\mt R}_{12}\bigl((1-\mu)\mt V_1	+ \mu \mt R_{12}\mt V_2\bigr) &= \mt 1 + \vc b_{12} \otimes \bigl((1-\mu) \vc n_1 + \mu\vc n_2\bigr),\qquad\text{for type I twins},\\
\label{intror2}
\hat{\mt R}_{12}\bigl((1-\mu)\mt V_1	+ \mu \mt R_{12}\mt V_2\bigr) &= \mt 1 +  \bigl((1-\mu) \vc a_1 + \mu\vc a_2\bigr)\otimes \vc m_{12},\qquad\text{for type II twins},\\
\label{intror3}
\mt R_{12}\mt V_2 - \mt V_1 &=  \vc b_{12}\otimes \vc m_{12},
\end{align}
for every $\mu\in[0,1]$. {In \eqref{intror1} and \eqref{intror2}, for every $\mu\in[0,1]$ the triples $(\mu,\vc b_{12},{(1-\mu) \vc n_1 + \mu\vc n_2})$ and $(\mu,(1-\mu) \vc a_1 + \mu\vc a_2,\vc m_{12})$ are solutions to the equation of the phenomenological theory of martensite crystallography for $\mt V_1,\mt V_2$ (see \eqref{auste marte} below or \cite{Batt}).} For any fixed $\mu_0\in[0,1]$ the average deformation gradient
\beq
\label{laminateNoInter}
\mt F_0 = \hat{\mt R}_{12}\bigl((1-\mu_0)\mt V_1	+ \mu_0 \mt R_{12}\mt V_2\bigr),
\eeq
is compatible with austenite without an interface layer (cf. Figure \ref{CC Standard}), and can hence easily propagate in austenite. We call such a constructed $\mt F_0$, an \textit{exactly compatible} laminate generated by $\mt V_1,\mt V_2$, of volume fraction $\mu_0$. 
\begin{figure}
\centering
\begin{tikzpicture}[scale = 0.8]
\draw[thin] (-1,4) -- (-4,1);
\draw[thin] (-1,4) -- (-1,1);
\draw[thin] (-1,4) -- (-5,3);
\draw[thin] (-7,1) -- (-5,3);
\draw[thin] (-5,3) -- (-5,6);
\draw[thin] (-8,5.25) -- (-5,6);
\draw[thin] (-8,3) -- (-5,6);

\draw [->,thick] (-6.5,6-0.375) -- (-6.5-0.4472/2,6-0.375+0.4472*4/2); 
\filldraw [red] (-6.5-0.4472/2,6-0.375+0.4472*4/2) circle (0pt) node[anchor=west,black] {$\vc o_a$};

\draw [->,thick] (-5,4.5) -- (-5+0.922,4.5); 
\filldraw [red] (-5+0.922,4.5) circle (0pt) node[anchor=west,black] {$\vc o_b$};

\draw [->,thick] (-6,2) -- (-6-0.922/1.4142,2 +0.922/1.4142 ); 
\filldraw [red] (-6-0.922/1.4142,2 +0.922/1.4142) circle (0pt) node[anchor=west,black] {$\vc m$};

\filldraw [red] (-7.5,4.75) circle (0pt) node[anchor=north,black] {$\tilde{\mt R}_a\mt U$};
\filldraw [red] (-6.,3.85) circle (0pt) node[anchor=north,black] {$\tilde{\mt R}_b\mt V$};
\filldraw [red] (-4.,2.65) circle (0pt) node[anchor=north,black] {$\tilde{\mt R}_a\mt U$};
\filldraw [red] (-2.5,1.75) circle (0pt) node[anchor=north,black] {$\tilde{\mt R}_b\mt V$};

\draw[thin] (2,5) -- (8,7);
\draw[thin] (2+1.4/3,5-1.4) -- (8+1.4/3,7-1.4);
\draw[thin] (2+2.5/3,5-2.5) -- (8+2.5/3,7-2.5);
\draw[thin] (2+3.9/3,5-3.9) -- (8+3.9/3,7-3.9);

\filldraw [red] (5-0.5/3,6-0.5) circle (0pt) node[anchor=north,black] {$\tilde{\mt R}_a\mt U$};
\filldraw [red] (5+0.7/3,6-1.7) circle (0pt) node[anchor=north,black] {$\tilde{\mt R}_b\mt V$};
\filldraw [red] (5+2/3,6-3) circle (0pt) node[anchor=north,black] {$\tilde{\mt R}_a\mt U$};
\filldraw [red] (-2.5,1.75) circle (0pt) node[anchor=north,black] {$\tilde{\mt R}_b\mt V$};

\draw [->,thick] (6,5+4/3) -- (6-0.922/3.1623,5+4/3+3*0.922/3.1623); 
\filldraw [red] (6-0.922/3.1623,5+4/3+3*0.922/3.1623) circle (0pt) node[anchor=west,black] {$\vc m$};

\filldraw [red] (-3,5.2) circle (0pt) node[anchor=south,black] {$\mt 1$};
\filldraw [red] (4,6.2) circle (0pt) node[anchor=south,black] {$\mt 1$};

\filldraw [red] (-4,7.) circle (0pt) node[anchor=south east,black] {Austenite};
\filldraw [red] (5.5,8.) circle (0pt) node[anchor=south east,black] {Austenite};
\filldraw [red] (-9,2.) circle (0pt) node[anchor=south east,black,rotate=270] {Martensite};
\filldraw [red] (9.0,3.5) circle (0pt) node[anchor=south east,black,rotate=270] {Martensite};

\end{tikzpicture}
\caption{\label{CC Standard} Example of exactly compatible laminates respectively when the cofactor conditions are satisfied by type I and by type II twins (cf. Theorem \ref{TypeI} and Theorem \ref{TypeII}). {In this picture, $\mt 1$ represents the identity matrix, the deformation gradient for the undistorted austenite phase, and denotes the austenite region. The regions denoted by $\tilde{\mt R}_a\mt U,\tilde{\mt R}_b\mt V$ are regions occupied by martensite.}}
\end{figure}
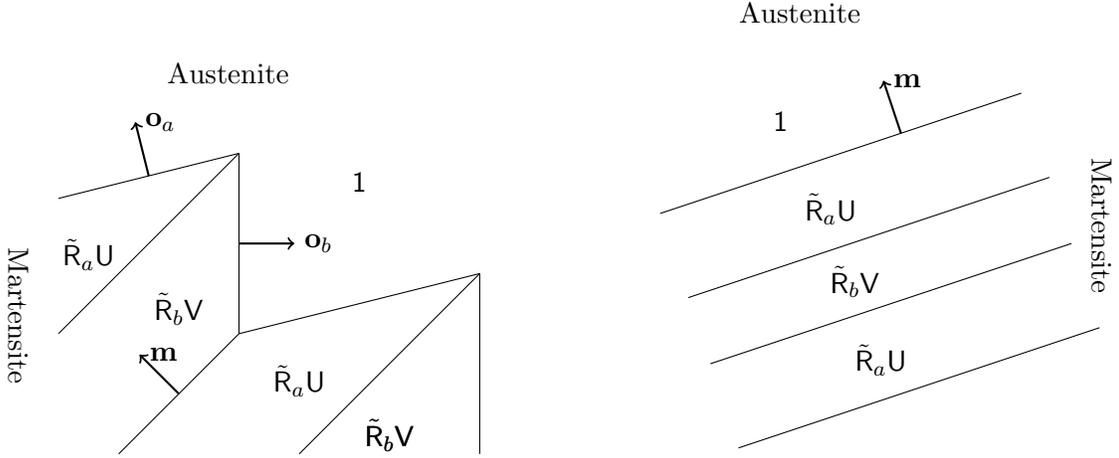
In cubic to monoclinic phase transitions, in general, given $\mt F_0$ as in \eqref{laminateNoInter}, there exists no exactly compatible laminate $\mt F_1$ generated by the martensitic variants $\mt V_3,\mt V_4\in\R^{3\times3}_{Sym^+}$, of volume fraction $\mu_1\in[0,1]$, and such that
\beq
\label{comptralam}
\rank \bigl(\mt F_0 -\mt F_1) =  1,
\eeq
that is, such that $\mt F_0$ and $\mt F_1$ are different but compatible across a plane. Such compatibility is possible only for specific values of $\mu_0,\mu_1$. However, in general, given $\mt F_0$ as in \eqref{laminateNoInter} there exists at most one exactly compatible laminate $\mt F_1$, such that \eqref{comptralam} is satisfied. Exceptions to this fact occur when $(\mt V_1,\vc b_{12},\vc m_{12})$ is a \textit{star-twin} (or a \textit{half-star twin}), {introduced in Definition \ref{Def star I} and Definition \ref{Def star} below}. In this case, there exist values of $\mu_0$ such that, given $\mt F_0$ as in \eqref{laminateNoInter}, there exist three exactly compatible laminates $\mt F_1,\mt F_2,\mt F_3\in\R^{3\times3}$ (resp. two exactly compatible laminates in the case of half-star twins) with 
\beq
\label{startwinsDefIntro}
\rank \bigl(\mt F_i -\mt F_j)= 1, \qquad\text{for every $i,j=0,1,2,3$.}
\eeq
Furthermore, any three of $\mt F_0,\mt F_1,\mt F_2,\mt F_3$ are linearly independent (see Figure \ref{starlaminatesII}).
For general cubic to monoclinic transformations there are four deformation parameters determining the transformation strain (cf. $a,b,c,d$ in \eqref{cubictomono}). The cofactor conditions impose two relations between these four parameters, so that there is a two-parameter family that satisfies them. For the existence of star twins a further relation has to be satisfied, reducing the set of possible deformation parameters to a one-dimensional family.
The compatibility between different laminates which form an exact interface with austenite is only on average. Nevertheless, this allows three different laminates, nucleated in different regions of the sample, to grow further after they meet, and not to stop due to incompatibility. Also, as emphasised in Remark \ref{3dInterf}, in the presence of type II star twins macroscopically curved interfaces whose normal does not lie in a plane are possible between austenite and martensite without an interface layer (cf. also Figure \ref{starCrosslaminatesII} and Figure \ref{starCrosslaminatesIIbis}).\\
\begin{figure}
\centering
\begin{tikzpicture}[scale=0.80]
\draw[thin] (-0.6,6) -- (-6,5.1);
\draw[thin] (-0.6,5.1) -- (-6,4.2);
\draw[thin] (-0.6,4) -- (-6,3.1);
\draw[thin] (-0.6,3.1) -- (-6,2.2);
\draw[thin] (-0.6,2) -- (-6,1.1);
\draw[thin] (-0.6,1.1) -- (-6,0.2);
\draw[thin] (-0.6,0) -- (-6,-1);
\draw[thin] (0.6,6) -- (6,5.1);
\draw[thin] (0.6,5.1) -- (6,4.2);
\draw[thin] (0.6,4) -- (6,3.1);
\draw[thin] (0.6,3.1) -- (6,2.2);
\draw[thin] (0.6,2) -- (6,1.1);
\draw[thin] (0.6,1.1) -- (6,0.2);
\draw[thin] (0.6,0) -- (6,-1);
\filldraw [red] (0,1.5) circle (0pt) node[anchor=west,black, rotate=90] {Interface layer};

\draw[thin] (12.6,{6 + 3}) -- (7.2,5.1);
\draw[thin] (12.6,{5.1 + 3}) -- (7.2,4.2);
\draw[thin] (12.6,{4 + 3}) -- (7.2,3.1);
\draw[thin] (12.6,{3.1 + 3}) -- (7.2,2.2);
\draw[thin] (12.6,{2 + 3}) -- (7.2,1.1);
\draw[thin] (12.6,{1.1 + 3}) -- (7.2,0.2);
\draw[thin] (12.6,{0 + 3}) -- (7.2,-1);
\filldraw [red] (6.6,1.5) circle (0pt) node[anchor=west,black, rotate=90] {Interface layer};

\draw [->,thick] ({-2.4-0.2},5.6) -- ({-2.55-0.2},6.5); 
\filldraw [red] (-2.5,6.3) circle (0pt) node[anchor=west,black] {$\vc n_\alpha$};
\draw [->,thick] ({2.4+0.2},5.6) -- ({2.55+0.2},6.5); 
\filldraw [red] (2.3,6.3) circle (0pt) node[anchor=east,black] {$\vc n_\beta$};
\draw [->,thick] (9.9,7.05) -- ({9.9-0.53},{7.05 + 0.74}); 
\filldraw [red] ({9.9-0.53},{7.05 + 0.74}) circle (0pt) node[anchor=east,black] {$\vc n_\gamma$};

\filldraw [red] (-5,4.4) circle (0pt) node[anchor=south,black] {$\mt R_1\mt V_1$};
\filldraw [red] (-5,3.5) circle (0pt) node[anchor=south,black] {$\mt R_2\mt V_2$};
\filldraw [red] (-5,2.4) circle (0pt) node[anchor=south,black] {$\mt R_1\mt V_1$};
\filldraw [red] (-5,1.5) circle (0pt) 
node[anchor=south,black] {$\mt R_2\mt V_2$};
\filldraw [red] (-5,0.4) circle (0pt) node[anchor=south,black] {$\mt R_1\mt V_1$};
\filldraw [red] (-5,-0.5) circle (0pt) 
node[anchor=south,black] {$\mt R_2\mt V_2$};

\filldraw [red] (5,4.4) circle (0pt) node[anchor=south,black] {$\mt R_3\mt V_3$};
\filldraw [red] (5,3.5) circle (0pt) node[anchor=south,black] {$\mt R_4\mt V_4$};
\filldraw [red] (5,2.4) circle (0pt) node[anchor=south,black] {$\mt R_3\mt V_3$};
\filldraw [red] (5,1.5) circle (0pt) 
node[anchor=south,black] {$\mt R_4\mt V_4$};
\filldraw [red] (5,0.4) circle (0pt) node[anchor=south,black] {$\mt R_3\mt V_3$};
\filldraw [red] (5,-0.5) circle (0pt) 
node[anchor=south,black] {$\mt R_4\mt V_4$};

\filldraw [red] (8.5,5.1) circle (0pt) node[anchor=south,black,rotate=35] {$\mt R_5\mt V_5$};
\filldraw [red] (8.5,{5.1-0.9}) circle (0pt) node[anchor=south,black,rotate=35] {$\mt R_6\mt V_6$};
\filldraw [red] (8.5,{5.1-2}) circle (0pt) node[anchor=south,black,rotate=35] {$\mt R_5\mt V_5$};
\filldraw [red] (8.5,{5.1-2.9}) circle (0pt) 
node[anchor=south,black,rotate=35] {$\mt R_6\mt V_6$};
\filldraw [red] (8.5,{5.1-4}) circle (0pt) node[anchor=south,black,rotate=35] {$\mt R_5\mt V_5$};
\filldraw [red] (8.5,{5.1-4.9}) circle (0pt) 
node[anchor=south,black,rotate=35] {$\mt R_6\mt V_6$};

\filldraw [red] (-4,6.2) circle (0pt) node[anchor=south,black] {$\mt 1$};
\filldraw [red] (4,6.2) circle (0pt) node[anchor=south,black] {$\mt 1$};
\filldraw [red] ({9.9-0.53+1},{7.05 + 0.74+0.5}) circle (0pt) node[anchor=south,black] {$\mt 1$};

\draw [<->,thin] (-3.89,5.29) -- (-3.78,4.63); 
\filldraw [red] (-3.78,4.63) circle (0pt) node[anchor=south west,black] {$\eps\mu$};
\draw [<->,thin] (-3.76,4.51) -- (-3.6,3.55); 
\filldraw [red] (-3.6,3.55) circle (0pt) node[anchor=south west,black] {$\eps(1-\mu)$};
\draw [<->,thin] (3.89,5.29) -- (3.78,4.63); 
\filldraw [red] (3.78,4.63) circle (0pt) node[anchor=south east,black] {$\eps\mu$};
\draw [<->,thin] (3.76,4.51) -- (3.6,3.55); 
\filldraw [red] (3.6,3.55) circle (0pt) node[anchor=south east,black] {$\eps(1-\mu)$};

\draw [<->,thin] (10.5,7.48) -- ({10.5+0.53*0.7},{7.48-0.74*0.7}); 
\filldraw [red] ({10.5+0.53*0.75},{7.48-0.74*0.75}) circle (0pt) node[anchor=south east,black,rotate=35] {$\eps\mu$};
\draw [<->,thin] ({10.5+0.53*1.7},{7.48-0.74*1.7}) -- ({10.5+0.53*0.8},{7.48-0.74*0.8}); 
\filldraw [red] ({10.5+0.53*1.75},{7.48-0.74*1.75}) circle (0pt) node[anchor=south east,black,rotate=35] {$\eps(1-\mu)$};

\filldraw [red] (5.5,8.) circle (0pt) node[anchor=south east,black] {Austenite};
\filldraw [red] (-7.2,1.) circle (0pt) node[anchor=south east,black,rotate=270] {Martensite};
\end{tikzpicture}
\caption{\label{starlaminatesII} Example of average compatibility between three different exactly compatible laminates occurring in type II star twins. {In this picture, $\mt 1$ represents the identity matrix, the deformation gradient for the undistorted austenite phase, and denotes the austenite region. Here, $\mt R_1,\dots,\mt R_6\in SO(3)$ are rotation matrices, and $\mt V_1,\dots, \mt V_6\in \R^{3\times 3}_{Sym^+}$ represent six deformation gradients corresponding to six martensitic variants (possibly $\mt V_i=\mt V_j$ for $i\neq j$).} We have $\mt R_i\mt V_i = \mt 1 + \vc a_i\otimes \vc n_\alpha$ for $i=1,2$, $\mt R_i\mt V_i = \mt 1 + \vc a_i\otimes \vc n_\beta$ for $i=3,4$, and $\mt R_i\mt V_i = \mt 1 + \vc a_i\otimes \vc n_\gamma$ for $i=5,6$, for some $\vc a_1,\dots,\vc a_6\in\R^3,\vc n_\alpha,\vc n_\beta,\vc n_\gamma\in\mathbb{S}^2$. Furthermore, there exist $\mu\in(0,1),\vc a^*\in\R^3$ such that $\mu \mt R_1\mt V_1 +(1-\mu)\mt R_2\mt V_2 = \mt 1 + \vc a^*\otimes \vc n_\alpha$, $\mu \mt R_3\mt V_3 +(1-\mu)\mt R_4\mt V_4 = \mt 1 + \vc a^*\otimes \vc n_\beta$ and $\mu \mt R_5\mt V_5 +(1-\mu)\mt R_6\mt V_6 = \mt 1 + \vc a^*\otimes \vc n_\gamma$, where $\rank(\mt R_{2i}\mt V_{2i}-\mt R_{2i-1}\mt V_{2i-1})=1$ for $i=1,2,3$. In the centre we have a transition layer which makes the two laminates exactly compatible in the limit $\eps\to 0.$ This image is a projection on a plane, but we remark that $\vc n_\alpha,\vc n_\beta,\vc n_\gamma$ are linearly independent. {The important novelty of star-twins is given by the fact that, unlike well-studied wedge-microstructures (see e.g., \cite[Sec. 7.3]{Batt}), the laminates in Figure \ref{starlaminatesII} above are compatible with austenite without an interface layer. Furthermore, in the case of star twins, the number of different laminates both compatible with austenite and with each other is four and not just two. {Due to this very special condition of compatibility, which occurs just for special values of the deformation parameters (cf. $a,b,c,d$ in \eqref{cubictomono}), and which adds on the cofactor conditions, these twins are called star-twins.}}
}
\end{figure}
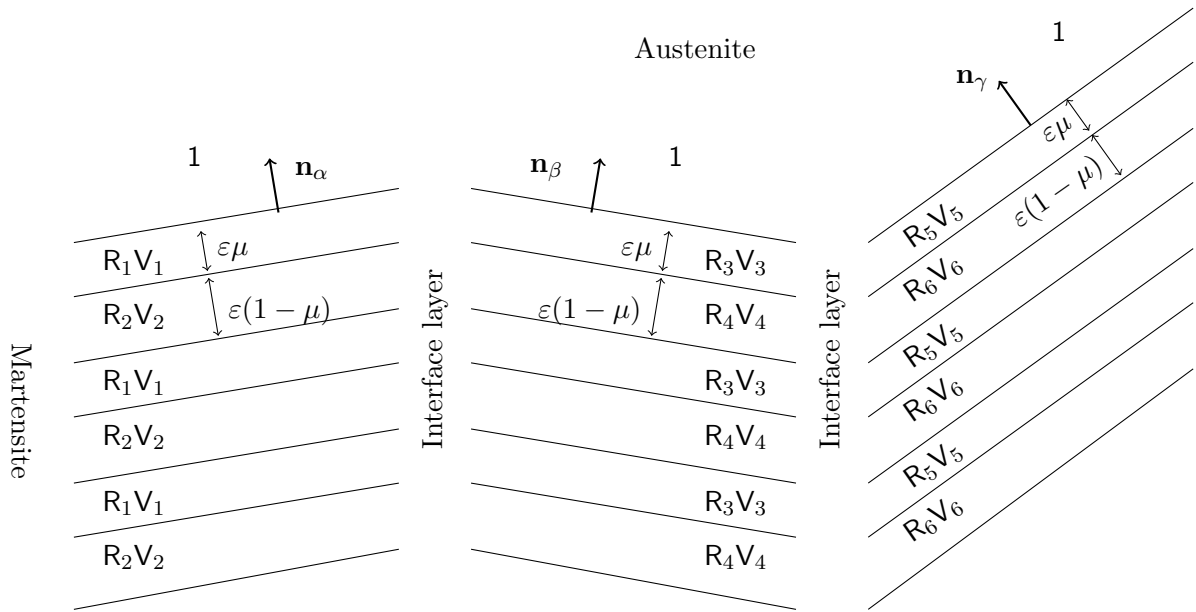
}
In Section \ref{starsecMono} we show under which conditions on the eigenvalues of the transformation matrices for cubic to monoclinic phase transitions a twin satisfying the cofactor conditions is actually a star twin. It is striking to notice that in Zn\textsubscript{45}Au\textsubscript{30}Cu\textsubscript{25}, the smallest eigenvalue is approximately $0.9363$ and the largest is $1.0600$. In order to have a type II star twin, the largest should be $1.0609$, so that the error is about $9\cdot10^{-4}$, very similar to the approximate error of $6\cdot10^{-4}$ which separates $\lambda_2$ from $1$ in this material. Denoting by $\mt U_1^r$ the measured deformation gradient in Zn\textsubscript{45}Au\textsubscript{30}Cu\textsubscript{25} related to the first martensitic variant (see \eqref{cubictomono}), {and by $\mt U_1^e$ the same matrix, but allowing a type II star twin, we have 
$$
\mt U_1^r = \left[\begin{array}{ ccc } 1.0015& 0.0073 & 0\\     0.0073&    1.0591 &  0\\
   0  & 0  &  0.9363 \end{array}\right],\qquad
\mt U_1^e = \left[\begin{array}{ ccc } 1.0010& 0.0078 & 0\\     0.0078&    1.0594 &  0\\
   0 &  0  &  0.9368 \end{array}\right],
$$
which are very close as $\|\mt U_1^r-\mt U_1^e\|\approx 1.1\cdot 10^{-3}$. As a comparison, the closest matrix to $\mt U_1^r$, say $\mt U_1^c$, satisfying the cofactor conditions and describing a cubic to monoclinic phase transformation is such that $\|\mt U_1^r-\mt U_1^e\|\approx 0.9\cdot 10^{-3}$. }
{Experimental images of martensitic microstructures for \Zn\, show the presence of inexact junctions between three different laminates (see Figure \ref{microstrutture star}). The influence of star twins on these microstructures needs however to be confirmed with further experimental investigations. 
\begin{figure}[h]
    \includegraphics[width=1.0\textwidth]{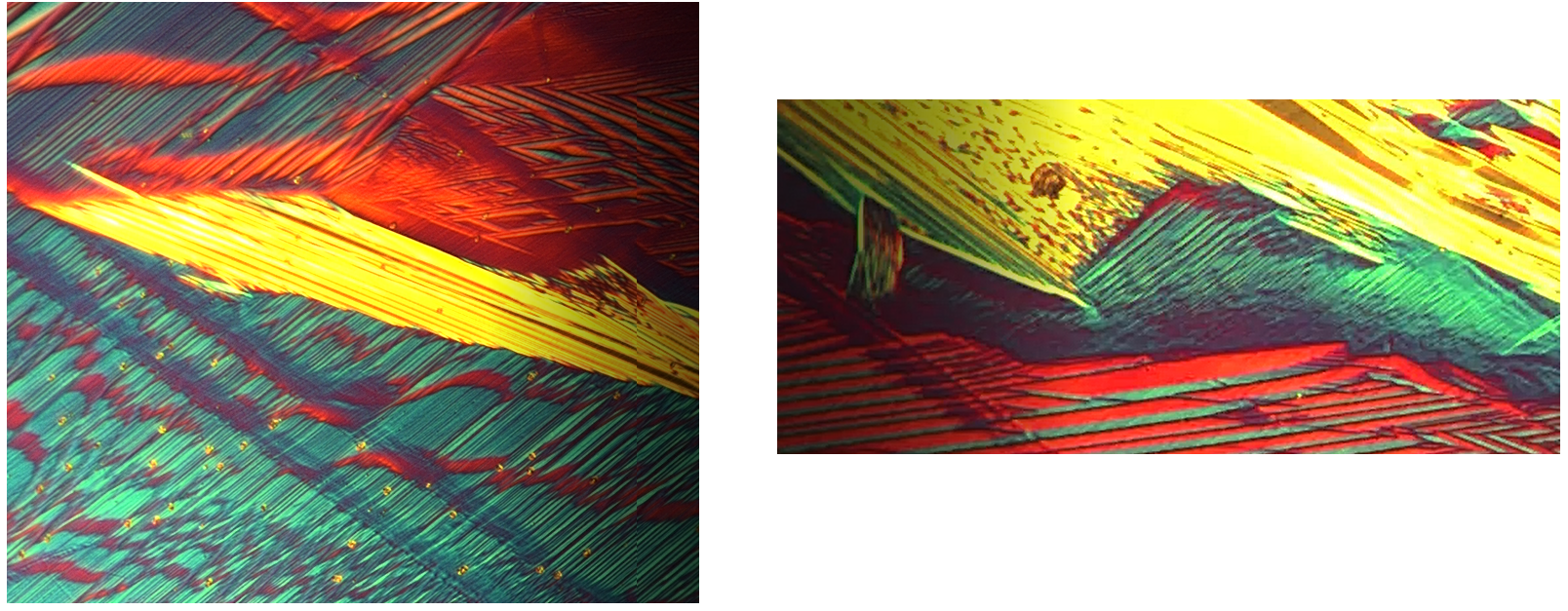}
    \caption{\label{microstrutture star} Examples of junctions between three different laminates observed in \Zn\,(cf. Figure \ref{starlaminatesII}). Courtesy of Xian Chen, Yintao Song and Richard James.}
\end{figure}
We remark that in Zn\textsubscript{45}Au\textsubscript{30}Cu\textsubscript{25} the type I twins are not close to being star twins. }
As proved in \cite{FDP2}, in a first approximation the average deformation gradients in the martensite phase of Zn\textsubscript{45}Au\textsubscript{30}Cu\textsubscript{25} are of the form
\beq
\label{DP result}
\mt 1 + \vc a(\vc x)\otimes\vc n(\vc x),\qquad \text{a.e. $\vc x$}.
\eeq
The presence of star twins makes it extremely easy to construct very complex average deformation gradients of the form \eqref{DP result}, and this might explain the presence of such colourful microstructures, as well as the  ability of this material to \textit{``perform a much wider and more efficient collection of adjustments of microstructure to environmental changes''} (ref. \cite{JamesSur2}).\\

\section{Twinning theory and the cofactor conditions}
\label{cofac cond sec}
In this section we recall the basic results from twinning theory, and we introduce the cofactor conditions following closely \cite{JamesHyst} and \cite{FDP2}.\\ 

Let $\R^{3\times3}_{Sym^+}$ be the set of symmetric positive definite $3\times3$ matrices, and let $\mt U,\mt V\in \R^{3\times3}_{Sym^+}$ be two matrices describing the change of lattice from austenite to two variants of martensite. As the martensitic variants are symmetry related, there exists a rotation $\mt R\in SO(3)$ satisfying $\mt V=\mt {RU}\mt R^T$. The first useful result is the following:
\begin{proposition}[{\cite[Prop. 12]{JamesHyst}}]
\label{EsisteEseR1}
Let $\mt U,\mt V\in\R^{3\times3}_{Sym^+}$, $\mt V=\mt {RU}\mt R^T$ for some $\mt R\in SO(3)$. 
Suppose further that they are compatible in the sense that there is a matrix $\hat {\mt R}\in SO(3)$ such that
\beq
\label{compatib condit}
\hat{\mt R}\mt U-\mt V= \vc b\otimes \vc m,
\eeq
$\vc b$, $\vc m\in\R^3$. Then there is a unit vector $\hat{\vc e}\in\R^3$ such that
\beq
\label{e hat}
\mt V = (-\mt {1}+2\hat{\vc e}\otimes \hat{\vc e})\mt U(-\mt {1}+2\hat{\vc e}\otimes \hat{\vc e}).
\eeq
Conversely, if \eqref{e hat} is satisfied, then there exist {exactly two solutions $(\hat{\mt R},\vc b\otimes\vc m)$, with $\hat{\mt R}\in SO(3)$, $\vc b,\vc m\in\R^3$ such that \eqref{compatib condit} holds.}
\end{proposition}
Equation \eqref{compatib condit} is called the compatibility condition for two variants of martensite, and, if \eqref{e hat} holds, has always two solutions $(\hat{\mt  R}_I, \vc b_I\otimes \vc m_I)$ and $(\hat {\mt R}_{II},\vc b_{II}\otimes \vc m_{II})$, where $\vc b_I,\vc b_{II},\vc m_I,\vc m_{II}$ are given by (see e.g., \cite{Batt})
\begin{align}
\label{10}
&\vc m_I = \hat{\vc e}, \qquad &&\vc b_{I}=2\Big(\frac{\mt U^{-1}\hat{\vc e}}{|\mt U^{-1}\hat{\vc e}|^2}-\mt U\hat{\vc e}\Big),\\
\label{11}
&\vc m_{II} = 2\Big(\hat{\vc e}-\frac{\mt U^{2}\hat{\vc e}}{|\mt U\hat{\vc e}|^2}\Big),\qquad &&\vc b_{II}=\mt U\hat{\vc e},
\end{align}
and where $\hat{\vc e}$ is as in \eqref{e hat}. If $\hat{\vc  e}$ satisfying \eqref{e hat} is unique up to change of sign, $(\mt U,\vc b_I,\vc m_I)$ and $(\mt U,\vc b_{II},\vc m_{II})$ are respectively called type I and type II twins generated by $\mt U,\mt V$. {In case there exist two different non-parallel unit vectors satisfying \eqref{e hat}, the resulting pair of solutions $(\hat{\mt  R}_I, \vc b_I\otimes \vc m_I)$ and $(\hat {\mt R}_{II},\vc b_{II}\otimes \vc m_{II})$ are called compound twins. It is possible to prove that, even if $\hat{\vc  e}$ satisfying \eqref{e hat} is not unique, there exist just two solution to \eqref{compatib condit}, each of which can be considered as both a type I and a type II twin. More precisely, given two different unit vectors satisfying \eqref{e hat}, namely $\hat{\vc e}_1$ and $\hat{\vc e}_2$, then (see e.g., \cite[Prop. 1]{JamesHyst})
$$
\vc b_I^1\otimes \vc m_I^1 :=2\Big(\frac{\mt U^{-1}\hat{\vc e}_1}{|\mt U^{-1}\hat{\vc e}_1|^2}-\mt U_1\hat{\vc e}_1\Big)\otimes\hat{\vc e}_1= \mt U\hat{\vc e}_2 \otimes 2\Big(\hat{\vc e}_2-\frac{\mt U^{2}\hat{\vc e}_2}{|\mt U\hat{\vc e}_2|^2}\Big)=:\vc b_{II}^2\otimes \vc m_{II}^2.
$$} 
Below, when we refer to type I and type II twin, we assume that there exist an up to a change of sign unique unit vector satisfying \eqref{e hat}. Furthermore, we sometimes abuse of notation and write that $\mt U_1,\mt U_2$ generate a compound twin if the solutions of the twinning equations \eqref{compatib condit} are compound twins.
\begin{remark}
\rm
\label{RemSuper}
{The definition of type I, type II and compound twins given above is not the one that can be usually found in the literature (see e.g., \cite{Batt,PitteriZanzotto98}), but the one that is given in \cite{JamesHyst}. 
For the benefit of the reader, we recall that in the literature twins are divided into five different categories (see \cite{PitteriZanzotto98}): conventional generic, which is divided into type I, type II and compound, non-conventional generic and non-conventional non-generic. \textit{Conventional twins} are the solutions $(\mt U,\vc b,\vc m)$ to \eqref{compatib condit} when there exist an unit vector $\vc e\in \mathbb{R}^3$ such that \eqref{e hat} is satisfied, and such that $(2\vc e\otimes\vc e-\mt 1)\in\mathcal{P}_{a}$, $\mathcal{P}_{a}$ being the symmetry group of austenite. If no such $\vc e\in\mathbb{S}^2$ exists, the  solutions $(\mt U,\vc b,\vc m)$ to \eqref{compatib condit} are called \textit{non-conventional twins}. Furthermore, we say that a solution $(\mt U,\vc b,\vc m)$ to \eqref{compatib condit} is a \textit{generic twin}, if its existence does not depend on the particular values of the transformation strains $\mt U,\mt V$, but only on the symmetry relating $\mt U,\mt V$; otherwise, we call the twin \textit{non-generic}. 
To emphasize the difference from generic conventional twins, what we defined above as type I, type II and compound twins are called type I, type II and compound domains in \cite{JamesHyst}. However the word domain is misleading for readers with mathematical background, and we therefore prefer to keep the word twins throughout this manuscript. The notion given here of type I, type II and compound twin  coincides with the classic one in case of generic conventional twins
, and allows to generalise all results below to non-conventional (and possibly non generic (cf. Remark \ref{RemGen})) twins, without any further technicality, or without entering this further categorisation.
}
\end{remark}

Let us now consider a simple laminate, i.e., a constant macroscopic gradient $\nabla \vc y$ equal a.e. to $\lambda \hat {\mt R}\mt V+(1-\lambda)\mt U$ for some $\lambda\in(0,1)$, some $\mathsf{R} \in SO(3)$ and some martensitic variants $\mathsf{U}, \mathsf{V} \in \mathbb{R}^{3\times 3}_{Sym^+}$ such that $\rank (\mathsf{R}\mathsf{U} - \mathsf{V}) = 1$. 
%
Following \cite{BallJames1,JamesHyst} we focus on the possibility for such $\nabla \vc y$
to be compatible with austenite. 
The existence of $(\mt R,\lambda,\vc a\otimes \vc n)$ solving
\beq
\label{auste marte}
\mt R\bigl[\lambda \hat{\mt  R}\mt V+(1-\lambda)\mt U\bigr]-\mt{1}=\mt R\bigl[\lambda(\mt U+\vc b\otimes \vc m)+(1-\lambda)\mt U\bigr]-\mt {1} = \vc a\otimes \vc n,
\eeq
that is a twinned laminate compatible with austenite, was first studied in \cite{Wechsler} and later in \cite{BallJames1}. 
Lattice deformations and parameters of materials that are usually considered in the literature lead to twins with exactly four solutions to equation \eqref{auste marte}. Nonetheless, in some cases the number of solutions can be just zero, one or two, and, under some particular condition on the lattice parameters, as in the case of the material discovered in \cite{JamesNew}, \eqref{auste marte} is satisfied for all $\lambda \in [0,1]$. The following result gives necessary and sufficient conditions for this to hold:
\begin{theorem}[{\cite[Thm. 2]{JamesHyst}}]
\label{thm cof cond}
Let $\mt U,\mt V\in \R^{3\times3}_{Sym^+}$ be distinct and such that there exist $\hat {\mt R}\in SO(3)$ and $\vc b,\vc m\in\R^3\setminus\{\vc 0\}$ satisfying
$$
\hat{\mt  R} \mt U =\mt  V + \vc b\otimes \vc m.
$$
Then, \eqref{auste marte} has a solution $\mt R\in SO(3)$, $\vc a,\vc n\in \R^3$ for each $\lambda\in [0,1]$ if and only if the following \textbf{cofactor conditions} hold:
\begin{enumerate}[label={(CC\arabic*)}]
\item \label{CC1} The middle eigenvalue $\lambda_2$ of $\mt U$ satisfies $\lambda_2=1$,
\item \label{CC2} $ \vc b\cdot \mt U\cof (\mt U^2-\mt {1})\vc m=0$,
\item \label{CC3} $\tr \mt U^2-\det \mt U^2-\frac{1}{4}|\vc b|^2|\vc m|^2-2\geq 0.
$
\end{enumerate}
\end{theorem} 
The condition (CC2) can be rewritten as (see \cite[Coroll. 5]{JamesHyst})
$$
(\vc v_2\cdot \vc b)(\vc v_2\cdot\vc m)= 0,
$$
where $\vc v_2$ is the eigenvector of $\mt U_1$ related to the eigenvalue $\lambda_2=1$. Another equivalent formulation of (CC2) for type I/II twins is given by the following chains of equivalences (see \cite[Prop. 6]{JamesHyst})
\begin{align}
\label{ccperI}
(\vc v_2\cdot \vc b_I)(\vc v_2\cdot\vc m_I)= 0 \Leftrightarrow (\vc v_2\cdot \vc b_I) = 0,\text{ and }(\vc v_2\cdot \vc m_I) \neq 0
\Leftrightarrow |\mt U_1^{-1}\hat{\vc e}|=1,\quad&\text{for type I twins}\\
\label{ccperII}
(\vc v_2\cdot \vc b_{II})(\vc v_2\cdot\vc m_{II})= 0 \Leftrightarrow (\vc v_2\cdot \vc m_{II}) = 0\text{ and }(\vc v_2\cdot \vc b_{II}) \neq 0
\Leftrightarrow |\mt U_1\hat{\vc e}|=1,\quad&\text{for type II twins}
\end{align}
where $\hat{\vc e}$ is given by \eqref{e hat}. 
\begin{remark}
\label{StrictL}
Let $(\mt U,\vc b,\vc m)$ be as in Theorem \ref{thm cof cond} and satisfy \ref{CC1}--\ref{CC3}. Then, \cite[Coroll. 3]{JamesHyst} states that the smallest and the largest eigenvalue of $\mt U$, namely $\lambda_1$ and $\lambda_3$, satisfy $\lambda_1<1<\lambda_3$.
\end{remark}
We now report two results from \cite{JamesHyst} related to the cofactor conditions in type I/II twins. These results state the possibility to have exact stress free interface between a martensitic laminate and austenite as shown in Figure \ref{CC Standard}. {Here and below we denote by $\mathbb{S}^2$ the set of vectors $\vc v\in\R^3$ such that $|\vc v|=1.$}
\begin{theorem}[{\cite[Thm. 7]{JamesHyst}}]
\label{TypeI}
Let $\mt U,\mt V\in\R^{3\times3}_{Sym^+}$ satisfying \eqref{e hat} be such that $\mt R_{I}\in SO(3), \vc b_{I}\in\R^3,\vc m_{I}\in\mathbb{S}^2$ is a type I solution of the twinning equation \eqref{compatib condit}. Suppose further that $(\mt U,\vc b_{I}, \vc m_{I})$ satisfies the cofactor conditions. Then there exist $\mt R_{\mt U},\mt R_{\mt V}$ and $\vc a\in \R^3,\vc n_{\mt U},\vc n_{\mt V}\in\mathbb{S}^2$ such that
$$
\mt R_{\mt U}\mt U=\mt 1+ \vc a\otimes \vc n_{\mt U},\qquad \mt R_{\mt V}\mt V=\mt 1+ \vc a\otimes \vc n_{\mt V},
$$ 
$\mt R_{\mt U}\vc b_I\times\vc a = \vc 0$, $(\vc n_{\mt U}\times\vc n_{\mt V})\cdot \vc m_I = 0$, and $\mt R_{\mt V}=\mt R_{\mt U}\mt R_{I}$.
\end{theorem}

\begin{theorem}[{\cite[Thm. 8]{JamesHyst}}]
\label{TypeII}
Let $\mt U,\mt V\in\R^{3\times3}_{Sym^+}$ satisfying \eqref{e hat} be such that $\mt R_{II}\in SO(3), \vc b_{II}\in\R^3,\vc m_{II}\in\mathbb{S}^2$ is a type II solution of the twinning equation \eqref{compatib condit}. Suppose further that $(\mt U,\vc b_{II}, \vc m_{II})$ satisfies the cofactor conditions. Then there exist $\mt R_{\mt U},\mt R_{\mt V}$ and $\vc a_{\mt U},\vc a_{\mt V}\in \R^3$ such that
$$
\mt R_{\mt U}\mt U=\mt 1+ \vc a_{\mt U}\otimes \vc m_{II},\qquad \mt R_{\mt V}\mt V=\mt 1+ \vc a_{\mt V}\otimes \vc m_{II},
$$ 
and $\mt R_{\mt V}=\mt R_{\mt U}\mt R_{II}$.
\end{theorem}

\section{Martensitic transformations from cubic to monoclinic lattices}
\label{cofac cond secTH}
We start this section by proving Proposition \ref{tutti quanti} below, stating that the presence of a twin system satisfying the cofactor conditions (CC1)--(CC3) often implies the existence of many other twin systems enjoying the same property
\begin{proposition}
\label{tutti quanti}
Let $\mt U,\mt V\in\R^{3\times3}_{Sym^+}$ satisfying \eqref{e hat} be such that $\mt R\in SO(3), \vc b\in\R^3,\vc m\in\mathbb{S}^2$ is a solution of the twinning equation \eqref{compatib condit}. Suppose further that $(\mt U,\vc b, \vc m)$ satisfies the cofactor conditions. Then, for every $\mt Q\in SO(3)$, $\mt Q\mt U\mt Q^T,\mt Q\mt V\mt Q^T\in\R^{3\times3}_{Sym^+}$ are such that $\mt Q\mt R\mt Q^T\in SO(3), \mt Q\vc b\in\R^3,\mt Q\vc m\in\mathbb{S}^2$ is a solution of the twinning equation \eqref{compatib condit} for $\mt Q\mt U\mt Q^T,\mt Q\mt V\mt Q^T$, and $(\mt Q\mt U\mt Q^T,\mt Q\vc b, \mt Q\vc m)$ satisfies the cofactor conditions. Furthermore, if $(\mt R,\vc b,\vc m)$ is a type I solution (or a type II or a compound solution) to the twining equation \eqref{compatib condit}, then so is $(\mt Q\mt R\mt Q^T,\mt Q\vc b,\mt Q\vc m)$.
\end{proposition}
\begin{proof}
Multiplying \eqref{compatib condit} by $\mt Q$ on the left, and by $\mt Q^T$ on the right we get
$$
\mt Q\mt R\mt Q^T\mt Q\mt U\mt Q^T-\mt Q\mt V\mt Q^T= (\mt Q\vc b)\otimes \mt Q\vc m.
$$
Furthermore, \eqref{e hat} becomes
$$
\mt Q\mt V\mt Q^T = \bigl(-\mt {1}+2(\mt Q\hat{\vc e})\otimes (\mt Q\hat{\vc e})\bigr)\mt Q\mt U\mt Q^T\bigl(-\mt {1}+2(\mt Q\hat{\vc e})\otimes (\mt Q\hat{\vc e})\bigr).
$$
Therefore, if $(\mt R, \vc b,\vc m)$ was a type I (or a type II) solution of the twinning equation \eqref{compatib condit} for $\mt U,\mt V$, then so is $(\mt Q\mt R\mt Q^T, \mt Q\vc b,\mt Q\vc m)$ for $\mt Q\mt U\mt Q^T,\mt Q\mt V\mt Q^T$. If the solutions of \eqref{compatib condit} for $\mt U,\mt V$ are compound twins, then so are the ones for $\mt Q\mt U\mt Q^T,\mt Q\mt V\mt Q^T$. Also, it is clear that (CC1) and (CC3) are satisfied.\\

It just remains to prove that (CC2) holds. But (CC2) can be rewritten as 
$$
(\lambda_1^2-1)(\lambda_3^2-1)(\vc b\cdot\vc v_2)(\vc m\cdot\vc v_2)=0.
$$
Here $\lambda_1\leq \lambda_2 \leq \lambda_3$ are the eigenvalues of $\mt U$, and $\vc v_2$ is the eigenvector of $\mt U$ related to the eigenvalue $\lambda_2=1.$ But as 
$\mt Q\mt U\mt Q^T \mt Q\vc v_2=\mt Q\vc v_2$, we get that
$$
(\lambda_1^2-1)(\lambda_3^2-1)(\mt Q\vc b\cdot\mt Q\vc v_2)(\mt Q\vc m\cdot\mt Q\vc v_2)=(\lambda_1^2-1)(\lambda_3^2-1)(\vc b\cdot\vc v_2)(\vc m\cdot\vc v_2)=0,
$$
which concludes the proof.
\end{proof} 
In case of cubic to monoclinic transformations, the twelve transformation matrices are given by (see \cite[Table 4.4]{Batt})
{\footnotesize
\beq
\label{cubictomono}
\begin{split}
\mt U_1 = \left[\begin{array}{ ccc } a & b & 0 \\ b & c & 0 \\ 0 & 0 & d \end{array}\right],\quad
\mt U_2 = \left[\begin{array}{ ccc } a & -b & 0 \\ -b & c & 0 \\ 0 & 0 & d \end{array}\right],\quad
\mt U_3 = \left[\begin{array}{ ccc } c & b & 0 \\ b & a & 0 \\ 0 & 0 & d \end{array}\right],\quad
\mt U_4 = \left[\begin{array}{ ccc } c & -b & 0 \\ -b & a & 0 \\ 0 & 0 & d \end{array}\right],\\
\mt U_5 = \left[\begin{array}{ ccc } a & 0 & b \\ 0 & d & 0 \\ b & 0 & c \end{array}\right],\quad
\mt U_6 = \left[\begin{array}{ ccc } a & 0 & -b \\ 0 & d & 0 \\ -b & 0 & c \end{array}\right],\quad
\mt U_7 = \left[\begin{array}{ ccc } c & 0 & b \\ 0 & d & 0 \\ b & 0 & a \end{array}\right],\quad
\mt U_8 = \left[\begin{array}{ ccc } c & 0 & -b \\ 0 & d & 0 \\ -b & 0 & a \end{array}\right],\\
\mt U_9 = \left[\begin{array}{ ccc } d & 0 & 0 \\ 0 & a & b \\ 0 & b & c \end{array}\right],\quad
\mt U_{10} = \left[\begin{array}{ ccc } d & 0 & 0 \\ 0 & a & -b \\ 0 & -b & c \end{array}\right],\quad
\mt U_{11} = \left[\begin{array}{ ccc } d & 0 & 0 \\ 0 & c & b \\ 0 & b & a \end{array}\right],\quad
\mt U_{12} = \left[\begin{array}{ ccc } d & 0 & 0 \\ 0 & c & -b \\ 0 & -b & a \end{array}\right].
\end{split}
\eeq}
We recall that here the axis of monoclinic symmetry corresponds to a $\langle1 0 0\rangle_{cubic}$ direction.
\begin {table}[h]
\begin{center}
\footnotesize
\begin{tabular}{ |l|l|l|l|l| }
\hline
$\mt R (\theta,\vc v)$ & type I/II twins (A) & type I/II twins (B) & compound twins (C)\\ \hline
 $\pi,(1,0,0)$ & & &$(1,2),(5,6)$\\
 $\pi,(1,0,0)$ & & &$(3,4),(7,8)$\\
 $\pi,(0,1,0)$ & & &$(1,2),(11,12)$\\
 $\pi,(0,1,0)$ & & &$(3,4),(9,10)$\\
 $\pi,(0,0,1)$ & & &$(5,6),(9,10)$\\
 $\pi,(0,0,1)$ & & &$(11,12),(7,8)$\\
\cline{4-4}
 $\pi,(1,0,1)$ &$(1,12),(2,11)$ &$(3,10),(4,9)$ & $(5,7),(6,8)$\\
  $\pi,(1,0,-1)$ &$(1,11),(2,12)$ &$(3,9),(4,10)$ & $(5,7),(6,8)$\\
 $\pi,(1,1,0)$ &$(5,10),(6,9)$ &$(8,11),(7,12)$ & $(1,3),(2,4)$\\
 $\pi,(1,-1,0)$ &$(5,9),(6,10)$ &$(7,11),(8,12)$ & $(1,3),(2,4)$\\
 $\pi,(0,1,1)$ &$(3,8),(4,7)$ &$(1,6),(2,5)$ & $(9,11),(10,12)$\\
 $\pi,(0,-1,1)$ &$(3,7),(4,8)$ &$(1,5),(2,6)$ & $(9,11),(10,12)$\\
 $\frac\pi2,(0,1,0)$ &$(1,12),(2,11)$ &$(3,10),(4,9)$ & \cellcolor[gray]{0.9}$(5,8),(6,7)$\\
 $-\frac\pi2,(0,1,0)$ &$(1,11),(2,12)$ &$(3,9),(4,10)$ & \cellcolor[gray]{0.9}$(5,8),(6,7)$\\
 $\frac\pi2,(0,0,1)$ &$(5,9),(6,10)$ &$(7,11),(8,12)$ & \cellcolor[gray]{0.9}$(1,4),(2,3)$\\
 $-\frac\pi2,(0,0,1)$ &$(5,10),(6,9)$ &$(8,11),(7,12)$ & \cellcolor[gray]{0.9}$(1,4),(2,3)$\\
 $\frac\pi2,(1,0,0)$ &$(3,7),(4,8)$ &$(1,5),(2,6)$ & \cellcolor[gray]{0.9}$(10,11),(9,12)$\\
 $-\frac\pi2,(1,0,0)$ &$(3,8),(4,7)$ &$(1,6),(2,5)$ & \cellcolor[gray]{0.9}$(10,11),(9,12)$\\
\hline
\end{tabular}
\end{center}
\caption {Table containing all possible twinning systems for cubic to monocinic transformations. {The twins in the shaded region are generic non-conventional twins following the terminology of \cite{PitteriZanzotto98}.}
} \label{Table 1} 
\end{table}
Following \cite[Table 1]{JamesHyst}, in Table \ref{Table 1} we listed all the possible twinning systems for cubic to monoclinic transformations, denoting by $(i,j)$ the twins generated by $\mt U_i,\mt U_j$. The angle and axis of the rotation $\mt R\in SO(3)$ such that $\mt U_j = \mt R\mt U_i\mt R^T$ are given in the first column of the table. 
{\begin{remark}
\rm 
\label{RemGen}
In Table \ref{Table 1} we are not considering non-conventional non-generic twins (see Remark \ref{RemSuper} and \cite{PitteriZanzotto98}) which can arise in cubic to monoclinic transformations. This is because an easy computation allows to prove that non-conventional non-generic twins arising in cubic to monoclinic transformations cannot satisfy the cofactor conditions.
\end{remark}}

Thanks to Proposition \ref{tutti quanti}, if a pair of variants in column (A) (or column (B)) of Table \ref{Table 1} generates a twin which satisfies the cofactor conditions, then all the pairs in the same column, that is column (A) (resp. column (B)) satisfy the cofactor conditions for the same type of twin solution. The situation of column (C) is more complex: if a pair of variants generates a twin which satisfies the cofactor conditions, then all possible compound-twins solutions generated by pairs of variants in column (C) satisfy (CC1)--(CC2). However, in general, not all possible compound-twin solutions satisfy (CC3). This is the case for example for
$$
{\mt U}_1 = \left[\begin{array}{ ccc } 0.95 & 0.0073 & 0 \\ 0.0073 & 1.2 & 0 \\ 0 & 0 & 1 \end{array}\right].
$$
In this case, (CC3) is satisfied by both compound solutions generated by $\mt U_1,\mt U_2$, but by neither of the two compound solutions generated by $\mt U_1,\mt U_3$. Nevertheless, Proposition \ref{tutti quanti} implies that if a pair of variants in the upper box of column (C) (or in the lower box of column (C)) generates a twin satisfying the cofactor conditions, then all other pairs of variants in the same box generate a twin satisfying the cofactor conditions. The difference from the table in \cite{JamesHyst} is that we passed from thirteen boxes of this type to four, that is, we have shown that if a twin satisfies the cofactor conditions, there are many different twins enjoying the same property. Furthermore, also thanks to Remark \ref{Solod1ecof} below, we have clarified the situation of compound twins, which was not completely investigated in \cite{JamesHyst}.

{In case of cubic to orthorhombic transformations, where again we assume that the axis of orthorhombic symmetry corresponds to a $\langle1 0 0\rangle_{cubic}$ direction,} the deformation gradients related to the six martensitic variants are (see \cite[Table 4.2]{Batt})
{\footnotesize
\beq
\label{cubictoortho}
\begin{split}
\tilde{\mt U}_1 = \left[\begin{array}{ ccc } a & b & 0 \\ b & a & 0 \\ 0 & 0 & d \end{array}\right],\quad
\tilde{\mt U}_2 = \left[\begin{array}{ ccc } a & -b & 0 \\ -b & a & 0 \\ 0 & 0 & d \end{array}\right],\quad
\tilde{\mt U}_3 = \left[\begin{array}{ ccc } a & 0 & b \\ 0 & d & 0 \\ b & 0 & a \end{array}\right],\\
\tilde{\mt U}_4 = \left[\begin{array}{ ccc } a & 0 & -b \\ 0 & d & 0 \\ -b & 0 & a \end{array}\right],\quad
\tilde{\mt U}_5 = \left[\begin{array}{ ccc } d & 0 & 0 \\ 0 & a & b \\ 0 & b & a \end{array}\right],\quad
\tilde{\mt U}_6 = \left[\begin{array}{ ccc } d & 0 & 0 \\ 0 & a &- b \\ 0 & -b & a \end{array}\right].
\end{split}
\eeq
}
Furthermore, Table \ref{Table 1} simplifies to Table \ref{Table 2}.
\begin {table}[h]
\begin{center}
\begin{tabular}{ |l|l|l|l|l| }
\hline
$\mt R (\theta,\vc v)$ & type I/II twins  & compound twins \\ \hline
 $\pi,(1,0,0)$ &  &$(1,2),(3,4)$\\
 $\pi,(0,1,0)$ &  &$(1,2),(5,6)$\\
 $\pi,(0,0,1)$ &  &$(3,4),(5,6)$\\
 $\pi,(1,0,1)$ &$(1,6),(2,5)$ & \\
  $\pi,(1,0,-1)$ &$(1,5),(2,6)$ &\\
 $\pi,(1,1,0)$ &$(3,6),(4,5)$ &\\
 $\pi,(1,-1,0)$ &$(3,5),(4,6)$ &\\
 $\pi,(0,1,1)$ & $(1,4),(2,3)$ &\\
 $\pi,(0,-1,1)$ &$(1,3),(2,4)$ &\\
\hline
\end{tabular}
\end{center}
\caption {Table containing all possible twinning systems for cubic to orthorhombic transformations.} \label{Table 2} 
\end{table}
In the following proposition, we prove that in materials undergoing cubic to monoclinic or cubic to orthorhombic transformations, the two martensitic variants cannot satisfy the cofactor conditions both with the type I and with the type II twinning systems.
\begin{proposition}
\label{nonsipuo}
Let $\mt U_i,\mt U_j\in\R^{3\times3}_{Sym^+}$ be as in \eqref{cubictomono} for some $i,j=1,\dots,12$, and $\mt U_i\neq \mt U_j$. 
Let us suppose also that $\hat{\vc e}\in\mathbb{S}^2$ satisfying
$$
\mt U_j = (-\mt 1+2\hat{\vc e}\otimes\hat{\vc e})\mt U_i(-\mt 1+2\hat{\vc e}\otimes\hat{\vc e})
$$
is unique up to a change of sign. Then the type I and type II solutions of the twinning equation \eqref{compatib condit} for $\mt U_i,\mt U_j$ cannot satisfy (CC1)--(CC2) at the same time.
\end{proposition}
\begin{proof}
By Corollary \ref{tutti quanti} (see also Table \ref{Table 1}) we can restrict ourselves to checking just two cases: $\mt U_i=\mt U_1,\mt U_j=\mt U_5$ and $\mt U_i=\mt U_1,\mt U_j=\mt U_{11}$. We focus on the latter, as the former can be deduced similarly. 

We now assume that both the type I and the type II twins generated by $\mt U_{1},\mt U_{11}$ satisfy (CC1)--(CC2), and we aim at a contradiction. Suppose first $d\neq 1$. By \eqref{ccperII} we have that a type II twin satisfies (CC1)--(CC2) if and only if the middle eigenvalue $\lambda_2$ of $\mt U_1$ satisfies $\lambda_2=1$ and $|\mt U_1\vc e|=1$ with $\vc e = 2^{-\frac12}(1,0,1)^T$. That is, we must satisfy at the same time
\begin{align}
\label{l1}
a+c=1+\lambda,\qquad ac-b^2=\lambda,\\ 
\label{tII}
a^2+b^2+d^2=2.
\end{align}
Here $\lambda$ is the eigenvalue of $\mt U_1$ which is neither $1$ nor $d$. 
At the same time, if the type I twin generated by $\mt U_1,\mt U_{11}$ satisfies (CC2), by \eqref{ccperII} $|\mt U_1^{-1}\vc e|=1$ and hence
\beq
\label{tI}
\frac{c^2}{\lambda^2}+\frac{b^2}{\lambda^2}+\frac{1}{d^2}=2.
\eeq
From \eqref{l1}--\eqref{tII} we deduce
$$
(1+\lambda)^2-c(1+\lambda)=2-d^2+\lambda.
$$
On the other hand, \eqref{l1} together with \eqref{tI} imply
\beq
\label{cisl}
c(1+\lambda)=\lambda^2(2-d^{-2})+\lambda.
\eeq
Collecting the last two identities to get rid of $c$ we obtain
$$
(1+\lambda)^2-2+d^2=\lambda^2(2-d^{-2})+2\lambda,
$$
which can be rewritten as
$$
\frac{\lambda^2}{d^2}(1-d^2)=(1-d^2).
$$
Having assumed $d\neq 1,$ and keeping in account that $\mt U_1$ is positive definite, we must have $d=\lambda$. However, as the middle eigenvalue of $\mt U_1$ has to be equal to one (cf. (CC1)), we deduce that $d=\lambda=1$, which implies  $a=1, c=1, b=0$, that is $\mt U_1=\mt U_{11}$. We thus reached a contradiction.\\
Suppose now $d = 1$, then $|\mt U_1^{-1}\vc e|=1,|\mt U_1\vc e|=1 $ reduce to solve
$$
a^2+b^2=1,\qquad c^2+b^2=(ac-b^2)^2.
$$
But these imply either $a=-c$, thus contradicting the fact that $\mt U_1$ is positive definite, or $a=\pm 1,b=0$, which in turn contradict being positive definite or the fact that $\mt U_1\neq \mt U_{11}$. 
\end{proof}
\begin{remark}
\label{dueconuno}
\rm
{
Following the strategy of Proposition \ref{nonsipuo} we can actually prove that in the pure cubic to monoclinic case (that is when $a\neq c$ in \eqref{cubictomono}) it is not possible to satisfy (CC1)--(CC2) with the type I twins (or with the type II twin) of both columns (A) and (B) of Table \ref{Table 1} at the same time.} In the degenerate case where $a = c$ in \eqref{cubictomono}, that is for cubic to orthorhombic transformations, column (A) and column (B) coincide (see Table \ref{Table 2}). Despite this result, it is easy to construct examples of matrices $\mt U_i$ as in \eqref{cubictomono}, such that the cofactor conditions are satisfied by the type I twins in column (A) of Table \ref{Table 1} (or column (B)) and by the type II twins in column (B) (resp. column (A)) of Table \ref{Table 1}. It is also possible to construct matrices $\mt U_i$ as in \eqref{cubictomono}, such that $d=1$, and the cofactor conditions are satisfied by some compound twins in column (C) and by the type I/II twins in column (A)/(B) of Table \ref{Table 1}. If $d=1$ and 
$\det \mt U_i=1$ it is possible to satisfy the cofactor conditions with type I, type II and compound twins at the same time.
\end{remark}

\section{Cofactor conditions for two wells}
\label{qcsec}
In this section we study the quasiconvex hull of the set $$K_{ij} : = SO(3)\mt U_i\cup SO(3)\mt U_j,$$ where $\mt U_i\neq\mt U_j$ are as in \eqref{cubictomono}. We recall that the set of matrices $K^{qc}$, that is the quasiconvex hull of $K_{ij}$, is the set of average deformation gradients which can be achieved by finely and homogeneously mixing the two variants of martensite $\mt U_i,\mt U_j$ (see e.g., \cite{BallJames2,Muller}). If $\mt U_i,\mt U_j$ generate a compound twin we show that a necessary condition to satisfy the cofactor conditions is to have $d=1$. Furthermore, we show that if $\mt U_i$ has a middle eigenvalue which is equal to one, but $d\neq 1$, then the only constant average deformation gradients in $K_{ij}^{qc}$ which are rank one connected to the identity, and are hence compatible with austenite, are the pure phases, that is $\mt R_1\mt U_i,\mt R_2\mt U_i,\mt R_3\mt U_j,\mt R_4\mt U_j$ for some $\mt R_1,\mt R_2,\mt R_3,\mt R_4$. 
On the other hand, if $\mt U_i,\mt U_j$ generate a type I/II solution of the twinning equation \eqref{compatib condit}, say $(\mt R, \vc b,\vc m)$, satisfying the cofactor conditions, then there exist two smooth functions $\mt R_1,\mt R_2 : [0,1]\to SO(3)$ such that the only matrices in $K_{ij}^{qc}$ rank one connected to the identity are given by $\mt R_k(\mu)(\mt U_i+\mu\vc b\otimes\vc m),$ with $\mu\in [0,1]$ and $k=1,2$.
\subsection{Compound twins}
We start by recalling the following result that characterizes the quasiconvex hull of a two well problem in some simple case
\label{Comp twin sec}
\begin{lemma}[{\cite[Thm 2.5.1]{dolzman}}]
\label{lemma Dolz}
Let $\mt A,\mt B \in \R^{3\times3}_{Sym^+}$ with $\det(\mt A)=\det(\mt B)=\mathfrak{D}>0$. Suppose that there exists $\lambda>0$ and $\vc v\in \mathbb{S}^2$ such that $\mt A\vc v=\mt B\vc v = \lambda \vc v$. Then, defined $K$ as
$$
K:=SO(3)\mt A\cup SO(3)\mt B,
$$
we have
$$
K^{qc} = 
\Bigl\{ \mt F\in \R^{3\times3}\colon \det(\mt F)=\mathfrak{D},\, 
\mt F^T\mt F\vc v = \lambda^2\vc v,\, |\mt F\vc e|\leq \max\{|\mt A\vc e|,|\mt B\vc e|\}, 
\text{ for each $\vc e\in \mathbb{S}^2$
 }	\Bigr\}.
$$
\end{lemma}

Also, a consequence of the proof of Lemma 6.2 in \cite{FDP2} is the following result  
\begin{lemma}
\label{lemma mio}
Let $\mt A,\mt B\in\R^{3\times3}_{Sym^+}$ be such that $A\vc e_3 = \mt B\vc e_3 = \lambda\vc e_3$ for some $\lambda\neq 1,\lambda>0$. Assume also $\det \mt A=\det\mt B=\mathfrak{D}>0$. 
Then, necessary conditions for the existence of $\vc a\in\R^3,\vc	 n\in\mathbb{S}^2$ satisfying $\mt 1+\vc a\otimes \vc n \in K^{qc}$ are
$$
|\vc a|=\lambda^{-1}|\mathfrak{D}-\lambda^2|,\qquad n_3^2 = \frac{\lambda^2(1-\lambda^2)}{\mathfrak{D}^2-\lambda^4}
$$
and
$$
\vc a = (\mathfrak{D}-\lambda^2)\vc n-\frac{1}{n_3}(1-\lambda^2)\vc e_3.
$$
\end{lemma}

Now let us consider the cubic to monoclinic transformation, and let $\mt U_i$ with $i=1,\dots,12$ be as in \eqref{cubictomono}. We can prove the following statement
\begin{proposition}
Let $\mt U_1\in\R^{3\times3}_{Sym^+}$ be as in \eqref{cubictomono}, with $a,c,d>0$, $b\geq0$, and $d\neq 1$. Let also $0<\lambda_1\leq\lambda_2=1\leq\lambda_3$ be the eigenvalues of $\mt U_1$. Then, for every $i,j\in\{1,\dots,4\}$ (or, equivalently, $i,j\in\{5,\dots,8\}$, or $i,j\in\{9,\dots,12\}$) with $\mt U_i\neq \mt U_j$, there exist exactly four matrices $\vc a^l\otimes\vc n^l$, $l=1,\dots,4$ satisfying
\beq
\label{Prop eq}
\mt 1 + \vc a^l\otimes\vc n^l \in \Bigl(SO(3)\mt U_i\cup SO(3)\mt U_j\Bigr)^{qc}.
\eeq
Furthermore, there exist four rotation matrices $\mt R^l$, $l=1,\dots,4$ such that 
\begin{align}
\label{R1 1}
\mt R^l\mt U_i = \mt 1 + \vc a^l\otimes\vc n^l, \qquad l=1,2\\
\label{R1 2}
\mt R^l\mt U_j = \mt 1 + \vc a^l\otimes\vc n^l, \qquad l=3,4.
\end{align}
\end{proposition}

\begin{proof}
Let us deal with the case $i=1,j=4$: the other cases can be treated in a similar way.\\

On the one hand, given the assumptions on $\mt U_1, \mt U_4$, \cite[Prop. 4]{BallJames1} assures the existence of four matrices $\vc a^l\otimes\vc n^l\in\R^{3\times3}$, and four rotations $\mt R^l\in SO(3)$, $l=1,\dots,4$, satisfying \eqref{R1 1}--\eqref{R1 2}.\\

On the other hand, from Lemma \ref{lemma mio} we know that the first two components of $\vc n$, namely $n_1,n_2$, must lie on a circle
\beq
\label{circle2345}
n_1^2+n_2^2 = 1-n_3^2 = \frac{\mathfrak D^2-d^2}{\mathfrak D^2-d^4},
\eeq
where $\mathfrak{D}=\det \mt U_1$, and that 
\beq
\label{LaF}
\mt F : = \mt 1+\vc a\otimes\vc n = \mt 1 + (\mathfrak{D}-d)\vc n\otimes\vc n - \frac{1}{n_3}(1-d^2)\vc e_3\otimes \vc n. 
\eeq
Let us now look for unit vectors $\vc v = (v_1,v_2,0)$ satisfying $|\mt U_1\vc v| = |\mt U_4\vc v|$, that is
$$
(av_1+bv_2)^2+(bv_1+cv_2)^2=(cv_1-bv_2)^2+(-bv_1+av_2)^2.
$$
There are up to a change of sign two such unit vectors $\vc v^{+},\vc v^{-}$ satisfying the above identity. Respectively they are such that
$$
{v_1^{+}}({a-c})={v_2^{+}}{(2-a-c-2b)},\qquad
{v_1^{-}}{(a-c)}={v_2^{-}}{(a+c-2b-2)}, \qquad {v_1^{+}}{v_1^{-}}=-{v_2^{+}}{v_2^{-}}.
$$
{This can be proved by using the fact that, as $\lambda_2=1$, $ac-b^2=\lambda$ and $a+c=1+\lambda$, with $\lambda=\lambda_1$ if $\lambda_3=d$ and $\lambda=\lambda_3$ if $\lambda_1=d$. We remark also that, ${v_1^{+}}{v_1^{-}}=-{v_2^{+}}{v_2^{-}}$ together with $|\vc v^+|=|\vc v^-|=1$ yield $v_1^{-}=-v_2^{+}$ and $v_2^{-}=v_1^{+}$.} 
Now, Lemma \ref{lemma Dolz} implies that a necessary condition for $\mt F$ to be in $K_{ij}^{qc}$ is 
$$
|\mt F\vc v^{+}|\leq |\mt U_1\vc v^{+}|,\qquad |\mt F\vc v^{-}|\leq |\mt U_1\vc v^{-}|.
$$
By \eqref{LaF}, after a few computations these two inequalities become
\begin{align}
\label{ineq1qcC}
v_1^+(v_1^+\alpha+v_2^+\beta)+v_2^+(v_1^+\beta+v_2^+\gamma) \leq (av_1^++bv_2^+)^2+(bv_1^++cv_2^+)^2,\\
\label{ineq2qcC}
v_1^-(v_1^-\alpha+v_2^-\beta)+v_2^-(v_1^-\beta+v_2^-\gamma) \leq (av_1^-+bv_2^-)^2+(bv_1^-+cv_2^-)^2,
\end{align}
where 
$$
\alpha = 1 +  \frac{n_1^2}{d^2}(\mathfrak{D}^2-d^4),\qquad 
\gamma = 1 +  \frac{n_2^2}{d^2}(\mathfrak{D}^2-d^4),\qquad 
\beta =  \frac{n_1n_2}{d^2}(\mathfrak{D}^2-d^4).
$$
Summing up the last two inequalities and exploiting the fact that $v_1^{-}=-v_2^{+}$, $v_2^{-}=v_1^{+}$, $(v_1^+)^2+(v_2^+)^2=1$ we get
$$
\alpha + \gamma \leq a^2+c^2+2b^2 = (a+c)^2-2\lambda =  1+\lambda^2
.$$
Here, as above, we used the fact that $ac-b^2=\lambda$ and $a+c=1+\lambda$. Now, we notice that \eqref{circle2345} implies
$$
\alpha+\gamma = 2 + (n_1^2+n_2^2)\frac{1}{d^2}(\mathfrak{D}^2-d^4)= 2 +\frac{1}{d^2}(\mathfrak{D}^2-d^2) = 1+\lambda^2. 
$$
Thus, \eqref{ineq1qcC}--\eqref{ineq2qcC} are actually equalities, that is
\begin{align*}
v_1^+(v_1^+\alpha+v_2^+\beta)+v_2^+(v_1^+\beta+v_2^+\gamma) = (av_1^++bv_2^+)^2+(bv_1^++cv_2^+)^2,\\
v_1^-(v_1^-\alpha+v_2^-\beta)+v_2^-(v_1^-\beta+v_2^-\gamma) = (av_1^-+bv_2^-)^2+(bv_1^-+cv_2^-)^2,
\end{align*}
Exploiting the values of $\alpha,\beta,\gamma$ we can rewrite these identities as
\begin{align*}
(v_1^+n_1 +v_2^+n_2)^2 = \frac{d^2}{\mathfrak{D}^2-d^4}\bigl((av_1^++bv_2^+)^2+(bv_1^++cv_2^+)^2-1\bigr):=r_1^2,\\
(v_2^+n_1 -v_1^+ n_2)^2 = \frac{d^2}{\mathfrak{D}^2-d^4}\bigl((av_1^+ - bv_2^+)^2+(bv_1^+ - cv_2^+)^2-1\bigr):=r_2^2.
\end{align*}
Therefore, $n_1,n_2$ must be at the same time on one of the lines $(v_1^+n_1 +v_2^+n_2)=\pm r_1$ and on one of the lines $(v_2^+n_1 -v_1^+n_2)=\pm r_2$. Therefore there exist a maximum of four couples $(n_1^l,n_2^l)$ such that $|\mt F \vc v^+|\leq |\mt U_1 \vc v^+|$ and $|\mt F \vc v^-|\leq |\mt U_1 \vc v^-|$. As $n_3$ can have both a positive and a negative sign, we hence found eight $\vc n$ and, by Lemma \ref{lemma mio} eight $\vc a$, such that \eqref{Prop eq} is satisfied. These can be expressed as
\[
\begin{split}
(\vc a^1,\vc n^1),\quad(\vc a^2,\vc n^2)&,\quad(\vc a^3,\vc n^3),\quad(\vc a^4,\vc n^4),\\(-\vc a^1,-\vc n^1),\quad(-\vc a^2,-\vc n^2)&,\quad(-\vc a^3,-\vc n^3),\quad(-\vc a^4,-\vc n^4).
\end{split}
\]
Therefore, there exist exactly four matrices $\vc a \otimes \vc n$ satisfying \eqref{Prop eq}, which are given by \eqref{R1 1}--\eqref{R1 2}.
\end{proof}

\begin{remark}
\label{Solod1ecof}
\rm
Consequence of the above result is that a compound twin in a cubic to monoclinic transformation (and hence also in its special cases as the cubic to orthorhombic or the cubic to tetragonal) can satisfy the cofactor conditions only if $d=1$. As shown in \cite[Remark 5.3]{FDP2}, if $d=1$ the set of matrices $\mt F\in (SO(3)\mt U_i\cup SO(3)\mt U_j)^{qc}$ such that $\mt F =\mt 1+\vc a\otimes\vc n$ for some $\vc a\in\R^3,\vc n\in \mathbb{S}^2$ can have dimension two if $\mt U_i,\mt U_j$ generate a compound twin. Indeed, by Lemma \ref{lemma Dolz}, all matrices in the quasiconvex hull of $K$ have $1$ as a singular value.
\end{remark}

\subsection{Type I/II twins}
Let now $\mt U,\mt V\in\R^{3\times3}_{Sym^+}$ satisfy \eqref{e hat}, and $(\mt U,\vc b,\vc m)$, $(\mt U,\hat{\vc b},\hat{\vc m})$ be the type I/type II solutions of the twinning equation \eqref{compatib condit}. Suppose further that either $(\mt U,\vc b,\vc m)$ or $(\mt U,\hat{\vc b},\hat{\vc m})$ satisfies the cofactor conditions. We are interested in which constant average deformation gradients obtained by finely mixing the two martensitic variants $\mt U$ and $\mt V$ are compatible with austenite. We can prove the following result.
\begin{proposition}
\label{1dilqc}
Suppose $(\mt U,\vc b,\vc m),(\mt U,\hat{\vc b},\hat{\vc m})$ are the type I/type II (non necessarily in order) solutions of the twinning equation \eqref{compatib condit} for $\mt U,\mt V\in\R^{3\times3}_{Sym^+}$, with $\mt U,\mt V$ satisfying \eqref{e hat}. Suppose further that $(\mt U,\vc b, \vc m)$ satisfies the cofactor conditions, while $(\mt U,\hat{\vc b},\hat{\vc m})$ does not satisfy (CC2). Then, $\mt F\in (SO(3)\mt U\cup SO(3)\mt V)^{qc}$ is such that $\mt F = \mt 1+\vc a\otimes \vc n$ for some $\vc a\in \R^3,\vc n\in \mathbb{S}^2$ if and only if $$\mt F^T\mt F = (\mt U+\mu\vc b\otimes\vc m)^T(\mt U+\mu\vc b\otimes\vc m),$$ for some $\mu\in[0,1]$. 
\end{proposition}
\begin{proof}
Let us first define
\begin{align*}
\vc u_1 := \mt U^{-1}\vc m |\mt U^{-1}\vc m|^{-1},\quad \vc u_3 := \vc b|\vc b|^{-1},\quad \vc u_2 = \vc u_1\times\vc u_3,\\
\delta=\frac12|\vc b||\vc U^{-1}\vc m|,\qquad \mt L:= \mt U^{-1}(\mt 1-\delta\vc u_3\otimes\vc u_1 ).
\end{align*} 
Thanks to \cite[Section 5]{BallJames2}, we know that, defined
$$
K_{\mt U\mt V}:= SO(3)\mt U\cup SO(3)\mt V,
$$ 
its quasiconvex hull is given by
$$
K_{\mt U\mt V}^{qc} = \bigl\{ \mt F \in \R^{3\times 3}\colon 
\mt F^T\mt F = \mt L^{-T} \mt M (\alpha,\beta,\gamma)\mt L^{-1},\, 0<\alpha\leq 1+\delta^2,\, 0<\gamma\leq 1,\, \alpha\gamma-\beta^2 = 1\bigr\},
$$ 
where 
$$
\mt M (\alpha,\beta,\gamma):= \alpha\vc u_1\otimes\vc u_1 + \vc u_2\otimes\vc u_2 + \gamma \vc u_3\otimes \vc u_3 + \beta \vc u_1\odot\vc u_3,
$$
and we denoted $\vc u_1\odot\vc u_3=\vc u_1\otimes\vc u_3+\vc u_3\otimes\vc u_1$. Let us first consider the scalar function
$$
f_1(\alpha,\beta,\gamma) = \det (\mt L^{-T} \mt M (\alpha,\beta,\gamma)\mt L^{-1}-\mt 1).
$$
We define $\phi(\alpha,\beta,\gamma) = \lambda_M(\alpha,\beta,\gamma)\lambda_m(\alpha,\beta,\gamma),$ where $\lambda_m,\lambda_M$ are respectively the largest and the smallest eigenvalue of $\mt L^{-T} \mt M (\alpha,\beta,\gamma)\mt L^{-1}-\mt 1.$ We denote by $\mathcal{A}$ the region
$$
\mathcal{A}:=\bigl\{(\alpha,\beta,\gamma)\in\R^3	\colon 0<\alpha\leq 1+\delta^2,\quad 0<\gamma\leq 1,\quad \alpha\gamma-\beta^2 = 1\bigr\}.
$$
We are interested in the set
$$
\mathcal{B}:=\bigl\{(\alpha,\beta,\gamma)\in \mathcal{A}\colon f_1(\alpha,\beta,\gamma)=0, \quad \phi(\alpha,\beta,\gamma)\leq 0\bigr\},
$$
which characterises the set of $\mt F\in (SO(3)\mt U\cup SO(3)\mt V)^{qc}$ which are rank-one connected to the identity (see e.g., \cite[Prop. 4]{BallJames1}). We first notice that
$$
f_1 = \frac1{\det\mt U^2} \det \bigl(\mt M - \mt L^T\mt L	\bigr),
$$
and that
$$
g(\beta,\gamma) := \det \mt U^2 f_1\Bigl( \frac{1+\beta^2}{\gamma},\beta,\gamma\Bigr) ,
$$
must be of the form
$$
g(\beta,\gamma) = \frac{c_1\beta^2+c_2}{\gamma} + c_3\beta + c_4\gamma +c_5 ,
$$
for some $c_1,c_2,c_3,c_4,c_5\in \R.$ We also point out that, by \eqref{10}--\eqref{11} together with \eqref{ccperI}--\eqref{ccperII}, we have
\beq
\label{e1LL}
\mt L^T\mt L\vc u_1 = \vc u_1,\quad\text{if $(\mt U,\vc b,\vc m)$ is a type I twin},\qquad\mt L^T\mt L\vc u_3 = \vc u_3,\quad\text{if $(\mt U,\vc b,\vc m)$ is a type II twin}.
\eeq
Let us define $\mt Q := 2\vc u_j\otimes\vc u_j-\mt 1$, with $j=1$ if $(\mt U,\vc b,\vc m)$ is a type I twin, and $j=3$ if $(\mt U,\vc b,\vc m)$ is a type II twin. By \eqref{e1LL} we have that
\[
\begin{split}
g(\beta,\gamma)= \det\bigl((\mt M( \gamma^{-1}(1+\beta^2),\beta,\gamma) - \mt L^T\mt L	)\bigr) = \det \bigl(\mt Q(\mt M( \gamma^{-1}(1+\beta^2),\beta,\gamma) - \mt L^T\mt L	)\mt Q\bigr) = 
g(-\beta,\gamma).
\end{split}
\]
For this reason, $g(\beta,\gamma)$ is even in $\beta$, and is of the form
\beq
\label{g2}
g(\beta,\gamma) = \frac{c_1\beta^2+c_2}{\gamma} + c_4\gamma +c_5 .
\eeq
Now, we notice that
\beq
\label{laminato mu}
(\mt U + \mu \vc b\otimes\vc m)\mt L = (\mt 1 + 2\mu \delta \vc u_3\otimes\vc u_1)\mt U\mt L =  (\mt 1 + \delta(2\mu-1)  \vc u_3\otimes\vc u_1).
\eeq
Therefore, the fact that $(\mt U,\vc b,\vc m)$ satisfies the cofactor conditions implies
$$
\frac1{\det \mt U^2}g(\beta,1) = \det \bigl((\mt U + \mu(\beta)\vc b\otimes\vc m)^T(\mt U + \mu(\beta)\vc b\otimes\vc m)-\mt 1\bigr) = 0,\qquad\text{for all $\beta\in[-\delta,\delta],$} 
$$
where $\mu(\beta)= \frac{\beta + \delta}{2\delta}.$ Thus, \eqref{g2} simplifies to
\beq
\label{g3}
g(\beta,\gamma) = c_2\Bigl(\frac{1}{\gamma}-1\Bigr) + c_4(\gamma -1),
\eeq
that means, $g$ is constant in $\beta.$ Now consider the solution of the twinning equation \eqref{compatib condit} other than $(\mt U,\vc b,\vc m)$, namely $(\mt U,\hat{\vc b},\hat{\vc m})$. By \cite[Prop. 5]{BallJames1} we have
\beq
\label{ildet}
\det\bigl( (\mt U + \mu \hat{\vc b}\otimes\hat{\vc m})^T(\mt U + \mu \hat{\vc b}\otimes\hat{\vc m}) -\mt 1\bigr) = c_0 \mu (1-\mu),
\eeq
for some $c_0\in\R$. We notice that
\beq
\label{altrotwin}
\mt L^T(\mt U + \mu \hat{\vc b}\otimes\hat{\vc m})^T(\mt U + \mu \hat{\vc b}\otimes\hat{\vc m})\mt L = \mt M\Bigl(1+\delta^2,\delta(2\mu-1), \frac{1 + \delta^2(2\mu-1)^2}{1+\delta^2}\Bigr),\qquad\mu\in[0,1].
\eeq
This can be shown by using \eqref{10}--\eqref{11} and the fact that, by \eqref{ccperI}--\eqref{ccperII}, $|\mt U^{-1}\hat{\vc e}|=1$ and $|\mt U\hat{\vc e}|=1$, respectively for type I and type II twins satisfying the cofactor conditions, with $\hat{\vc e}$ being as in \eqref{e hat}. Thus, setting 
$$
\hat\beta(\mu) := \delta(2\mu-1),\qquad \hat\gamma(\mu) := \frac{1 + \delta^2(2\mu-1)^2}{1+\delta^2},$$ 
\eqref{ildet} becomes
$$
\det\bigl( (\mt U + \mu \hat{\vc b}\otimes\hat{\vc m})^T(\mt U + \mu \hat{\vc b}\otimes\hat{\vc m}) -\mt 1\bigr) = \frac{c_0(1+\delta^2)}{4\delta^2}(1-\hat\gamma(\mu)).
$$
Therefore, by \eqref{altrotwin},
$$
g(\hat\beta(\mu),\hat\gamma(\mu)) = \det \mt U^2 \frac{c_0(1+\delta^2)}{4\delta^2}(1-\hat\gamma(\mu)).   
$$
But recalling that $g(\beta,\gamma)$ is independent of $\beta$, a comparison with \eqref{g3} yields
\beq
\label{c0neq0}
g(\beta,\gamma) = \det \mt U^2 \frac{c_0(1+\delta^2)}{4\delta^2}(1-\gamma).
\eeq
Here, $c_0\neq 0$ as we assumed that $(\mt U,\hat{\vc b},\hat{\vc m})$ does not satisfy (CC2). Therefore, by \eqref{laminato mu} and Theorem \ref{thm cof cond}, we get that $(\alpha,\beta,\gamma)\in\mathcal{B}$ if and only if $\gamma=1$, $\beta \in[-1,1]$ and $\alpha = 1+\beta^2$, that is (see \eqref{laminato mu}), if and only if
$$
\mt L^{-T}\mt M\mt L^{-1} = (\mt U + \mu {\vc b}\otimes{\vc m})^T(\mt U + \mu {\vc b}\otimes{\vc m}),\qquad\mu\in[0,1],
$$
which is the claimed result. 
\end{proof}
\begin{remark}
\rm
From \eqref{c0neq0} we notice that if both $(\mt U,{\vc b},{\vc m})$ and $(\mt U,\hat{\vc b},\hat{\vc m})$ satisfy the cofactor conditions $f_1=0$ everywhere at the interior of $\mathcal{A}$. Therefore, in this case, the set $\mathcal{B}$ is two-dimensional and not one-dimensional as in the case of Proposition \ref{1dilqc}. Indeed, $\phi$ is a smooth function of $\alpha,\beta,\gamma$, and is strictly negative on the boundary of $\mathcal{A},$ except for at most two points. Thus, by continuity, there exists an open two-dimensional region contained in $\mathcal{A}$ (and hence in $(SO(3)\mt U\cup SO(3)\mt V)^{qc}$) where $\phi$ is negative.
\end{remark}
\begin{remark}
\rm
Under the hypotheses of Proposition \ref{1dilqc}, the set of matrices in $K_{\mt U,\mt V}^{qc}$ which are rank-one connected to the identity coincides with two smooth and one-dimensional curves of finite length. The dimension of this set is hence the same as it is in the case of twins not satisfying the cofactor conditions. Indeed, for example, in \cite{BallCarstensen1} (see also Lemma \ref{lemma mio} above) it is shown that for 
$$ \mt U = \diag(a, b ,c),\qquad\mt V = \diag(b, a ,c),$$
for some $a,b,c>0$, $a,b,c\neq 1$, the set of matrices in $K_{\mt U,\mt V}^{qc}$ which are rank-one connected to the identity coincides with four smooth and one-dimensional curves of finite length. Nonetheless, the difference is in the fact that, when the cofactor conditions are satisfied, the microstructures which can form an interface with austenite are just simple laminates, and not laminates within laminates. 
\end{remark}

The following corollary is a straightforward consequence of \cite[Prop. 4]{BallJames1} and Proposition \ref{1dilqc}:
\begin{corollary}
\label{1dilqcij}
Suppose $(\mt U_i,\vc b,\vc m),(\mt U_i,\hat{\vc b},\hat{\vc m})$ are the type I/type II solutions of the twinning equation \eqref{compatib condit} for $\mt U_i,\mt U_j\in\R^{3\times3}_{Sym^+}$ (not necessarily in order), with $\mt U_i\neq\mt U_j$ as in \eqref{cubictomono}. 
Suppose further that $(\mt U_i,\vc b, \vc m)$ satisfies the cofactor conditions. 
Then, there exist $\mt R_1,\mt R_2\colon [0,1]\to SO(3)$ such that $\mt F\in K_{ij}^{qc}$ is of the form $\mt F = \mt 1+\vc a\otimes \vc n$ for some $\vc a\in \R^3,\vc n\in \mathbb{S}^2$ if and only if $$\mt F = \mt R_i(\mu)(\mt U+\mu\vc b\otimes\vc m),$$ for some $\mu\in[0,1], i=1,2$. 
\end{corollary}

\section{Type I and II star twins}
\label{StarSec}
This section is devoted to introducing the definition of star twins, and to show some basic consequences of these special conditions of supercompatibility. 
Below, $\vc a, \vc a_{\mt U}, \vc a_{\mt V}, \vc n_{\mt U}, \vc n_{\mt V}$ are as in Theorem \ref{TypeI} and Theorem \ref{TypeII}. 
\begin{definition}
\label{Def star I}
Let $\mathcal M$ be a subset of $\R^{3\times3}_{Sym^+}.$ Let $\mt U,\mt V\in\mathcal M$, $\mt U\neq\mt V$, satisfying \eqref{e hat} and let $\mt R_{I}\in SO(3), \vc b_{I}\in\R^3,\vc m_{I}\in\mathbb{S}^2$ be a type I solution of the twinning equation \eqref{compatib condit} for $\mt U,\mt V$. Suppose further that $(\mt U,\vc b_{I}, \vc m_{I})$ satisfies the cofactor conditions. Then we say that $(\mt U,\vc b_{I}, \vc m_{I})$ is a type I star twin generated by $(\mt U,\mt V)$ if there exist three different rotations $\mt Q_1,\mt Q_2,\mt Q_3\in SO(3)$, $\mt Q_i \neq \mt1$ for each $i=1,\dots,3$, and $\mu^*\in(0,1)$ satisfying
\begin{enumerate}[label={(S\arabic*)}]
\item\label{S1s} $\mt Q_i\mt U\mt Q_i^T, \mt Q_i\mt V\mt Q_i^T\in\mathcal M$, for each $i=1,\dots,3$,
\item\label{S2s} $\mt Q_i\bigl(\mu^*\vc n_{\mt U}+(1-\mu^*)\vc n_{\mt V}\bigr) = \chi_i\bigl(\mu^*\vc n_{\mt U}+(1-\mu^*)\vc n_{\mt V}\bigr)$ for each $i=1,\dots,3$ and with $\chi_i=\pm 1$,
\item \label{S3s}$\bigl((\mt Q_i \vc a )\times (\mt Q_j \vc a)\bigr) \cdot\vc a \neq 0 $ for each $i,j=1,\dots,3$, $i\neq j$ and $\bigl((\mt Q_1 \vc a )\times (\mt Q_2 \vc a)\bigr) \cdot(\mt Q_3 \vc a ) \neq 0 $.
\end{enumerate}
If the number of $\mt Q_i$ satisfying \ref{S1s}--\ref{S3s} is two and not three we say that $(\mt U,\vc b_{I}, \vc m_{I})$ is a type I half-star twin.  
\end{definition}

\begin{definition}
\label{Def star}
Let $\mathcal M$ be a subset of $\R^{3\times3}_{Sym^+}.$ Let $\mt U,\mt V\in\mathcal M$, $\mt U\neq\mt V$, and let $\mt R_{II}\in SO(3), \vc b_{II}\in\R^3,\vc m_{II}\in\mathbb{S}^2$ be a type II solution of the twinning equation \eqref{compatib condit} for $\mt U,\mt V$. Suppose further that $(\mt U,\vc b_{II}, \vc m_{II})$ satisfies the cofactor conditions. Then we say that $(\mt U,\vc b_{II}, \vc m_{II})$ is a type II star twin generated by $(\mt U,\mt V)$ if there exist three different rotations $\mt Q_1,\mt Q_2,\mt Q_3\in SO(3)$, $\mt Q_i \neq \mt1$ for each $i=1,\dots,3$, and $\mu^*\in(0,1)$ satisfying
\begin{enumerate}[label={(T\arabic*)}]
\item\label{T1s} $\mt Q_i\mt U\mt Q_i^T, \mt Q_i\mt V\mt Q_i^T\in\mathcal M$, for each $i=1,\dots,3$,
\item\label{T2s} $\mt Q_i\bigl(\mu^*\vc a_{\mt U}+(1-\mu^*)\vc a_{\mt V}\bigr) = \chi_i\bigl(\mu^*\vc a_{\mt U}+(1-\mu^*)\vc a_{\mt V}\bigr)$ for each $i=1,\dots,3$ and with $\chi_i=\pm 1$,
\item \label{T3s}$\bigl((\mt Q_i \vc m_{II} )\times (\mt Q_j \vc m_{II})\bigr) \cdot\vc m_{II} \neq 0 $ for each $i,j=1,\dots,3$, $i\neq j$ and $\bigl((\mt Q_1 \vc m_{II} )\times (\mt Q_2 \vc m_{II})\bigr) \cdot(\mt Q_3 \vc m_{II} ) \neq 0 $.
\end{enumerate}
If the number of $\mt Q_i$ satisfying \ref{T1s}--\ref{T3s} is two and not three we say that $(\mt U,\vc b_{II}, \vc m_{II})$ is a type II half-star twin.  
\end{definition}
\begin{remark}
\label{remark sui mus}
\rm
It would be possible to generalise \ref{T2s} (or similarly \ref{S2s}) in the definition of star twins by requiring the existence of $\mu_1^*,\mu_2^*\in[0,1]$ such that
$$
\mt Q_i\bigl(\mu^*_1\vc a_{\mt U}+(1-\mu^*_1)\vc a_{\mt V}\bigr) = \chi_i\bigl(\mu^*_2\vc a_{\mt U}+(1-\mu^*_2)\vc a_{\mt V}\bigr) \text{ for each $i=1,\dots,3$ and with $\chi_i=\pm 1$.}
$$
However, a fine observation of the proof of Theorem \ref{main thm II} (resp. Theorem \ref{main thm I}) yields that, in the cubic to monoclinic case, the only interesting case is when $\mu_1^*=\mu_2^*$, making the generalisation not relevant to our context.
\end{remark}
\begin{remark}
\rm
Definition \ref{Def star I} and Definition \ref{Def star} both imply the existence of four exactly compatible laminates $\mt F_0,\mt F_1,\mt F_2,\mt F_3$ satisfying \eqref{startwinsDefIntro}. Conversely, in cubic to monoclinic transformations, given four exactly compatible type I (or of type II) laminates $\mt F_0,\mt F_1,\mt F_2,\mt F_3$ satisfying \eqref{startwinsDefIntro}, we have that Proposition \ref{tutti quanti}, Remark \ref{dueconuno} and Remark \ref{remark sui mus} imply the satisfaction of Definition \ref{Def star I} (resp. Definition \ref{Def star}).
\end{remark} 
The following two propositions state that, in the presence of star-twins, the set of average deformation gradients which can form stress free phase interfaces with austenite is unusually large:
\begin{proposition}
\label{LMhull}
Let $\mathcal M$ be a finite subset of $\R^{3\times3}_{Sym^+},$ $\mt U,\mt V\in\mathcal M$ satisfying \eqref{e hat}, and let $(\mt U,\vc b_{II}, \vc m_{II})$ be a type II half-star twin generated by $(\mt U,\mt V)$. Then,
$$
\mt 1+\bigl(\mu^*\vc a_{\mt U}+(1-\mu^*)\vc a_{\mt V}\bigr)\otimes \vc m_{II},\qquad\mt 1+\bigl(\mu^*\vc a_{\mt U}+(1-\mu^*)\vc a_{\mt V}\bigr)\otimes \mt Q_1\vc m_{II},\qquad\mt 1+\bigl(\mu^*\vc a_{\mt U}+(1-\mu^*)\vc a_{\mt V}\bigr)\otimes \mt Q_2\vc m_{II},
$$
and all their convex combinations are contained in the set $K^{qc}$, where
$$
K:=\bigcup_{\mt M\in\mathcal{M}}SO(3)\mt M .
$$
If $(\mt U,\vc b_{II}, \vc m_{II})$ is a type II star twin generated by $(\mt U,\mt V)$, then also
$$
\mt 1+\bigl(\mu^*\vc a_{\mt U}+(1-\mu^*)\vc a_{\mt V}\bigr)\otimes \mt Q_3\vc m_{II}\in K^{qc}.
$$
\end{proposition}
\begin{proof}
We restrict to the case of half-star twins; the case of star twins follows similarly. From \cite{Muller} we know that given a set $C\in\R^{3\times3}$, the lamination convex hull of $C$, namely $C^{lc}$, is contained in $C^{qc}$. We recall for the benefit of the reader that
$$
C^{lc}:=\bigcup_{i=1}^\infty C^{(i)}, 
$$
where $C^{(1)}=C$ and
$$
C^{(i+1)}:=C^{(1)}\cup\bigl\{\mu\mt A+(1-\mu)\mt B\colon \mt A,\mt B\in C^{(i)},\quad\rank(\mt A-\mt B)=1,\quad\mu\in(0,1)	\bigr\}.
$$
Therefore, let $K:=\bigcup_{M\in\mathcal{M}}SO(3)M$. By Theorem \ref{TypeII} we know that
\[
\begin{split}
\mt R_{\mt U}\mt U=\mt 1+ \vc a_{\mt U}\otimes \vc m_{II},&\quad \mt R_{\mt V}\mt V=\mt 1+ \vc a_{\mt V}\otimes \vc m_{II},\\
\mt Q_1 \mt R_{\mt U}\mt Q_1^T \mt Q_1\mt U\mt Q_1^T=\mt 1+ \mt Q_1\vc a_{\mt U}\otimes \mt Q_1\vc m_{II},&\quad 
\mt Q_1\mt R_{\mt V}\mt Q_1^T\mt Q_1\mt V\mt Q_1^T=\mt 1+ \mt Q_1\vc a_{\mt V}\otimes \mt Q_1\vc m_{II},\\
\mt Q_2 \mt R_{\mt U}\mt Q_2^T \mt Q_2\mt U\mt Q_2^T=\mt 1+ \mt Q_2\vc a_{\mt U}\otimes \mt Q_2\vc m_{II},&\quad \mt Q_2\mt R_{\mt V}\mt Q_2^T\mt Q_2\mt V\mt Q_2^T=\mt 1+ \mt Q_2\vc a_{\mt V}\otimes \mt Q_2\vc m_{II},
\end{split}
\]
are all in $K$. Thus,
\[
\begin{split}
\mu^*\mt R_{\mt U}\mt U+(1-\mu^*)\mt R_{\mt V}\mt V &=\mt 1+ \bigl(\mu^*\vc a_{\mt U}+(1-\mu^*)\vc a_{\mt V}\bigr)\otimes \vc m_{II}, \\
 \mu^*\mt Q_1\mt R_{\mt U}\mt Q_1^T\mt Q_1\mt U\mt Q_1^T+(1-\mu^*)\mt Q_1\mt R_{\mt V}\mt Q_1^T\mt Q_1\mt V\mt Q_1^T &=\mt 1+ \bigl(\mu^*\vc a_{\mt U}+(1-\mu^*)\vc a_{\mt V}\bigr)\otimes \mt Q_1\vc m_{II},\\
 \mu^*\mt Q_2\mt R_{\mt U}\mt Q_2^T\mt Q_2\mt U\mt Q_2^T+(1-\mu^*)\mt Q_2\mt R_{\mt V}\mt Q_2^T\mt Q_2\mt V\mt Q_2^T &=\mt 1+ \bigl(\mu^*\vc a_{\mt U}+(1-\mu^*)\vc a_{\mt V}\bigr)\otimes \mt Q_2\vc m_{II},\end{split}
\]
are in $K^{(2)}$, and all their convex combinations are in $K^{(3)}$.
\end{proof}
\begin{remark}
\rm
\label{3dInterf}
Type II half-star/star twins are interesting also because, by combining laminates, we can easily construct macroscopic curved interfaces between austenite and martensite whose normal does not lie in a plane. Indeed, in the notation of Proposition \ref{LMhull}, let us define $${\vc a}^*:=\mu^*\vc a_{\mt U}+(1-\mu^*)\vc a_{\mt V}.$$ 
Then,
$$
\mt G(\mu_1,\mu_2):=\mt 1 + \vc a^*\otimes \bigl(\mu_1(\mu_2\vc m_{II}+(1-\mu_2)\mt Q_1\vc m_{II})+(1-\mu_1)\mt Q_2\vc m_{II} \bigr)
$$
is in $K^{qc}$ for every $\mu_1,\mu_2\in[0,1],$ and $\vc m_{II},\mt Q_1\vc m_{II},\mt Q_2\vc m_{II}$ span $\R^3.$ Therefore, given a smooth bounded domain $\Omega$, one can choose $\mu_1,\mu_2\colon\Omega\to[0,1]$ to be space dependent, and, provided
$$
\mu_1\bigl(\nabla\mu_2\times (\vc m_{II}-\mt Q_1\vc m_{II}) + \nabla\mu_1\times(\mu_2\vc m_{II}+(1-\mu_2)\mt Q_1\vc m_{II})
-\nabla\mu_1\times \mt Q_2\vc m_{II}=0,
$$
$\mt G(\mu_1(\vc x),\mu_2(\vc x))$ is an average deformation gradient. Choosing for example 
$$
\mu_2(\vc x) = f(\vc x\cdot(\vc m_{II}-\mt Q_1\vc m_{II})),\qquad
 \mu_1(\vc x) = c_1 \biggl(\int_0^{\vc x\cdot(\vc m_{II}-\mt Q_1\vc m_{II})}f(s)\,\mathrm{d}s + \vc x\cdot (\mt Q_1\vc m_{II}-\mt Q_2\vc m_{II})\biggr) + c_2,
$$
for some smooth $f\colon \R\to[0,1]$ and some $c_1,c_2\in\R$ such that $\mu_1(\vc x)\in[0,1]$ for every $\vc x\in \Omega$, will give an average deformation gradient of martensite, which is compatible with austenite across an interface whose normal does not lie in a plane.
\end{remark}
In the same way we can prove
\begin{proposition}
\label{LMhullI}
Let $\mathcal M$ be a finite subset of $\R^{3\times3}_{Sym^+},$ $\mt U,\mt V\in\mathcal M$ and let $(\mt U,\vc b_{I}, \vc m_{I})$ be a type I half-star twin generated by $(\mt U,\mt V)$. Then
$$
\mt 1+\vc a\otimes \bigl(\mu^*\vc n_{\mt U}+(1-\mu^*)\vc n_{\mt V}\bigr),\qquad\mt 1+\mt Q_1\vc a\otimes \bigl(\mu^*\vc n_{\mt U}+(1-\mu^*)\vc n_{\mt V}\bigr),\qquad\mt 1+\mt Q_2\vc a\otimes \bigl(\mu^*\vc n_{\mt U}+(1-\mu^*)\vc n_{\mt V}\bigr),
$$
and all their convex combinations are contained in the set $K^{qc}$, where
$$
K:=\bigcup_{M\in\mathcal{M}}SO(3)M.
$$
If $(\mt U,\vc b_{I}, \vc m_{I})$ is a type I star twin generated by $(\mt U,\mt V)$, then also
$$
\mt 1+\mt Q_3\vc a\otimes \bigl(\mu^*\vc n_{\mt U}+(1-\mu^*)\vc n_{\mt V}\bigr)\in K^{qc}.
$$
\end{proposition}

{
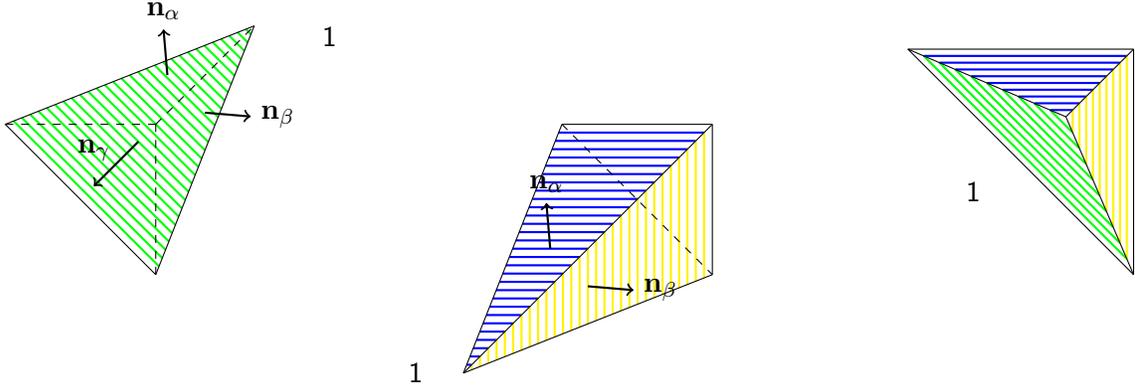
\begin{figure}
\centering
\begin{tikzpicture}[scale=2.0]

\foreach \x in {-0.1,-0.2,...,-2.9}{
\draw[green, thick] ({1+\x/6},{-\x/6},{\x}) -- ({-\x/6},{1+\x/6},{\x});
}

\draw[thin] (0,1,0) -- (1,0,0);
\draw[thin,dashed] (0,1,0) -- (1,1,0);
\draw[thin,dashed] (1,0,0) -- (1,1,0);
\draw[thin] (1,0,0) -- (0.5,0.5,-3);
\draw[thin] (0,1,0) -- (0.5,0.5,-3);
\draw[thin,dashed] (1,1,0) -- (0.5,0.5,-3);

\draw [->,thick] ({1/2},{3/4},{-3/2}) -- ({1/2},{3/4+0.99/3},{-3/2+0.2/3}); 
\draw [->,thick] ({3/4},{1/2},{-3/2}) -- ({3/4+0.99/3},{1/2},{-3/2+0.2/3}); 
\draw [->,thick] ({1/2},{1/2},{-3/3}) -- ({1/2-0.69/3},{1/2-0.69/3},{-3/4-0.23/3}); 
\filldraw [red] ({1/2},{3/4+0.99/3},{-3/2+0.2/3}) circle (0pt) node[anchor=south,black] {$\vc n_\alpha$};
\filldraw [red] ({3/4+0.99/3},{1/2},{-3/2+0.2/3}) circle (0pt) node[anchor=west,black] {$\vc n_\beta$};
\filldraw [red] ({1/2-0.69/3},{1/2-0.69/3+0.1},{-3/4-0.23/3}) circle (0pt) node[anchor=south,black] {$\vc n_\gamma$};
\filldraw [red] (0.5+0.5,0.5-0.2,-3) circle (0pt) node[anchor=south,black] {$\mt 1$};

\foreach \x in {0.1,0.2,...,2.9}{
\draw[blue, thick] ({\x/6+3.7},{1-\x/6},{\x}) -- ({1-\x/6+3.7},{1-\x/6},{\x});
}
\foreach \x in {0.1,0.2,...,2.9}{
\draw[yellow, thick] ({1-\x/6+3.7},{\x/6},{\x}) -- ({1-\x/6+3.7},{1-\x/6},{\x});
}

\draw[thin,dashed] ({0+3.7},1,0) -- ({1+3.7},0,0);
\draw[thin] ({0+3.7},1,0) -- ({1+3.7},1,0);
\draw[thin] ({1+3.7},0,0) -- ({1+3.7},1,0);
\draw[thin] ({1+3.7},0,0) -- ({0.5+3.7},0.5,3);
\draw[thin] ({0+3.7},1,0) -- ({0.5+3.7},0.5,3);
\draw[thin] ({1+3.7},1,0) -- ({0.5+3.7},0.5,3);

\filldraw [red] ({1/2+3.7},{3/4+0.99/3},-{-3/2+0.2/3}) circle (0pt) node[anchor=south,black] {$\vc n_\alpha$};
\draw [->,thick] ({1/2+3.7},{3/4},-{-3/2}) -- ({1/2+3.7},{3/4+0.99/3},-{-3/2+0.2/3}); 
\filldraw [red] ({3/4+0.99/3+3.7},{1/2},-{-3/2+0.2/3}) circle (0pt) node[anchor=west,black] {$\vc n_\beta$};
\draw [->,thick] ({3/4+3.7},{1/2},-{-3/2}) -- ({3/4+0.99/3+3.7},{1/2},-{-3/2+0.2/3}); 

\filldraw [red] (0.5+3.7-0.2,0.5,3) circle (0pt) node[anchor=east,black] {$\mt 1$};

\foreach \x in {0.1,0.2,...,1}{
\draw[green, thick] ({1.5*(0.7*\x)+6},{1.5*(1-0.3*\x)}) -- ({1.5*(1-0.3*\x)+6},{1.5*(0.7*\x)});
}
\foreach \x in {0.1,0.2,...,1}{
\draw[blue, thick] ({1.5*(1-0.3*\x)+6},{1.5*(1-0.3*\x)}) -- ({1.5*(0.7*\x)+6},{1.5*(1-0.3*\x)});
}
\foreach \x in {0.1,0.2,...,1}{
\draw[yellow, thick] ({1.5*(1-0.3*\x)+6},{1.5*(1-0.3*\x)}) -- ({1.5*(1-0.3*\x)+6},{1.5*(0.7*\x)});
}
\draw[thin] ({6},1.5) -- ({1.5+6},0);
\draw[thin] ({6},1.5) -- ({1.5+6},1.5);
\draw[thin] ({1.5+6},0) -- ({1.5+6},1.5);
\draw[thin] ({1.5+6},0) -- ({0.7*1.5+6},{0.7*1.5});
\draw[thin] ({6},1.5) -- ({0.7*1.5+6},{0.7*1.5});
\draw[thin] ({1.5+6},1.5) -- ({0.7*1.5+6},{0.7*1.5});

\filldraw [red] ({0.7*1.5+6-0.5},{0.7*1.5-0.5}) circle (0pt) node[anchor=east,black] {$\mt 1$};

\end{tikzpicture}
\caption{\label{starCrosslaminatesII} {Three different views of a pyramidal microstructure constructed with type II half-star twins (also possible with 
star twins). 
The polyhedron is divided into a blue, a yellow and a green region, in each of which we have a different exactly compatible laminate, namely $\mt F_1=\mt 1 + \vc a^*\otimes \vc n_\alpha,$ $\mt F_2=\mt 1 + \vc a^*\otimes \vc n_\beta$ and $\mt F_3=\mt 1 + \vc a^*\otimes \vc n_\gamma$, for some $\vc a^*\in\R^3\setminus\{\vc 0\}$. The three constant average deformation gradients $\mt F_1,\mt F_2,\mt F_3$ satisfy \eqref{startwinsDefIntro}, and $\vc n_\alpha,\vc n_\beta,\vc n_\gamma$ are linearly independent. 
At the three faces of the pyramid with normals $\vc n_\alpha,\vc n_\beta,\vc n_\gamma$ the martensitic laminates are compatible with austenite without an interface layer. Such pyramid can be used as a building block for three-dimensional curved interfaces. {In this picture, $\mt 1$ represents the identity matrix, the deformation gradient for the undistorted austenite phase, and denotes the austenite region.}
}
}
\end{figure}

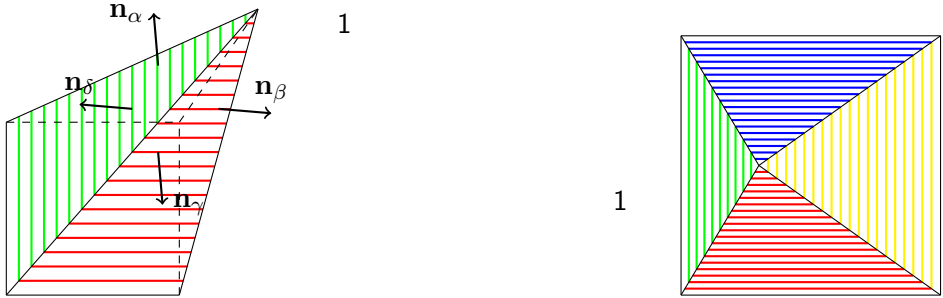
\begin{figure}
\centering
\begin{tikzpicture}[scale=2.30]

\foreach \x in {0.1,0.2,...,2}{
\draw[green, thick] ({0.3*\x/2},{1-\x*0.5/2},{-3*\x/2}) -- ({0.3*\x/2},{\x*0.5/2},{-3*\x/2});
}
\foreach \x in {0.1,0.2,...,2}{
\draw[red, thick] ({1-\x*0.7/2},{0.5*\x/2},{-3*\x/2}) -- ({0.3*\x/2},{\x*0.5/2},{-3*\x/2});
}

\draw[thin] (0,1,0) -- (0,0,0);
\draw[thin] (1,0,0) -- (0,0,0);
\draw[thin,dashed] (0,1,0) -- (1,1,0);
\draw[thin,dashed] (1,0,0) -- (1,1,0);
\draw[thin] (1,0,0) -- (0.3,0.5,-3);
\draw[thin] (0,1,0) -- (0.3,0.5,-3);
\draw[thin,dashed] (1,1,0) -- (0.3,0.5,-3);
\draw[thin] (0,0,0) -- (0.3,0.5,-3);

\draw [->,thick] ({0.15},{1/2},{-3/2}) -- ({0.15-0.99/3},{1/2},{-3/2-0.2/3}); 
\draw [->,thick] ({0.65},{1/2},{-3/2}) -- ({0.65+0.99/3},{1/2},{-3/2+0.2/3});  
\draw [->,thick] ({0.3},{1/4},{-3/2}) -- ({0.3},{1/4-0.99/3},{-3/2-0.2/3}); 
\draw [->,thick] ({0.3},{3/4},{-3/2}) -- ({0.3},{3/4+0.99/3},{-3/2+0.2/3});  

\filldraw [red] ({0.3},{3/4+0.99/3},{-3/2+0.2/3}) circle (0pt) node[anchor=east,black] {$\vc n_\alpha$};
\filldraw [red] ({0.65+0.99/3},{1/2},{-3/2+0.2/3}) circle (0pt) node[anchor=south,black] {$\vc n_\beta$};
\filldraw [red] ({0.3},{1/4-0.99/3},{-3/2-0.2/3}) circle (0pt) node[anchor=west,black] {$\vc n_\gamma$};
\filldraw [red] ({0.15-0.99/3},{1/2},{-3/2-0.2/3}) circle (0pt) node[anchor=south,black] {$\vc n_\delta$};

\filldraw [red] (0.5+0.3,0.5-0.2,-3) circle (0pt) node[anchor=south,black] {$\mt 1$};

\foreach \x in {0.1,0.2,...,1}{
\draw[green, thick] ({1.5*(0.3*\x)+3.9},{1.5*(0.5*\x)}) -- ({1.5*(0.3*\x)+3.9},{1.5*(1-0.5*\x)});
}
\foreach \x in {0.1,0.2,...,2}{
\draw[red, thick] ({1.5*(0.3*\x/2)+3.9},{1.5*(0.5*\x/2)}) -- ({1.5*(1-0.7*\x/2)+3.9},{1.5*(0.5*\x/2)});
}
\foreach \x in {0.1,0.2,...,2}{
\draw[yellow, thick] ({1.5*(1-0.7*\x/2)+3.9},{1.5*(1-0.5*\x/2)}) -- ({1.5*(1-0.7*\x/2)+3.9},{1.5*(0.5*\x/2)});
}
\foreach \x in {0.1,0.2,...,2}{
\draw[blue, thick] ({1.5*(1-0.7*\x/2)+3.9},{1.5*(1-0.5*\x/2)}) -- ({1.5*(0.3*\x/2)+3.9},{1.5*(1-0.5*\x/2)});
}

\draw[thin] ({3.9},1.5) -- ({0+3.9},0);
\draw[thin] ({3.9},0) -- ({1.5+3.9},0);
\draw[thin] ({3.9},1.5) -- ({1.5+3.9},1.5);
\draw[thin] ({1.5+3.9},0) -- ({1.5+3.9},1.5);
\draw[thin] ({1.5+3.9},0) -- ({0.3*1.5+3.9},{0.5*1.5});
\draw[thin] ({3.9},1.5) -- ({0.3*1.5+3.9},{0.5*1.5});
\draw[thin] ({1.5+3.9},1.5) -- ({0.3*1.5+3.9},{0.5*1.5});
\draw[thin] ({0+3.9},0) -- ({0.3*1.5+3.9},{0.5*1.5});
\filldraw [red] ({0.7*1.5+3.9-1.3},{0.7*1.5-0.5}) circle (0pt) node[anchor=east,black] {$\mt 1$};

\end{tikzpicture}
\caption{\label{starCrosslaminatesIIbis}  {Two different views of a pyramidal microstructure constructed with type II star twins. 
The polyhedron is divided into a blue, a yellow, a red and a green region, in each of which we have a different exactly compatible laminate, namely $\mt F_1=\mt 1 + \vc a^*\otimes \vc n_\alpha,$ $\mt F_2=\mt 1 + \vc a^*\otimes \vc n_\beta,$ $\mt F_3=\mt 1 + \vc a^*\otimes \vc n_\gamma$ and $\mt F_4=\mt 1 + \vc a^*\otimes \vc n_\delta$, for some $\vc a^*\in\R^3\setminus\{\vc 0\}$. The four constant average deformation gradients $\mt F_1,\mt F_2,\mt F_3,\mt F_4$ satisfy \eqref{startwinsDefIntro}, and $\vc n_\alpha,\vc n_\beta,\vc n_\gamma,\vc n_\delta$ are linearly independent. 
At the four faces of the pyramid with normals $\vc n_\alpha,\vc n_\beta,\vc n_\gamma,\vc n_\delta$ the martensitic laminates are compatible with austenite without an interface layer. Such pyramid can be used as a building block for three-dimensional curved interfaces.
{In this picture, $\mt 1$ represents the identity matrix, the deformation gradient for the undistorted austenite phase, and denotes the austenite region.}
}
}
\end{figure}
}

\section{Star twins in cubic to monoclinic transformations}
\label{starsecMono}
In this section we want characterise the matrices $\mt U_i$ as in \eqref{cubictomono} such that there exist type I and type II  star twins. We start with the following simple lemma
\begin{lemma}
\label{rotaz lemma}
Let $\vc v\in\R^3$, $\vc v\neq 0$, then $\mt Q\in SO(3)$ satisfies $\mt Q\vc v=-\vc v$ if and only if $\mt Q$ is a rotation of $\pi$ and axis $\vc u_{\mt R}\in\mathbb{S}^2$, that is, $\mt R=2\vc u_{\mt R}\otimes\vc u_{\mt R}-\mt 1$, and $\vc u_{\mt R}\cdot\vc v=0$. The matrix $\mt Q\in SO(3)$, $\mt Q\neq\mt1$ satisfies $\mt Q\vc v=\vc v$ if and only if the axis of $\mt Q$ is parallel to $\vc v$.
\end{lemma}
\begin{proof}
We can write $\mt R$ in terms of its angle of rotation $\theta$ and its rotation axis $\vc u_{\mt R}$ as
$$
\mt R = \cos(\theta)\mt 1+\sin(\theta)[\vc u_{\mt R}]_\times+(1-\cos(\theta))\vc u_{\mt R}\otimes\vc u_{\mt R},
$$
where $[\vc u_{\mt R}]_\times$ is the cross product matrix of $\vc u_{\mt R}$. Therefore, after decomposing $\vc v$ into a parallel and an orthogonal part to $\vc u_{\mt R}$, $\vc v_{\vc u_{\mt R}}$ and $\vc v_{\vc u_{\mt R}^\perp}$, we get that $\mt R\vc v= - \vc v$ is equivalent to
$$
\cos(\theta)\bigl(\vc v_{\vc u_{\mt R}}+\vc v_{\vc u_{\mt R}^\perp}) + \sin(\theta)\bigl({\vc u_{\mt R}}\times\vc v_{\vc u_{\mt R}^\perp}	\bigr)+(1-\cos(\theta))\vc u_{\mt R}\bigl(\vc u_{\mt R}\cdot\vc v_{\vc u_{\mt R}}\bigr)=-\vc v_{\vc u_{\mt R}}-\vc v_{\vc u_{\mt R}^\perp}.
$$
Thus, multiplying the equation by $\vc v_{\vc u_{\mt R}^\perp}$ we get that either $\theta = \pi$, or $\vc v_{\vc u_{\mt R}^\perp}=\vc 0$. But multiplying the equation by $\vc v_{\vc u_{\mt R}}$ leads to $\vc v_{\vc u_{\mt R}}=\vc 0$, and therefore $\theta=\pi$. The second statement follows from the definition of rotation axis.
\end{proof}
The lemma below states that, if a type II twin generated by $\mt U_i,\mt U_j$, where $(i,j)$ are a pair in column (A) (resp. (B)) of Table \ref{Table 1}, is a type I (or a type II) star twin, then all the other type I (resp. type II) twins generated by $\mt U_k,\mt U_l$, where $(k,l)$ is a generic pair in column (A) (resp. (B)) of Table \ref{Table 1}, is a type II star twin. In what follows, we denote by $\mathcal{P}_{24}$ the symmetry group of cubic austenite.
\begin{lemma}
\label{lemma tanti}
Let $\mathcal{M}:=\cup_{i=1}^{12}\mt U_i$, where the $\mt U_i$ are the positive definite matrices, with $a,b,c,d>0$, given in \eqref{cubictomono}. Suppose also that $\mt U_i,\mt U_j$, $i\neq j =1,\dots,12$, generate a type I/II star twin (resp. a half-star twin). Then, for every $\mt Q\in \mathcal{P}_{24}$ 
the type I/II twins generated by $\mt Q\mt U_i\mt Q^T,\mt Q\mt U_j\mt Q^T$ are star twins (resp. a half-star twin).
\end{lemma}
\begin{proof}
We prove the result for type II star twins, the case of type I star twins follows a similar argument.\\
The fact that $\mt U_k,\mt U_l$, where $\mt U_k=\mt Q\mt U_i\mt Q^T,\mt U_l=\mt Q\mt U_j\mt Q^T$ for some $\mt Q\in \mathcal{P}_{24}$, generate a type II twin satisfying the cofactor conditions is a consequence of Proposition \ref{tutti quanti}. Let $\mu^*\in(0,1)$, $\mt Q_1,\mt Q_2,\mt Q_3\in SO(3)$ be such that $(\mt U_i,\vc a_{II},\vc n_{II})$, the type II twin generated by $\mt U_i,\mt U_j$, satisfies Definition \ref{Def star}. Let also $\hat{\mt Q}_m = \mt Q \mt Q_m\mt Q^T$ for every $m=1,2,3$, and $(\mt U_k,\hat{\vc b}_{II},\hat{\vc m}_{II})$ be the type II twin generated by $\mt U_k,\mt U_l$. We claim that $\mu^*$, $\hat{\mt Q}_1,\hat{\mt Q}_2,\hat{\mt Q}_3\in SO(3)$ are such that $(\mt U_k,\hat{\vc b}_{II},\hat{\vc m}_{II})$ satisfies \ref{T1s}--\ref{T3s} in Definition \ref{Def star}. Indeed, we recall that $\mt Q\mt V\mt Q^T\in \mathcal{M}$ for every $\mt V\in\mathcal M,\mt Q\in \mathcal P_{24}$ (see e.g., \cite{Batt}). Thus, as
$$
\hat{\mt Q}_m\mt U_k\hat{\mt Q}_m = \mt Q\mt Q_m\mt U_i \mt Q_m\mt Q,\qquad \hat{\mt Q}_m\mt U_l\hat{\mt Q}_m = \mt Q\mt Q_m\mt U_j \mt Q_m\mt Q,
$$
and given that $\mt U_i,\mt U_j$ satisfy \ref{T1s}, we can write
$$
\hat{\mt Q}_m\mt U_k\hat{\mt Q}_m ,\, \hat{\mt Q}_m\mt U_l\hat{\mt Q}_m \in \mathcal{M},\qquad\text{for all $m=1,2,3$,}
$$
that is, \ref{T1s}. As $\vc a_{\mt U_k}=\mt Q\vc a_{\mt U_i},\vc a_{\mt U_l}=\mt Q\vc a_{\mt U_j}$, we also have
$$
\hat{\mt Q}_m  (\mu^*\vc a_{\mt U_k}+(1-\mu^*)\vc a_{\mt U_l})=\mt Q (\mu^*\vc a_{\mt U_i}+(1-\mu^*)\vc a_{\mt U_j}) = \mu^*\vc a_{\mt U_k}+(1-\mu^*)\vc a_{\mt U_l}, \qquad\text{for all $m=1,2,3$,}
$$
which is \ref{T2s}. Finally, as $\hat{\vc m}_{II}=\mt Q\vc m_{II}$, we have
\[
\begin{split}
\Bigl(\bigl(\hat{\mt Q}_1\hat{\vc m}_{II}\bigl)\times\bigl(\hat{\mt Q}_2\hat{\vc m}_{II}\bigl)\Bigr)\cdot \bigl(\hat{\mt Q}_3\hat{\vc m}_{II}\bigl) = \Bigl(\bigl(\mt Q_1{\vc m}_{II}\bigl)\times\bigl(\mt Q_2{\vc m}_{II}\bigl)\Bigr)\cdot \bigl(\mt Q_3{\vc m}_{II}\bigl)\neq 0,\\
\Bigl(\bigl(\hat{\mt Q}_m\hat{\vc m}_{II}\bigl)\times\bigl(\hat{\mt Q}_n\hat{\vc m}_{II}\bigl)\Bigr)\cdot \hat{\vc m}_{II} = \Bigl(\bigl(\mt Q_m{\vc m}_{II}\bigl)\times\bigl(\mt Q_n{\vc m}_{II}\bigl)\Bigr)\cdot {\vc m}_{II}\neq 0,
\end{split}
\]
for every $m,n=1,2,3$, which is \ref{T3s}.
\end{proof}
We are now in the position to prove the following result:
\begin{theorem}
\label{main thm II}
Let $\mathcal{M}:=\cup_{i=1}^{12}\mt U_i$, where the $\mt U_i$ are the positive definite matrices, with $a,b,c,d>0$, given in \eqref{cubictomono}{\color{red}}. Then, a type II twin generated by $\mt U_i,\mt U_j$, $i\neq j =1,\dots,12$, satisfying the cofactor conditions is a half-star twin if and only if the eigenvalues of $\mt U_1$, namely $0<\lambda_1\leq \lambda_2=1 \leq \lambda_3$ satisfy 
one of the following conditions
\beq
\label{half star 2 twins cond}
\begin{split}
4d^2\bigl( d^2+\lambda^2_3-2\bigr) = (d-\lambda_3)^2(1-d^2),\qquad&\text{if $d\neq 1$ and $\lambda_1=d,$}\\
4d^2\bigl( d^2+\lambda^2_1-2\bigr) = (d-\lambda_1)^2(1-d^2),\qquad&\text{if $d\neq 1$ and $\lambda_3=d,$}\\
\lambda_1^2\bigl(5 \lambda_3^2-1\bigr) = 8\lambda_1\lambda_3-5+\lambda_3^2,\qquad&\text{if $d=1$ and $\lambda_3>1$.}
\end{split}
\eeq
A type II twin generated by $\mt U_i,\mt U_j$, $i\neq j =1,\dots,12$, satisfying the cofactor conditions is a star twin if and only if $a\neq c$ 
and one of the following holds
\beq
\label{star 2 twins cond}
\begin{split}
d^2\bigl( d^2+\lambda^2_3-2\bigr) = (d-\lambda_3)^2(1-d^2),\qquad&\text{if $d\neq 1$ and $\lambda_1=d,$ }\\
d^2\bigl( d^2+\lambda^2_1-2\bigr) = (d-\lambda_1)^2(1-d^2),\qquad&\text{if $d\neq 1$ and $\lambda_3=d$.}
\end{split}
\eeq
\end{theorem}
\begin{remark}
\label{quante condizioni}
\rm 
It is possible to solve the equations in \eqref{half star 2 twins cond}--\eqref{star 2 twins cond} to get a direct relation between the eigenvalues of the $\mt U_i$'s. Indeed, \eqref{half star 2 twins cond} is equivalent to
\begin{align*}
\lambda_1 = \frac {d-d^3- d2\sqrt{2}\sqrt{6d^2-1-3d^4}}{1-5d^2},\qquad &\text{if }\lambda_3=d\in (1,\sqrt{1+2^{\frac12}3^{-\frac12}}),\\
\lambda_1 = \frac {d-d^3 + d2\sqrt{2}\sqrt{6d^2-1-3d^4}}{1-5d^2},\qquad &\text{if }\lambda_3=d\in (\sqrt{5^{-1}9},\sqrt{1+2^{\frac12}3^{-\frac12}}),\\
\lambda_3 = \frac {d-d^3- d2\sqrt{2}\sqrt{6d^2-1-3d^4}}{1-5d^2},\qquad &\text{if }\lambda_1=d\in (\sqrt{1-2^{\frac12}3^{-\frac12}},1),\\
\lambda_1= \frac{4\lambda_3\pm\sqrt5(\lambda_3^2-1) }{5\lambda_3^2-1},\qquad& \text{if $\lambda_3>1,\, d=1$}
,
\end{align*}
when $0<\lambda_1\leq 1 \leq \lambda_3.$ Similarly, the condition \eqref{star 2 twins cond} can be rewritten as
\begin{align*}
\lambda_1 = \frac {d-d^3- d\sqrt{6d^2-2-3d^4}}{1-2d^2},\qquad &\text{if }\lambda_3=d\in (1,\sqrt{1+3^{-\frac12}}),\\
\lambda_1 = \frac {d-d^3+ d\sqrt{6d^2-2-3d^4}}{1-2d^2},\qquad &\text{if }\lambda_3=d\in (\sqrt{2^{-1}3},\sqrt{1+3^{-\frac12}}),\\
\lambda_3 = \frac {d-d^3- d\sqrt{6d^2-2-3d^4}}{1-2d^2},\qquad &\text{if }\lambda_1=d\in (\sqrt{1-3^{-\frac12}},1),
\end{align*}
when $0<\lambda_1\leq 1 \leq \lambda_3.$
\end{remark}
\begin{proof}
We deal with the case where the twins generated by $\mt U_i,\mt U_j$, with $(i,j)$ in column (A) of Table \ref{Table 1}, satisfy the cofactor conditions. The case where $(i,j)$ are in column (B) of Table \ref{Table 1} can be treated similarly and leads to the same results. Furthermore, thanks to Lemma \ref{lemma tanti} and Proposition \ref{tutti quanti} we can carry on the proof by considering just $(\mt U_1,\vc b_{II},\vc m_{II})$ the type II twin generated by $\mt U_1,\mt U_{11}$. We divide the proof in three cases, based on the eigenvalues of $\mt U_1$ $\lambda_1\leq \lambda_2=1\leq \lambda_3$: the cases are $\lambda_2=1>d$, $\lambda_2 = 1<d$ and $0<
\lambda_1<d=1<\lambda_3$. We can exclude the case where both $d$ and another eigenvalue are equal to $1$ because, by \eqref{l1}--\eqref{tII} this would imply $b=0$ and $\mt U_1=\mt U_{11}$. \\

Case 1: $0<d < 1 \leq \lambda$, $\lambda=ac-b^2$. Here, the eigenvectors of $\mt U_1$ related to $\lambda_1 = d,\lambda_2 =1$, and $\lambda_3=\lambda$ are 
$$
\vc e_1=(0,0,1)^T,\qquad \vc e_2 = (-z_2,z_1,0)^T,\qquad \vc e_3 = (z_1,z_2,0)^T,
$$ 
while the ones for $\mt U_{11}$ are given by  
$$
\vc v_1=(1,0,0)^T,\qquad \vc v_2 = (0,z_1,-z_2)^T,\qquad \vc v_3 = (0,z_2,z_1)^T.
$$ 
Here, $z_1,z_2$ are such that $\frac{z_1}{z_2}=\frac{\lambda-c}{b}$, $z_1^2+z_2^2=1$, and we assume without loss of generality that $z_1>0.$ We can neglect the cases $z_1 = 0$ and $z_2=0$, as they would imply $b=0$ and $\mt U_1=\mt U_{11}$. Using \cite[Prop. 4]{BallJames1} we get that
$$
\mt R_1^\pm \mt U_1 = \mt 1+\vc a_1^\pm \otimes \vc n_1^\pm,\qquad \mt R_{11}^\pm \mt U_{11} = \mt 1+\vc a_{11}^\pm \otimes \vc n_{11}^\pm,
$$
where
\begin{align}
\label{bm1}\vc a_1^\pm = \beta_0(-\lambda\eta_1\vc e_1\pm d\eta_2\vc e_3),\quad \vc n_1^\pm = \eta_1\vc e_1\pm\eta_2\vc e_3,\\
\label{bm2}\vc a_{11}^\pm = \beta_0(-\lambda\eta_1\vc v_1\pm d\eta_2\vc v_3),\quad \vc n_{11}^\pm = \eta_1\vc v_1\pm\eta_2\vc v_3,
\end{align}
and
$$
\eta_1=-\frac{\sqrt{1-d^2}}{\sqrt{\lambda^2-d^2}},\qquad \eta_2=\frac{\sqrt{\lambda^2-1}}{\sqrt{\lambda^2-d^2}},\qquad \beta_0 = \lambda - d.
$$
Since the cofactor conditions are satisfied, by Theorem \ref{TypeII} we know that one of the following identities must hold
$$
\vc n_1^+ \times \vc n_{11}^+=0,\qquad \vc n_1^-\times \vc n_{11}^-=0,\qquad\vc n_1^- \times \vc n_{11}^+=0,\qquad \vc n_1^+\times \vc n_{11}^-=0.
$$
Keeping in mind that \eqref{l1}--\eqref{tII} imply $\eta_1=-z_1\eta_2$, we get that, in the notation of Theorem \ref{TypeII}, the only possibility up to a sign change is
\begin{align*}
\vc a_{\mt U_1}=\beta_0(d\eta_1, -z_2d\eta_2,&-\lambda\eta_1),\qquad \vc a_{\mt U_{11}}=\beta_0(-\lambda\eta_1, -z_2d\eta_2,d\eta_1),\\
&\vc m_{II}=(\eta_1, -z_2\eta_2,\eta_1),
\end{align*}
Therefore, 
\beq
\label{mix b 1}
\mu \vc a_{\mt U_1}+(1-\mu)\vc a_{\mt U_{11}} = \eta_1\beta_0 \Bigl(\mu d-(1-\mu)\lambda,\, d\frac{z_2}{z_1},\,-\mu \lambda+(1-\mu)d\Bigr).
\eeq
We are now interested in finding which rotations $\mt R\in SO(3)$, and which $\mu\in[0,1]$ are such that \ref{T1s}--\ref{T3s} are satisfied. First, we claim that $\mt R\in\mathcal{P}_{24}$ and $\mt R$ is a rotation whose angle and axis are given in the first column of Table \ref{Table 1}. Indeed, given $\mt R\in SO(3)$, by Proposition \ref{tutti quanti} $(\mt R\mt U_1\mt R^T,\mt R\vc b_{II},\mt R\vc m_{II})$ is a type II twin generated by $\mt R\mt U_1\mt R^T,\mt R\mt U_{11}\mt R^T$ and satisfies the cofactor conditions. Therefore, \ref{T1s} together with Remark \ref{dueconuno} imply that $\mt R\mt U_1\mt R^T,\mt R\mt U_{11}\mt R^T$ must be a pair in column (A) of Table \ref{Table 1} (or in the second column of Table \ref{Table 2} in the degenerate case where $a=c$). As a consequence, there exists $\mt Q\in \mathcal{P}_{24}$, with $\mt Q$ being a rotation of angle and axis given in the first column of Table \ref{Table 1}, such that $\mt Q\mt U_1\mt Q^T = \mt R\mt U_1\mt R^T$ and $\mt Q\mt U_{11}\mt Q^T=\mt R\mt U_{11}\mt R^T$. Thus, $\mt U_1 = \mt Q^T\mt R\mt U_1\mt R^T\mt Q$, $\mt U_{11}=\mt Q^T\mt R\mt U_{11}\mt R^T\mt Q$, and Remark \ref{StrictL} together with Lemma \ref{rotaz lemma}, imply that $\mt Q^T\mt R = \mt 1$ or $\mt Q^T\mt R = (2\vc v\otimes \vc v-\mt 1)$ where $\vc v\in\mathbb{S}^2$ must be at the same time an eigenvector for $\mt U_1$ and for $\mt U_{11}$. Under our hypotheses, there exists no $\vc v\in\mathbb{S}^2$ being at the same time an eigenvector for $\mt U_1$ and for $\mt U_{11}$. Thus $\mt Q = \mt R$ concluding the proof of the claim.

Now, by Lemma \ref{rotaz lemma} we have to check for which $\mu$, and which rotation axis $\vc u\in\mathbb{S}^2$ among the ones given in the first column of Table \ref{Table 1}, $\mu \vc a_{\mt U_1}+(1-\mu)\vc a_{\mt U_{11}}$ is parallel or perpendicular to $\vc u$. Table \ref{Table 0} below illustrates the possible cases.
\begin {table}[h]
\begin{center}
\begin{tabular}{ |l|l|l| }
\hline
Rotation axis $\vc e$ & Parallel & Orthogonal \\ \hline
 $(1,0,0)$ & never & $\mu = \frac{\lambda}{\lambda+d}$\\
 $(0,1,0)$ & $d=\lambda$ (never) & never\\
 $(0,0,1)$ & never & $\mu = \frac{d}{\lambda+d}$\\
 $(1,0,1)$ & never & $\lambda=d$ (never)\\
 $(1,0,-1)$& never & $\mu =\frac12$ (but $\mt R\vc n_{II}=-\vc n_{II})$\\
 $(1,1,0)$ & $\mu = \frac{d}{\lambda+d}$ and $\frac{z_2}{z_1}=\frac{d-\lambda}{d}$ & $d\frac{z_2}{z_1}=(1-\mu)\lambda-\mu d$\\
 $(1,-1,0)$  & $\mu = \frac{d}{\lambda+d}$ and $\frac{z_2}{z_1}=\frac{\lambda-d}{d}$ & $d\frac{z_2}{z_1}=-(1-\mu)\lambda+\mu d$\\
 $(0,1,1)$  & $\mu = \frac{\lambda}{\lambda+d}$ and $\frac{z_2}{z_1}=\frac{d-\lambda}{d}$ & $d\frac{z_2}{z_1}=-(1-\mu)d+\mu \lambda$\\
 $(0,-1,1)$  & $\mu = \frac{\lambda}{\lambda+d}$ and $\frac{z_2}{z_1}=\frac{\lambda-d}{d}$ & $d\frac{z_2}{z_1}=(1-\mu)d-\mu \lambda$\\
 \hline
\end{tabular}
\end{center}
\caption {
\label{Table 0} Conditions on $\mu$ to have $(2\vc e\otimes\vc e -\mt 1)(\mu \vc a_{\mt U_1}+(1-\mu)\vc a_{\mt U_{11}}) = \mu \vc a_{\mt U_1}+(1-\mu)\vc a_{\mt U_{11}},$ under the assumption that $d<1\leq \lambda.$} 
\end{table}
As a result, we have two different rotation axes $\vc u$ which are either parallel or perpendicular to $\mu\vc a_{\mt U_1}+(1-\mu)\vc a_{\mt U_{11}}$ if and only if
%
\begin{enumerate}
\item $\mu = \frac{1}{2}$ and $\frac{z_2}{z_1}=\pm\frac{d-\lambda}{2d}$, 
\item $\mu = \frac{d}{\lambda+d}$ and $\frac{z_2}{z_1}=\pm\frac{d-\lambda}{d}$, 
\item $\mu = \frac{\lambda}{\lambda+d}$ and $\frac{z_2}{z_1}=\pm\frac{d-\lambda}{d}$.
\end{enumerate}
We remark that, in the last two cases, the rotation axes which are either parallel or perpendicular to $\mu\vc a_{\mt U_1}+(1-\mu)\vc a_{\mt U_{11}}$ are actually three and not two. As $(\mt U_1,\vc b_{II}, \vc m_{II})$ satisfies the cofactor conditions as a type II twin, from \eqref{l1}--\eqref{tII} we deduce that
$$
b= \frac{\sqrt{-d^4+d^2(3-\lambda^2)+(\lambda^2-2)}}{1+\lambda},\qquad \lambda - c = \frac{1-d^2}{1+\lambda},
$$
which implies
$$
\frac{z_2}{z_1}=\frac{\sqrt{-d^4+d^2(3-\lambda^2)+(\lambda^2-2)}}{1-d^2}.
$$
After some computations we finally deduce that, under the hypothesis $d<1\leq \lambda$, there exist at least two different rotations such that \ref{T1s}--\ref{T2s} are satisfied if and only if
\begin{align*}
\frac{z_2}{z_1}=\pm\frac{d-\lambda}{2d} &\Longleftrightarrow 
d^24\bigl( d^2+\lambda^2-2\bigr) = (d-\lambda)^2(1-d^2),\\
\frac{z_2}{z_1}=\pm\frac{d-\lambda}{d} &\Longleftrightarrow d^2\bigl( d^2+\lambda^2-2\bigr) = (d-\lambda)^2(1-d^2).
\end{align*}
In order to be prove that these conditions are equivalent to a type II half-star or star twin, we need to check that \ref{T3s} is satisfied. This is always the case for half-star twins. For type II star twins this happens if and only if $\eta_1\neq - z_2\eta_2$. But, recalling that $\eta_1=-z_1\eta_2$, we need $z_1^2\neq z_2^2.$ As $\frac{z_1^2}{z_2^2}=\frac{a-1}{c-1}$, we hence deduce that \ref{T3s} is satisfied by type II star twins if and only if $a\neq c$. \\

Case 2: $0<\lambda \leq 1 <d$, $\lambda=ac-b^2$. Here, the eigenvectors related to $\lambda,1,d$ for $\mt U_1$ and $\mt U_{11}$ are respectively $\vc e_i$ and $\vc v_i$, given by 
\begin{align*}
\vc e_1=(z_1,z_2,0)^T,\qquad \vc e_2 = (-z_2,z_1,0)^T,\qquad \vc e_3 = (0,0,1)^T ,\\
\vc v_1=(0,z_2,z_1)^T,\qquad \vc v_2 = (0,z_1,-z_2)^T,\qquad \vc v_3 = (1,0,0)^T.
\end{align*} 
where, again, we assume without loss of generality that $z_1>0$. Here
$$
\mt R_1^\pm \mt U_1 = \mt 1+\vc a_1^\pm \otimes \vc n_1^\pm,\qquad \mt R_{11}^\pm \mt U_{11} = \mt 1+\vc a_{11}^\pm \otimes \vc n_{11}^\pm,
$$
where, $\vc n_1^\pm,\vc n_{11}^\pm$ are still given by \eqref{bm1}--\eqref{bm2}, and 
$$
\vc a_1^\pm = \beta_0(-d\eta_1\vc e_1\pm \lambda\eta_2\vc e_3),\qquad
\vc a_{11}^\pm = \beta_0(-d\eta_1\vc v_1\pm \lambda\eta_2\vc v_3),
$$
but this time
$$
\eta_1=-\frac{\sqrt{1-\lambda^2}}{\sqrt{d^2-\lambda^2}},\qquad \eta_2=\frac{\sqrt{d^2-1}}{\sqrt{d^2-\lambda^2}},\qquad \beta_0 = d-\lambda.
$$
From these, we can again find
$$
\vc a_{\mt U_1}=\beta_0(d\eta_2, -z_2d\eta_1,-\lambda\eta_2),\qquad \vc a_{\mt U_{11}}=\beta_0(-\lambda\eta_2, -z_2d\eta_1,d\eta_2),\qquad\vc m_{II}=(-\eta_2, z_2\eta_1,-\eta_2),
$$
and $z_1\eta_1=-\eta_2$. 
Thus,
$$
\mu \vc a_{\mt U_1}+(1-\mu)\vc a_{\mt U_{11}} = \eta_2\beta_0 \Bigl(\mu d-(1-\mu)\lambda,\, d\frac{z_2}{z_1},\,-\mu \lambda+(1-\mu)d\Bigr),
$$
which is the same as \eqref{mix b 1}. We can hence argue as in the first case, and deduce that, if $d>1$, $(\mt U_1,\mt U_{11})$ generate a half-star twin if and only if 
\begin{align*}
d^24\bigl( d^2+\lambda^2-2\bigr) = (d-\lambda)^2(1-d^2).
\end{align*}
In the same way, $(\mt U_1,\mt U_{11})$ generates a star twin if and only if $a\neq c$, and
\begin{align*}
d^2\bigl( d^2+\lambda^2-2\bigr) = (d-\lambda)^2(1-d^2).\end{align*}
\\
Case 3: $0<\lambda_1 < d=1 <\lambda_3$, where $\lambda_1\lambda_3=ac-b^2$ and $\lambda_1+\lambda_3=a+c$. Repeating the above arguments, this time we deduce that
\begin{align*}
\vc a_{\mt U_1}=\beta_0\eta_1(-z_1(\lambda_1 +\lambda_3), (\lambda_1z_1^2z_2^{-1}-\lambda_3z_2),0),&\qquad 
\vc a_{\mt U_{11}}=\beta_0\eta_1(0,(\lambda_1z_1^2z_2^{-1}-\lambda_3z_2),-z_1(\lambda_1 +\lambda_3)),\\
\vc m_{II}&=(0, \eta_1z_2^{-1},0),
\end{align*}
where
$$
\eta_1=-\frac{\sqrt{1-\lambda_1^2}}{\sqrt{\lambda_3^2-\lambda_1^2}},\qquad \eta_2=\frac{\sqrt{\lambda_3^2-1}}{\sqrt{\lambda_3^2-\lambda_1^2}},\qquad \beta_0 = \lambda_3-\lambda_1,
$$
and $z_1\eta_1=z_2\eta_2$. Thus,
$$
\mu \vc a_{\mt U_1}+(1-\mu)\vc a_{\mt U_{11}} = -\eta_1\beta_0 \Bigl(\mu z_1(\lambda_1 +\lambda_3),\, (\lambda_3z_2-\lambda_1z_1^2z_2^{-1}),\,(1-\mu) z_1(\lambda_1 +\lambda_3)\Bigr).
$$
Define
$$
\beta_1=z_1(\lambda_1 +\lambda_3),\qquad \beta_2 = (\lambda_3z_2-\lambda_1z_1^2z_2^{-1}).
$$
In this case, Table \ref{Table 0} becomes Table \ref{Table 01}.
\begin {table}[h]
\begin{center}
\begin{tabular}{ |l|l|l| }
\hline
Rotation axis $\vc e$ & Parallel & Orthogonal \\ \hline
 $(1,1,0)$ & $\mu = 1$ and $\beta_1=\beta_2$ & $\mu\beta_1=-\beta_2$\\
 $(1,-1,0)$ & $\mu = 1$ and $\beta_1=-\beta_2$ & $\mu\beta_1=\beta_2$\\
 $(0,1,1)$ & $\mu = 0$ and $\beta_1=\beta_2$ & $(1-\mu)\beta_1=-\beta_2$\\
 $(0,-1,1)$ & $\mu = 0$ and $\beta_1=-\beta_2$ & $(1-\mu)\beta_1=\beta_2$\\
 \hline
\end{tabular}
\end{center}
\caption {\label{Table 01}Conditions on $\mu$ to have $(2\vc e\otimes\vc e -\mt 1)(\mu \vc a_{\mt U_1}+(1-\mu)\vc a_{\mt U_{11}}) = \mu \vc a_{\mt U_1}+(1-\mu)\vc a_{\mt U_{11}},$ under the assumption that $\lambda_1<1< \lambda_3.$} 
\end{table}
Here we neglected all the rotations whose axes, due to Lemma \ref{rotaz lemma}, do not satisfy \ref{T3s}.
Thus, the only option to form a type II half-star twin is
$$\mu = \frac{1}{2}\quad\text{ and }\quad\frac12\beta_1=\pm\beta_2.$$
As $0<\lambda_1,\lambda_3$ and $\lambda_1,\lambda_3\neq 1$, $\frac12\beta_1=\pm\beta_2$ simplifies to
$$
\lambda_1^2(5\lambda_3^2-1)-8\lambda_1\lambda_3-\lambda_3^2+5=0,
$$
and the solutions satisfying $\lambda_3>1>\lambda_1$ are 
\[
\lambda_1= \frac{4\lambda_3\pm\sqrt5(\lambda_3^2-1) }{5\lambda_3^2-1},\qquad \text{$\lambda_3>1$}.
\]
\end{proof}

In a similar way, we can prove
\begin{theorem}
\label{main thm I}
Let $\mathcal{M}:=\cup_{i=1}^{12}\mt U_i$, where the $\mt U_i$ are the positive definite matrices, with $a,b,c,d>0$, given in \eqref{cubictomono}. Then, a type I twin generated by $\mt U_i,\mt U_j$, $i\neq j =1,\dots,12$, satisfying the cofactor conditions is a I half-star twin if and only if the eigenvalues of $\mt U_1$, namely $\lambda_1 \leq \lambda_2=1 \leq \lambda_3$ satisfy one of the following conditions
\beq
\label{half star 1 twins cond}
\begin{split}
4\bigl( 2d^2\lambda^2_3-\lambda^2_3-d^2\bigr) = (d-\lambda_3)^2(1-d^2),\qquad&\text{if $d\neq 1$ and $d=\lambda_1,$}\\
4\bigl( 2d^2\lambda^2_1-\lambda^2_1-d^2\bigr) = (d-\lambda_1)^2(1-d^2),\qquad&\text{if $d\neq 1$ and $d=\lambda_3,$}\\
\lambda_1^2\bigl(5 \lambda_3^2-1\bigr) = 8\lambda_1\lambda_3-5+\lambda_3^2,\qquad&\text{if $d=1$ and $\lambda_3>1$.}
\end{split}
\eeq
A type I twin generated by $\mt U_i,\mt U_j$, $i\neq j =1,\dots,12$, satisfying the cofactor conditions is a I star twin if and only if $a\neq c$ and one of the following holds
\beq
\label{star 1 twins cond}
\begin{split}
\bigl( 2d^2\lambda^2_3-\lambda^2_3-d^2\bigr) = (d-\lambda_3)^2(1-d^2),\qquad&\text{if $d\neq 1$ and $d=\lambda_1,$ }\\
\bigl( 2d^2\lambda^2_1-\lambda^2_1-d^2\bigr) = (d-\lambda_1)^2(1-d^2),\qquad&\text{if $d\neq 1$ and $d=\lambda_3$.}
\end{split}
\eeq
\end{theorem}
\begin{remark}
\label{quante condizioni 1}
\rm 
It is possible to solve the equations in \eqref{half star 1 twins cond}--\eqref{star 1 twins cond} to get a direct relation between the eigenvalues of the $\mt U_i$'s. Indeed, \eqref{half star 1 twins cond} is equivalent to
\begin{align*}
\lambda_1 = \frac {d^3-d + d2\sqrt{2}\sqrt{6d^2-3-d^4}}{9d^2-5},\qquad &\text{if }\lambda_3=d\in (1,\sqrt{3+\sqrt6}),\\
\lambda_1 = \frac {d^3-d - d2\sqrt{2}\sqrt{6d^2-3-d^4}}{9d^2-5},\qquad &\text{if }\lambda_3=d\in (\sqrt{5},\sqrt{3+\sqrt6}),\\
\lambda_3 = \frac {d^3-d + d2\sqrt{2}\sqrt{6d^2-3-d^4}}{9d^2-5},\qquad &\text{if }\lambda_1=d\in (\sqrt{3-\sqrt6},1),\\
\lambda_3 = \frac {d^3-d - d2\sqrt{2}\sqrt{6d^2-3-d^4}}{9d^2-5},\qquad &\text{if }\lambda_1=d\in (\sqrt{3-\sqrt6},\frac{\sqrt5}{3}),\\
{
\lambda_1= \frac{4\lambda_3\pm\sqrt5(\lambda_3^2-1) }{5\lambda_3^2-1},}\qquad& \text{if $\lambda_3>1,\, d=1$}
,
\end{align*}
when $0<\lambda_1\leq 1 \leq \lambda_3.$ Similarly, the condition \eqref{star 2 twins cond} can be rewritten as
\begin{align*}
\lambda_1 = \frac {d^3-d+ d\sqrt{6d^2-3-2d^4}}{3d^2-2},\qquad &\text{if }\lambda_3=d\in (1,\sqrt{\frac{(3+\sqrt3)}2}),\\
\lambda_1 = \frac {d^3-d- d\sqrt{6d^2-3-2d^4}}{3d^2-2},\qquad &\text{if }\lambda_3=d\in (\sqrt 2,\sqrt{\frac{(3+\sqrt3)}2}),\\
\lambda_3 = \frac {d^3-d+ d\sqrt{6d^2-3-2d^4}}{3d^2-2},\qquad &\text{if }\lambda_1=d\in (\sqrt{\frac{(3{
-}\sqrt3)}2},1),\\
\lambda_3 = \frac {d^3-d - d\sqrt{6d^2-3-2d^4}}{3d^2-2},\qquad &\text{if }\lambda_1=d\in (\sqrt{\frac{(3{
-}\sqrt3)}2},\sqrt\frac23)
,
\end{align*}
$0<\lambda_1\leq 1 \leq \lambda_3.$
\end{remark}

{\begin{remark}
\rm
Lemma \ref{lemma tanti} together with \eqref{mix b 1} imply that, for cubic to monoclinic transformations, there are six different vectors $\vc a^*:=\mu^* \vc a_{\mt U}+(1-\mu^*) \vc a_{\mt U}$ which are as in Definition \ref{Def star} and Proposition \ref{LMhull}. Similarly, there are six different vectors $\vc n^*:=\mu^* \vc n_{\mt U}+(1-\mu^*) \vc n_{\mt U}$ which are as in Definition \ref{Def star I} and Proposition \ref{LMhullI}
\end{remark}
}

\begin{figure}[!htbp]
      \centering%
		{\label{sep dots}\includegraphics[scale=0.8]{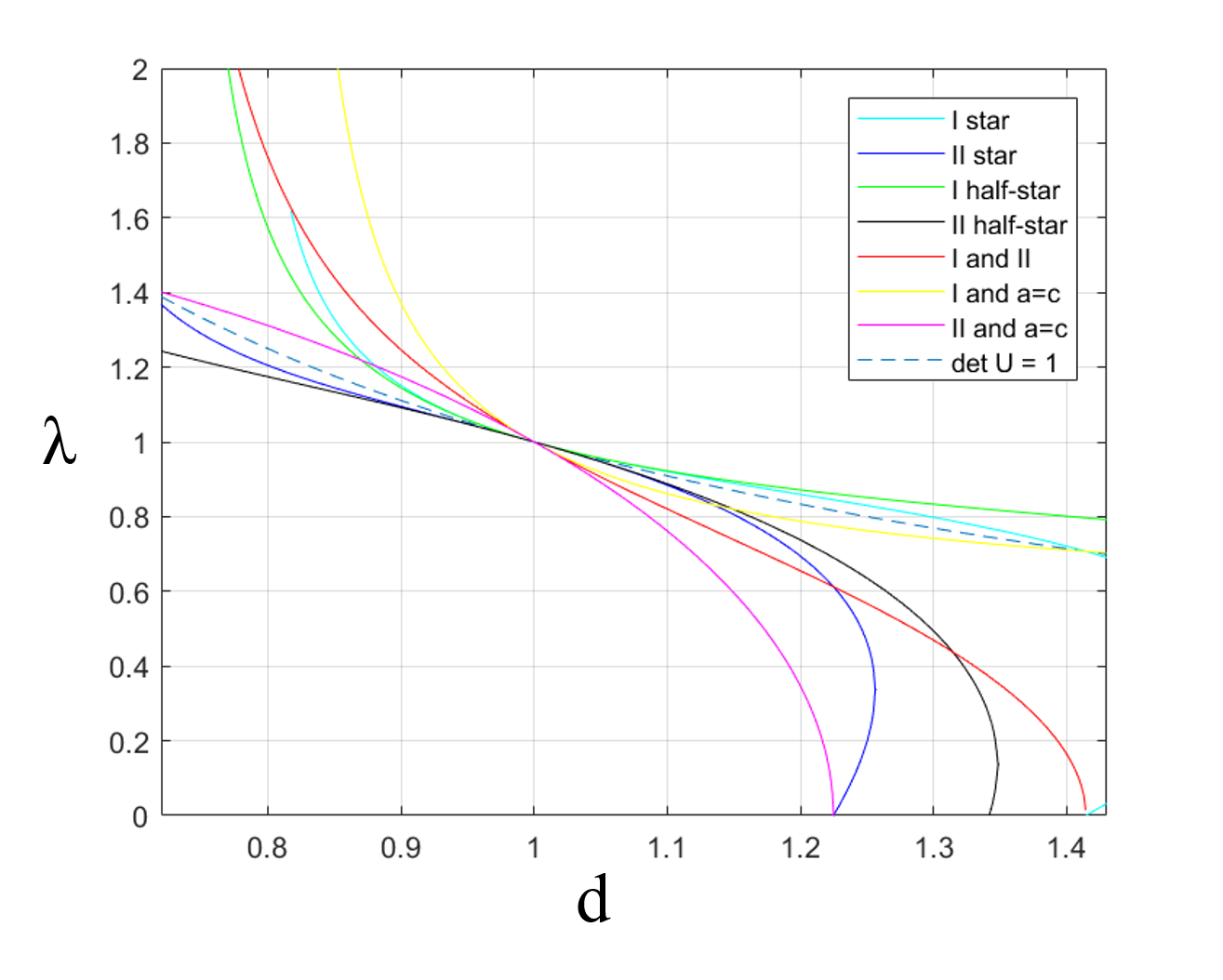}}
\caption{\label{punti irr}
{Let $\mt U_i\in\R^{3\times 3}_{Sym^+}$ be as in \eqref{cubictomono}, with eigenvalues $1,\lambda,d$ (where $d$ is as in \eqref{cubictomono}).
We plot here the relation between the eigenvalues $\lambda$ and $d$ when the $\mt U_i$'s satisfy further compatibility conditions: 
in cyan and in blue the curves of $\lambda$ as a function of $d$ when respectively type I star and type II star twins exist (cf. Remark \ref{quante condizioni} and Remark \ref{quante condizioni 1}). In green and in black the curves of $\lambda$ as a function of $d$ when respectively type I half-star and type II half-star twins exist. In red the curve of $\lambda$ as a function of $d$ when the cofactor conditions are satisfied by both, some type I and some type II twins (see Remark \ref{dueconuno}). In magenta and in yellow we plot the dependence of $\lambda$ in terms of $d$ when $a=c$ (cubic to orthorhombic transformation), and the cofactor conditions are satisfied by type II and type I twin respectively (same as \cite[Fig. 3]{ChlubaJames}). Finally, the dashed line is the curve where $\det \mt U_i =1$ (that is $\lambda=\frac1d$). 
All the above mentioned cases, except for the last one,  are satisfied only by a one-parameter family of matrices as in \eqref{cubictomono}, and we chose as independent parameter $d$. The four deformation parameters determining the transformation strain (cf. $a,b,c,d$ in \eqref{cubictomono}) can be deduced in all the above mentioned cases, except for the last one, in an unique way from $d$ under the physically reasonable assumption that $\lambda>0.5$ (without this further assumption, in some cases we have two possible $\lambda$ given a fixed $d$ (cf. Figure \ref{punti irr}), and therefore two possible ways to determine the transformation strain).
}
}
\end{figure}

\section{How closely does a material satisfy the cofactor conditions?}
\label{metricclose}
In practice the cofactor conditions are never satisfied exactly. For this reason, an interesting question for application is how closely must a material satisfy the cofactor conditions, in order to behave as if these were satisfied exactly. The problem is complex not only from the point of view of establishing a threshold, but especially in terms of choosing the right metric. There are various ways to measure how closely a material satisfies the cofactor conditions, but, up to our knowledge, just two have been used in the literature up till now. Provided (CC3) holds, the first, more intuitive, way is to check (CC1)-(CC2) directly, that is, to see how close the numbers $|\lambda_2-1|, |\vc b\cdot\mt U\cof(\mt U^2-\mt 1)\vc m|$ are to zero. This is, for example, the way the cofactor conditions are measured in \cite{Chluba,ChlubaJames}. However, there is no physical motivation behind these quantities. On the other hand, provided (CC3) holds, a second way to measure how closely the cofactor conditions are satisfied is to measure how small are the quantities $|\lambda_2-1|$, and $\bigl||\mt U\hat{\vc e}|-1\bigr|, \bigl||\mt U^{-1}\hat{\vc e}|-1\bigr|$, respectively in the case of type II and type I twins. Here $\hat{\vc e}$ is as in Proposition \ref{e hat}. This approach has been adopted for example in \cite{JamesNew} and relies on \eqref{ccperI}--\eqref{ccperII}. However, if we apply these two different metrics to Zn\textsubscript{45}Au\textsubscript{30}Cu\textsubscript{25}, we obtain contradictory results. {Indeed, while following the first approach we get that the lowest value of $|\vc b\cdot\mt U\cof(\mt U^2-\mt 1)\vc m|$ among the possible twin systems is approximately equal to $4.1\cdot10^{-5}$ for type I twins, and to $3.8\cdot10^{-5}$ for type II twins, while we have that the second approach leads to lowest values of $\bigl||\mt U^{-1}\hat{\vc e}|-1\bigr|=8.1\cdot10^{-3}$ and $\bigl||\mt U\hat{\vc e}|-1\bigr|=4.2\cdot10^{-4}$.} Therefore, while the first approach entails that the cofactor conditions are closely satisfied by both type I and type II twins, the second states that Zn\textsubscript{45}Au\textsubscript{30}Cu\textsubscript{25} satisfies the cofactor conditions with type II twins  much better than with type I twins. Furthermore, for the completely different alloy Ti\textsubscript{74}Nb\textsubscript{23}Al\textsubscript{3} (see \cite{Inamura}), where the experimentally observed microstructures look very different from those in Zn\textsubscript{45}Au\textsubscript{30}Cu\textsubscript{25}, we have $\lambda_2=0.999996$ and approximate values for $|\vc b\cdot\mt U\cof(\mt U^2-\mt 1)\vc m|$ of $4.4\cdot10^{-5}$ and $3.8\cdot10^{-5}$ respectively for type I and type II twins. These values would lead to the conclusion that Ti\textsubscript{74}Nb\textsubscript{23}Al\textsubscript{3} satisfies the cofactor conditions as closely as Zn\textsubscript{45}Au\textsubscript{30}Cu\textsubscript{25}, even if their behaviours are quite different. However, in Ti\textsubscript{74}Nb\textsubscript{23}Al\textsubscript{3} we get $\bigl||\mt U^{-1}\hat{\vc e}|-1\bigr|=9.9\cdot10^{-3}$ and $\bigl||\mt U\hat{\vc e}|-1\bigr|=8.3\cdot10^{-3}$, which seem to confirm that the cofactor conditions are not so closely satisfied in this material (see Table \ref{Table Fin} below).\\

Following the results of Section \ref{Comp twin sec} (see Remark \ref{Solod1ecof}), it seems reasonable to measure how closely a compound twin satisfies (CC1)--(CC2) by computing the quantity $|d-1|$, and by checking that $d$ is the middle eigenvalue of the $\mt U_i$'s. {For type I and type II, we want to use a different strategy, which is based on critical shear stresses}, and can be related to reversibility. Theorem \ref{TypeI} and Theorem \ref{TypeII} tell us that if a material satisfies the cofactor conditions, then there exist triple junctions. These allow a lot of flexibility in the microstructures, as one can create a laminate with an arbitrary volume fraction, which is compatible with austenite without an interface layer. For these reasons, the presence of triple junctions might play an important role in the reversibility of the transformations. We want thus to measure how close a certain twin is to create triple junctions. {Under some simplifying assumptions, below we quantify the shear stress necessary to deform austenite in such away that triple junctions are possible. If this shear stress is small, then the material behaves as if the cofactor conditions where satisfied; if this shear stress is large, then triple junctions are energetically too expensive to be observed. We then apply the new metric to \Zn\, and to Ti\textsubscript{74}Nb\textsubscript{23}Al\textsubscript{3}. We deduce that, in the former material, the cofactor conditions are better satisfied by type II twins than by type I twins; in the latter material our metric seems to confirm that triple junctions are energetically not convenient without incurring in plastic effects.} {Finally, in Table \ref{Table Fin} we show the results obtained with the different metrics, while in Figure \ref{Schema Fin} we provide an easy algorithm to use our metric. It is worth noticing that, with our metric, it is easier to compare how closely two different materials satisfy the cofactor conditions. Indeed, our algorithm produces one number for each twin, and not three (one for (CC1), one for (CC2) and one for (CC3)) as the metrics currently in use.}
\\

Below we consider $\mt U,\mt V\in \R^{3\times 3}_{Sym^+}$ satisfying \eqref{e hat}, $\mt U \neq\mt V$, and assume that $(\mt U,\vc b_I,\vc m_I)$, $(\mt U,\vc b_{II},\vc m_{II})$ are respectively the type I and type II twins generated by $\mt U,\mt V$. We now want to look for $\epsilon \in\R^{3\times 3},$ such that
\beq
\label{supercomp3}
\tilde{\mt R}_a \mt U = \mt 1 + \epsilon + \vc c_a\otimes\vc o_a,\qquad
\tilde{\mt R}_b \mt V = \mt 1 + \epsilon + \vc c_b\otimes\vc o_b,\qquad
\tilde{\mt R}_b \mt V - \tilde{\mt R}_a \mt U  = \tilde{\mt R}_a \vc b\otimes\vc m,
\eeq
for some $\tilde{\mt R}_a, \tilde{\mt R}_b\in SO(3),$ $\vc c_a,\vc c_b\in\R^3,$ $\vc o_a,\vc o_b\in \mathbb S^2$ and where either $\vc b  = \vc b_{I}$ 
and $\vc m=\vc m_{I}$ or $\vc b =\vc b_{II}$ 
and $\vc m = \vc m_{II}.$ Physically, we want to find an elastic deformation of austenite, which allows for triple junctions, and hence is compatible with a laminate without an interface layer (see Figure \ref{tripleJ}). This is an approximation of what happens in reality, where also martensite is deformed, {and where $\epsilon$ is not constant in the austenite, but might decrease away from the triple junctions in type I twins, and away from the habit plane in type II twins. Nonetheless, the above assumptions allow us both to stick to simple computations, and to get qualitative results for the regions {of higher stress.}}
\begin{figure}
\centering
\begin{tikzpicture}[scale=0.8]
\draw[thin] (-1,4) -- (-4,1);
\draw[thin] (-1,4) -- (-1,1);
\draw[thin] (-1,4) -- (-5,3);
\draw[thin] (-7,1) -- (-5,3);
\draw[thin] (-5,3) -- (-5,6);
\draw[thin] (-8,5.25) -- (-5,6);
\draw[thin] (-8,3) -- (-5,6);

\draw [->,thick] (-6.5,6-0.375) -- (-6.5-0.4472/2,6-0.375+0.4472*4/2); 
\filldraw [red] (-6.5-0.4472/2,6-0.375+0.4472*4/2) circle (0pt) node[anchor=west,black] {$\vc o_a$};

\draw [->,thick] (-5,4.5) -- (-5+0.922,4.5); 
\filldraw [red] (-5+0.922,4.5) circle (0pt) node[anchor=west,black] {$\vc o_b$};

\draw [->,thick] (-6,2) -- (-6-0.922/1.4142,2 +0.922/1.4142 ); 
\filldraw [red] (-6-0.922/1.4142,2 +0.922/1.4142) circle (0pt) node[anchor=west,black] {$\vc m$};

\filldraw [red] (-7.5,4.75) circle (0pt) node[anchor=north,black] {$\tilde{\mt R}_a\mt U$};
\filldraw [red] (-6.,3.85) circle (0pt) node[anchor=north,black] {$\tilde{\mt R}_b\mt V$};
\filldraw [red] (-4.,2.65) circle (0pt) node[anchor=north,black] {$\tilde{\mt R}_a\mt U$};
\filldraw [red] (-2.5,1.75) circle (0pt) node[anchor=north,black] {$\tilde{\mt R}_b\mt V$};

\draw[thin] (2,5) -- (8,7);
\draw[thin] (2+1.4/3,5-1.4) -- (8+1.4/3,7-1.4);
\draw[thin] (2+2.5/3,5-2.5) -- (8+2.5/3,7-2.5);
\draw[thin] (2+3.9/3,5-3.9) -- (8+3.9/3,7-3.9);

\filldraw [red] (5-0.5/3,6-0.5) circle (0pt) node[anchor=north,black] {$\tilde{\mt R}_a\mt U$};
\filldraw [red] (5+0.7/3,6-1.7) circle (0pt) node[anchor=north,black] {$\tilde{\mt R}_b\mt V$};
\filldraw [red] (5+2/3,6-3) circle (0pt) node[anchor=north,black] {$\tilde{\mt R}_a\mt U$};
\filldraw [red] (-2.5,1.75) circle (0pt) node[anchor=north,black] {$\tilde{\mt R}_b\mt V$};

\draw [->,thick] (6,5+4/3) -- (6-0.922/3.1623,5+4/3+3*0.922/3.1623); 
\filldraw [red] (6-0.922/3.1623,5+4/3+3*0.922/3.1623) circle (0pt) node[anchor=west,black] {$\vc m$};

\filldraw [red] (-3,5.2) circle (0pt) node[anchor=south,black] {$\mt 1+\epsilon$};
\filldraw [red] (4,6.2) circle (0pt) node[anchor=south,black] {$\mt 1+\epsilon$};

\filldraw [red] (-4,7.) circle (0pt) node[anchor=south east,black] {Austenite};
\filldraw [red] (5.5,8.) circle (0pt) node[anchor=south east,black] {Austenite};
\filldraw [red] (-9,2.) circle (0pt) node[anchor=south east,black,rotate=270] {Martensite};
\filldraw [red] (9.0,3.5) circle (0pt) node[anchor=south east,black,rotate=270] {Martensite};
\end{tikzpicture}
\caption{\label{tripleJ} Example of triple junctions for type I and type II twins with perturbed austenite.}
\end{figure}
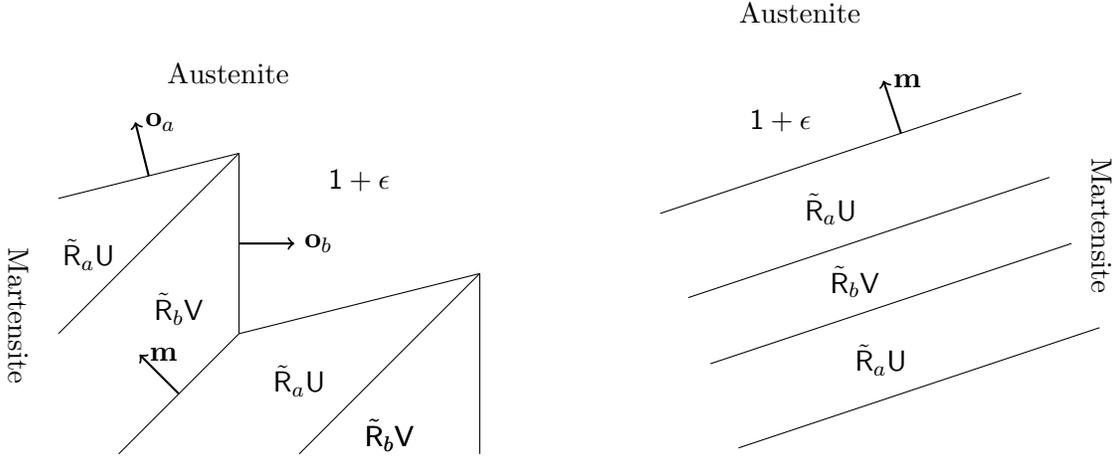
Since we are interested in the case where we are close to satisfying the cofactor conditions, by Theorem \ref{TypeI} and Theorem \ref{TypeII} we expect this deformation to be small, and we hence stick to the context of linear elasticity, where we can express the deformation gradient of the elastically perturbed austenite as $\mt 1+\epsilon.$ Furthermore, as $\epsilon$ is expected to be small, we adopt the following approximation
\beq
\label{approx def lin}
\epsilon+\epsilon^T \approx \epsilon+\epsilon^T + \epsilon^T\epsilon =  (\tilde{\mt R}_a \mt U - \vc c_a\otimes\vc o_a)^T(\tilde{\mt R}_a \mt U - \vc c_a\otimes\vc o_a) - \mt 1. 
\eeq
We remark that the conditions in \eqref{supercomp3} imply 
$$
\vc c_b\otimes\vc o_b-\vc c_a\otimes\vc o_a = \tilde{\mt R}_a\vc b\otimes\vc m
$$
that is, either $\vc o_a = \pm\vc o_b = \pm\vc m,$ or $\vc c_a \parallel \vc c_b \parallel \tilde{\mt R}_a\vc b.$
Therefore, the only possibilities are:
\begin{align}
\vc c_a\otimes\vc o_a = -\tilde{\mt R}_a\vc b\otimes \vc o,\\
\vc c_a\otimes\vc o_a = -\tilde{\mt R}_a\vc c\otimes \vc m,
\end{align}
for some $\vc c, \vc o\in\R^3.$\\

We now want to minimise the energy of $\epsilon+\epsilon^T$ with respect to $\vc c,\vc o$ depending on the case. For simplicity, following the approach in \cite{JamesMuller}, we consider just the shear component of the energy, that is we consider the energy of $\epsilon + \epsilon^T$ to be given by
\beq
\label{energia linear}
{\frac{G}{2}\Bigl\|	\frac{\epsilon+\epsilon^T}2\Bigr\|^2,}
\eeq
where $G$ is the shear modulus for the material. {This energy is not keeping in account the anisotropies of austenite around the transformation temperature, but provides a good lower bound to the anisotropic energy describing the system.} The following two lemmas are useful in what follows, as they allow to find the energy minimisers, under the assumption that our elastic deformation is small enough. 
\begin{lemma}
\label{ILII}
Let $\mt U,\mt V\in\R^{3\times3}_{Sym^+}$ satisfy \eqref{e hat}, and let $(\mt U,\vc b,\vc m)$ be a type I or a type II solution of the twinning equation \eqref{compatib condit} for $\mt U,\mt V$. Suppose further that the minimum and the maximum eigenvalues of $\mt U$, denoted respectively by $\lambda_1$ and $\lambda_3$, satisfy $2\lambda_1^2-\lambda_3^2>0$, and define $$\mt C(\vc c) = (\mt U + \vc c\otimes\vc m)^T(\mt U + \vc c\otimes\vc m)-\mt1.$$ 
Assume also 
$$
\min_{\vc c\in\R^3} \|\mt C(\vc c)\|^2 < \min\{(2\lambda_1^2-\lambda_3^2)^2,1\}.
$$
Then,
$$
\|\mt C(\hat{\vc c})\|^2 = \min_{\vc c\in\R^3} \|\mt C(\vc c)\|^2,
$$ 
if and only if $\hat{\vc c}=\hat{\vc c}^\pm =\pm \frac{\mt U^{-1}\vc m}{|\mt U^{-1}\vc m|} - \mt U\vc m$. The eigenvalues of $\mt C_*:=\mt C(\hat{\vc c}^+)=\mt C(\hat{\vc c}^-)$ are 
$$
\frac12\Bigl(\tr (\mt C_*)-\sqrt{2\tr(\mt C^2_*)-(\tr \mt C_*)^2} \Bigr),\quad 0,\quad \frac12\Bigl(\tr (\mt C_*)+\sqrt{2\tr(\mt C^2_*)-(\tr \mt C_*)^2} \Bigr).
$$
Finally, $\|\mt C_*\|^2 = 0$ if and only if $(\mt U,\vc b,\vc m)$ is a type II twin satisfying the cofactor conditions.
\end{lemma}
\begin{proof}
We first notice that $\|\mt C(\vc c)\|^2 $ is a smooth function of $\vc c$, so its minimizer $\vc c$ must satisfy the equation
\beq
\label{dazero}
\bigl(\mt U + \vc c\otimes \vc m\bigr) \bigl(\mt U^2-  \mt 1 + |\vc c|^2\vc m\otimes\vc m + \mt U\vc c\odot\vc m\bigr)\vc m = \vc 0,
\eeq
where we denoted by $\odot$ the symmetrised tensor product $\vc u\odot \vc w:=\vc u\otimes \vc w+\vc w\otimes \vc u $. We introduce a change of variable $\vc v = \vc c + \mt U\vc m$, so that, after rearranging the terms, \eqref{dazero} becomes
\beq
\label{eigenvectorSD}
\bigl(\mt U^2 -  \mt 1  - \mt U\vc m\otimes\mt U\vc m\bigr)\vc v = -|\vc v|^2\vc v.
\eeq
Therefore, $\vc v$ is either zero or is an eigenvector of $$\mt D := \bigl(\mt U^2 -  \mt 1  - \mt U\vc m\otimes\mt U\vc m\bigr),$$ and the related eigenvalue must me be equal to $-|\vc v|^2$. We first notice that $-1$ is an eigenvalue for $\mt D$ related to the eigenvector $\vc w_1 := \frac{\mt U^{-1}\vc m}{|\mt U^{-1}\vc m|}.$ Let $\vc w_2,\vc w_3$ be the other two eigenvectors of $\mt D$. We have
$$
\mt D = -\vc w_1\otimes\vc w_1 +\bar{\lambda}_2\vc w_2\otimes\vc w_2+\bar{\lambda}_3\vc w_3\otimes\vc w_3,
$$ 
for some $\bar{\lambda}_2,\bar{\lambda}_3\in\R, \bar{\lambda}_2\leq\bar{\lambda}_3.$ Let now $\alpha_1,\alpha_2,\alpha_3\in[0,1]$ be such that
$$
\frac{\mt U\vc m}{|\mt U\vc m|} = \alpha_1\vc w_1+\alpha_2\vc w_2 +\alpha_3\vc w_3.
$$
The following identity must hold 
\beq
\label{stimalpha1}
\alpha^2_1+\alpha^2_2+\alpha^2_3=1.
\eeq
Furthermore,
\beq
\label{stimalpha2}
\alpha_1 = \biggl|\vc w_1\cdot \frac{\mt U\vc m}{|\mt U\vc m|}\biggr| = \frac{1}{|\mt U\vc m||\mt U^{-1}\vc m|} \geq \frac{\lambda_1}{\lambda_3},
\eeq
where $\lambda_1,\lambda_3$ are respectively the minimum and the maximum eigenvalue of $\mt U$. Thus,
\beq
\label{passo1}
(\lambda_3^2-1) \geq \max_{\vc w\in\mathbb{S}^2, \vc w\bot\vc w_1} \mt D\vc w\cdot\vc w = \bar{\lambda}_3 \geq \bar{\lambda}_2 = \min_{\vc w\in\mathbb{S}^2, \vc w\bot\vc w_1} \mt D\vc w\cdot\vc w \geq (\lambda_1^2-1) - \max_{\vc w\in\mathbb{S}^2, \vc w\bot\vc w_1}|\mt U\vc m\cdot \vc w|^2. 
\eeq
The last term on the right hand side can be estimated by,
\beq
\label{passo1emezz}
- |\mt U\vc m\cdot \vc w|^2\geq- |\mt U\vc m|\bigl((\alpha_2\vc w_2+\alpha_3\vc w_3)\cdot\vc w\bigr)^2\geq -\lambda_3^2(\alpha_2^2+\alpha_3^3)\geq \lambda_1^2-\lambda_3^2,
\eeq
where we made use of \eqref{stimalpha1}--\eqref{stimalpha2} in the last inequality. Thus, collecting the inequalities in \eqref{passo1}--\eqref{passo1emezz}, we prove
\beq
\label{passo2}
\bar{\lambda}_2 \geq 2\lambda_1^2-1-\lambda_3^2.
\eeq
In this way, we have
$$
\|\mt C(\vc v -\mt U\vc m)\|^2\geq \max_{\vc w \in \mathbb{S}^2}\big| \mt C(\vc v -\mt U\vc m)\vc w\cdot\vc w \big|^2 \geq \big|\mt C(\vc v -\mt U\vc m)\vc m\cdot\vc m \big|^2 = (1-|\vc v|^2)^2. 
$$
Now, \eqref{eigenvectorSD} implies that, if $\vc v\neq \vc w_1,$ $|\vc v|^2$ is equal to $-\bar{\lambda}_2,-\bar{\lambda}_3,0$. Hence, by \eqref{passo1}--\eqref{passo2}, together with the assumption $2\lambda_1^2-\lambda_3^2>0$, we can conclude
$$
\|\mt C(\vc v -\mt U\vc m)\|^2\geq \min\{(2\lambda_1^2-\lambda_3^2)^2,1\}. 
$$
We have hence proved that, if $\min_{\vc c\in\R^3}\|\mt C(\vc c)\|^2 < \min\{(2\lambda_1^2-\lambda_3^2)^2,1\}$, the only minimizers are $\hat{\vc c}^\pm = \pm \vc w_1-\mt U\vc m.$ In this case, we obtain
$$
\mt C_* : = \mt C(\pm\vc w_1 -\mt U\vc m) = \mt U^2-\mt 1 + (1+|\mt U\vc m|^2)\vc m\otimes\vc m-\mt U^2\vc m\odot\vc m.
$$
Clearly, $\vc m$ is an eigenvector of $\mt C_*$ related to the null eigenvalue. Thus, the two remaining eigenvalues are given by the solutions $\rho$ of
$$
-\rho^2 + \tr(\mt C_*)\rho - \frac12\bigl((\tr(\mt C_*))^2-\tr(\mt C_*^2)\bigl) = 0.
$$
We now claim that, if $(\mt U,\vc b,\vc m)$ satisfies the cofactor conditions as a type II twin, then $\mt C_*=\mt 0$. Indeed, by \eqref{ccperII} we have that $\vc m\cdot \vc v_2 = 0$, where $\vc v_2$ is such that $\mt U\vc v_2 = \vc v_2$. Then, clearly $\vc w_1\cdot \vc v_2 = 0,$ and $\vc v_2$ is an eigenvector for $\mt C_*$ related to a null eigenvalue. Therefore, as $\mt C_*$ is symmetric, we just need to show that $\tr\mt C_*= \tr \mt U^2 -2 - |\mt U\vc m|^2 = \mt 0$. But if the cofactor conditions hold as a type II twin, then, thanks to Theorem \ref{TypeII} and \cite[Prop. 4]{BallJames1} we deduce that
$$
|\mt U\vc m|^2 =  (\lambda_1^2+\lambda_3^2-1),
$$  
that is $\tr \mt C_* = 0$, and thus $\mt C_*=\mt 0$. Conversely, let $\mt C_* = \mt 0$, and define $\tilde{\mt R}_a\in \R^{3\times 3}$ by
$$
\tilde{\mt R}_a^T = \mt U + \hat{\vc c}^\pm \otimes \vc m.
$$
First, by the fact that $\mt C_* =\mt 0$, we have $\tilde{\mt R}_a\tilde{\mt R}_a^T=\mt 1$, thus $\tilde{\mt R}_a^T$ is an orthogonal matrix, and $\det \tilde{\mt R}_a^T =\pm 1$. Furthermore, 
$$
\det \tilde{\mt R}_a = \det\mt U (1+ \mt U^{-1}\hat{\vc c}^\pm\cdot \vc m) = \pm |\mt U^{-1}\vc m|\det\mt U .
$$
Therefore, if we choose $\vc c_a = \hat{\vc c}^+$, $\tilde{\mt R}_a \in SO(3)$, we have
$$
\tilde{\mt R}_a \mt U = \mt 1 - \tilde{\mt R}_a{\vc c}_a\otimes\vc m,\qquad
\tilde{\mt R}_b \mt V - \tilde{\mt R}_a \mt U  = \tilde{\mt R}_a \vc b\otimes\vc m,\qquad
\tilde{\mt R}_b \mt V = \mt 1 - \tilde{\mt R}_a({\vc c}_a +\vc b)\otimes\vc m,
$$
for some $\tilde{\mt R}_b\in SO(3)$. We can hence construct laminates
$$
\tilde{\mt R}_a (\mt U +\mu \vc b\otimes \vc m) = \mu \tilde{\mt R}_v \mt V + (1-\mu)\tilde{\mt R}_a \mt U = \mt 1 + \tilde{\mt R}_a({\vc c}_a +\mu\vc b)\otimes\vc m,
$$ 
which are rank-one connected to $\mt 1$ for arbitrary volume fractions $\mu\in[0,1]$. As a consequence of Theorem \ref{thm cof cond} $(\mt U,\vc b,\vc m)$ satisfies the cofactor conditions. Furthermore, as $\mt C_*=\mt 0,$ $\vc v_2 := \vc m\times \mt U^2\vc m$ satisfies $\mt U\vc v_2 =\vc v_2$, and is thus the eigenvector corresponding to the eigenvalue of $\mt U$ equal to one. Therefore, \eqref{ccperI} finally implies that $(\mt U,\vc b,\vc m)$ cannot be a type I twin, and, by assumption, must hence be a type II twin.
\end{proof}

In a similar way, we can prove the following result
\begin{lemma}
\label{ILI}
Let $\mt U,\mt V\in\R^{3\times3}_{Sym^+}$ satisfy \eqref{e hat}, and let $(\mt U,\vc b,\vc m)$ be a type I or a type II solution of the twinning equation \eqref{compatib condit} for $\mt U,\mt V$. Define $$\mt E(\vc o) = (\mt U + \vc b\otimes\vc o)^T(\mt U + \vc b\otimes\vc o)-\mt1.$$ 
Assume also 
$$
\min_{\vc o\in\R^3} \|\mt E(\vc o)\|^2 < 1. 
$$
Then,
$$
\|\mt E(\hat{\vc o})\|^2 = \min_{\vc o\in\R^3} \|\mt E(\vc o)\|^2,
$$ 
if and only if $\hat{\vc o}=\hat{\vc o}^\pm =\pm \frac{\mt U^{-1}\vc b}{|\mt U^{-1}\vc b| |\vc b|} - \frac{\mt U\vc b}{|\vc b|^2}$. The eigenvalues of $\mt E_*=\mt E(\hat{\vc o}^+)=\mt E(\hat{\vc o}^-)$ are 
$$
\frac12\Bigl(\tr (\mt E_*)-\sqrt{2\tr(\mt E^2_*)-(\tr \mt E_*)^2} \Bigr),\quad 0,\quad \frac12\Bigl(\tr(\mt E_*)+\sqrt{2\tr(\mt E^2_*)-(\tr \mt E_*)^2} \Bigr).
$$
Finally, $\|\mt E(\hat{\vc o})\|^2 = 0$ if and only if $(\mt U,\vc b,\vc m)$ is a type I twin satisfying the cofactor conditions.
\end{lemma}
\begin{proof}
Again, $\|\mt E(\vc o)\|^2 $ is a smooth function of $\vc o$, so its minimizer $\vc o$ must be a stationary point, that is
\[
\bigl(\mt U^2-  \mt 1 + |\vc b|^2\vc o\otimes\vc o + \mt U\vc b\odot\vc o\bigr)\bigl(\mt U\vc b + |\vc b|^2  \vc o\bigr)  = \vc 0.
\]
The change of variable $\vc v = |\vc b|^2 \vc o + \mt U\vc b$ thus entails
$$
\biggl(\mt U^2 -  \mt 1  -\frac1{|\vc b|^2}\mt U\vc b\otimes \mt U\vc b \biggr)\vc v = -\frac{|\vc v|^2}{|\vc b|^2}\vc v.
$$
As in the case of Lemma \ref{ILII}, $\vc v$ is either zero or is an eigenvector of $$\mt F := \mt U^2 -  \mt 1  -\frac1{|\vc b|^2}\mt U\vc b\otimes \mt U\vc b ,$$ and the related eigenvalue must me be equal to $-\frac{|\vc v|^2}{|\vc b|^2}$. As in the previous case, $-1$ is an eigenvalue for $\mt F$ related to the eigenvector $\vc w_1 := \frac{\mt U^{-1}\vc b}{|\mt U^{-1}\vc b|}.$ If we choose $\vc v\neq \pm\vc  w_1,$ then it must hold $\vc v\cdot \vc w_1 = 0,$ and therefore
$$
\mt E\Bigl(	\pm 
\frac{\vc v}{|\vc b|^2}-\frac {\mt U\vc b}{|\vc b|^2}\Bigr) = \mt U^2 -  \mt 1  -\frac1{|\vc b|^2}\mt U\vc b\otimes \mt U\vc b +\frac{1}{|\vc b|^2}\vc v\otimes\vc v,
$$
has $-1$ as eigenvalue related to the eigenvector $\vc w_1$. In this case, hence,
$$
\Bigr \| \mt E\Bigl(	\pm \frac{\vc v}{|\vc b|^2}-\frac {\mt U\vc b}{|\vc b|^2}\Bigr) \Bigl\|^2 \geq \max_{\vc w\in\mathbb{S}^2}\Bigr | \mt E\Bigl(	\pm \frac{\vc v}{|\vc b|^2}-\frac {\mt U\vc b}{|\vc b|^2}\Bigr)\vc w\cdot\vc w \Bigl|^2\geq \Bigr | \mt E\Bigl(	\pm \frac{\vc v}{|\vc b|^2}-\frac {\mt U\vc b}{|\vc b|^2}\Bigr)\vc w_1\cdot\vc w_1 \Bigl|^2 \geq 1. 
$$
Therefore, $\|\mt E(\vc n)\| <1$ if and only if $\vc v=\pm|\vc b|^{-1}\vc w_1$. On the other hand,
$$
 \mt E_* : = \mt E\Bigl(	\pm 
\frac{\vc w_1}{|\vc b|}-\frac {\mt U\vc b}{|\vc b|^2}\Bigr) = \mt U^2 -  \mt 1  -\frac1{|\vc b|^2}\mt U\vc b\otimes \mt U\vc b + \vc w_1\otimes\vc w_1,
$$
has always $0$ as an eigenvalue related to the eigenvector $\vc w_1$. Therefore the other two eigenvalues are given by the solutions $\rho$ of
$$
-\rho^2 + \tr(\mt E_*)\rho - \frac12\bigl((\tr(\mt E_*))^2-\tr(\mt E_*^2)\bigl) = 0.
$$
We now want to prove that, if $(\mt U,\vc b,\vc m)$ satisfies the cofactor conditions as a type I twin, then $\mt E_*=0.$ By \eqref{ccperI} we have that $\vc b\cdot \vc v_2 = 0$, where $\vc v_2$ is such that $\mt U\vc v_2 = \vc v_2$. So that $\vc v_2$ is also an eigenvector of $\mt E_*$, and $\mt E_*\vc v_2 =\vc 0$, as much as $\vc w_1\cdot\vc v_2 = 0.$ Therefore, we just need to check that $\tr \mt E_* = 0,$ that is, $|\vc b|^{-2}|\mt U\vc b|^2 = (\lambda_1^2+\lambda_3^2-1)$. But again, 
by Theorem \ref{TypeI}, we know that $|\vc b|^{-1}\mt R_{\mt U}\vc b = |\vc a|^{-1}\vc a$, and 
$$
|\vc b|^{-1}\mt U\vc b = |\vc b|^{-1}\mt U \mt R_{\mt U}^T\mt R_{\mt U}\vc b = |\vc a|^{-1}(\mt 1+\vc n_{\mt U}\otimes\vc a)\vc a .
$$
Therefore, exploiting \cite[Prop. 4]{BallJames1} we thus deduce
$$
\frac{|\mt U\vc b|^2}{|\vc b|^2}= \bigg\|\frac{\vc a}{\|\vc a\|} + \|\vc a\|\vc n_{\mt U}\bigg\|^2 = (\lambda_1^2+\lambda_3^2-1),
$$  
as required. Conversely, let $\mt E_* = \mt 0$, and define $\tilde{\mt R}_a\in \R^{3\times 3}$ by 
$$
\tilde{\mt R}_a^T = \mt U + \vc b\otimes \hat{\vc o}^\pm .
$$
As $\mt E_* =\mt 0$, we have $\tilde{\mt R}_a\tilde{\mt R}_a^T=\mt 1$, thus $\tilde{\mt R}_a^T$ is an orthogonal matrix, and $\det \tilde{\mt R}_a^T =\pm 1$. Furthermore, 
$$
\det \tilde{\mt R}_a = \det\mt U (1+ \mt U^{-1}\vc b\cdot \hat{\vc o}^\pm) = \pm |\mt U^{-1}\vc b|\det\mt U.
$$
Therefore, if we choose $\vc o_a = \hat{\vc o}^+$, $\tilde{\mt R}_a \in SO(3)$, we have
$$
\tilde{\mt R}_a \mt U = \mt 1 - \tilde{\mt R}_a{\vc b}\otimes\vc o_a,\qquad
\tilde{\mt R}_b \mt V - \tilde{\mt R}_a \mt U  = \tilde{\mt R}_a \vc b\otimes\vc m,\qquad
\tilde{\mt R}_b \mt V = \mt 1 - \tilde{\mt R}_a\vc b\otimes(\vc o_a+\vc m),
$$
for some $\tilde{\mt R}_b\in SO(3)$. We can hence construct laminates compatible with $\mt 1$ for arbitrary volume fractions, and hence, by Theorem \ref{thm cof cond} the cofactor conditions must be satisfied. Furthermore, as $\mt E_*=\mt 0,$ $\vc v_2 := (\mt U\vc b)\times (\mt U^{-1}\vc b)$ satisfies $\mt U\vc v_2 =\vc v_2$, and hence $\vc v_2$ is the eigenvector related to the eigenvalue of $\mt U$ equal to one. Therefore, $\vc v_2\cdot \vc b = \vc v_2 \cdot \mt U\vc b =0$, and \eqref{ccperII} finally implies that $(\mt U,\vc b,\vc m)$ cannot be a type II twin. As a consequence, by hypotheses, $(\mt U,\vc b,\vc m)$ must be a type I twin.
\end{proof}
\begin{remark}
\rm
{Lemma \ref{ILII} and Lemma \ref{ILI} can be also used to measure how far a compound twin is to form triple junctions. Indeed, triple junctions can arise if and only if either $\mt C_*=\mt 0$ or $\mt E_*=\mt 0$. After some computations we obtain that, for a twin generated by $\mt U_1,\mt U_2\in\R^{3\times3}$, with $\mt U_1,\mt U_2$ as in \eqref{cubictomono} (and hence all symmetry related twins in Table \ref{Table 1}) 
\beq
\label{Nolayerabove1}
\begin{split}
\mt C_* &= \mt 0 \quad\Longleftrightarrow\quad
d=1 \text{ and either } a^2+b^2 = \det \mt U^2 \text{ or } c^2+b^2 = \det \mt U^2,\\
\mt E_* &= \mt 0\quad\Longleftrightarrow\quad
 d=1 \text{ and either } a^2+b^2 = 1 \text{ or } c^2+b^2 = 1.
\end{split}
\eeq
For a twin generated by $\mt U_1,\mt U_3\in\R^{3\times3}$, with $\mt U_1,\mt U_3$ as in \eqref{cubictomono} (and hence all symmetry related twins in Table \ref{Table 1}) 
\beq
\label{Nolayerabove2}
\begin{split}
\mt C_* &= \mt 0\quad\Longleftrightarrow\quad
 d=1 \text{ and either } (a+b)^2+(b+c)^2 = 2\det \mt U^2 \text{ or } (a-b)^2+(b-c)^2 = 2\det \mt U^2,\\
\mt E_* &= \mt 0\quad\Longleftrightarrow\quad
 d=1 \text{ and either } (a+b)^2+(b+c)^2 = 2 \text{ or } (a-b)^2+(b-c)^2 = 2.
\end{split}
\eeq
It would be interesting to construct a new alloy satisfying the cofactor condition with a compoud twin ($d=1$ in \eqref{cubictomono} or \eqref{cubictoortho} for cubic to monoclinic and cubic to orthorhombic transformations), but not satisfying any of the conditions in \eqref{Nolayerabove1}--\eqref{Nolayerabove2} above (or their equivalent if the transformation is not from cubic to monoclinic or from cubic to orthorhombic), allowing the twin to form triple junctions. This would help to better understand the influence of the lack of transition layer between phases on the great reversibility of the transformation observed in materials satisfying the cofactor conditions. 
}
\end{remark}

By putting together \eqref{approx def lin} and \eqref{energia linear}, together with Lemma \ref{ILII} and Lemma \ref{ILI} we have that, if we are close to satisfy the cofactor conditions, the stress induced by $\epsilon$ is given by
{\begin{align*}
\sigma(\epsilon) = {\frac G2}\mt E_* 
, &\qquad\text{for type I},\\
\sigma(\epsilon) = {\frac G2}\mt C_* 
, &\qquad\text{for type II},
\end{align*}
and where
$$
\mt E_* : =  \mt U^2 -  \mt 1  -\frac1{|\vc b|^2}\mt U\vc b\otimes \mt U\vc b + \frac{\mt U^{-1}\vc b}{|\mt U^{-1}\vc b|}\otimes\frac{\mt U^{-1}\vc b}{|\mt U^{-1}\vc b|}, \qquad \mt C_* : =  \mt U^2 -  \mt 1  + (1+|\mt U\vc m|^2) \vc m\otimes \vc m - \mt U^2\vc m\odot\vc m .
$$
If we think now at reversibility as the lack of plastic effects during the phase transition, we can say that we are close to satisfy the cofactor conditions if our principal stresses satisfy some yield criterion. Adopting for simplicity Tresca's yield criterion, we can say that we are closely satisfying the cofactor conditions if
\begin{align*}
{\frac G4} \max\{|\lambda_1^{\mt E_*}-\lambda_2^{\mt E_*}|,|\lambda_1^{\mt E_*}-\lambda_3^{\mt E_*}|,|\lambda_2^{\mt E_*}-\lambda_3^{\mt E_*}|\} \leq \sigma_C, &\qquad\text{for type I},\\
\frac G4 \max\{|\lambda_1^{\mt C_*}-\lambda_2^{\mt C_*}|,|\lambda_1^{\mt C_*}-\lambda_3^{\mt C_*}|,|\lambda_2^{\mt C_*}-\lambda_3^{\mt C_*}|\} \leq \sigma_C, &\qquad\text{for type II},
\end{align*}
where $\sigma_C$ is the shear yield stress and $\lambda_i^{\mt E_*},\lambda_i^{\mt C_*}$, $i=1,2,3$ are respectively the eigenvalues of $\mt E_*$ and $\mt C_*$. We recall that, as proved in Lemma \ref{ILI} and in Lemma \ref{ILII}, at least one of the $\lambda_i^{\mt E_*}$ and one of the $\lambda_i^{\mt C_*}$ should be equal to zero. In the lack of experimental values for $G$ and $\sigma_C$, we can measure how closely the cofactor conditions are satisfied in a certain alloy by computing
\begin{align}
\label{NewMetric1}
\max\{|\lambda_1^{\mt E_*}-\lambda_2^{\mt E_*}|,|\lambda_1^{\mt E_*}-\lambda_3^{\mt E_*}|,|\lambda_2^{\mt E_*}-\lambda_3^{\mt E_*}|\}, &\qquad\text{for type I twins},\\
\label{NewMetric2}
\max\{|\lambda_1^{\mt C_*}-\lambda_2^{\mt C_*}|,|\lambda_1^{\mt C_*}-\lambda_3^{\mt C_*}|,|\lambda_2^{\mt C_*}-\lambda_3^{\mt C_*}|\}, &\qquad\text{for type II twins}. 
\end{align}
We report in Table \ref{Table Fin} the results obtained by computing \eqref{NewMetric1}--\eqref{NewMetric2} for Zn\textsubscript{45}Au\textsubscript{30}Cu\textsubscript{25} and Ti\textsubscript{74}Nb\textsubscript{23}Al\textsubscript{3}. Values for Zn\textsubscript{45}Au\textsubscript{30}Cu\textsubscript{25} are $0.0173$ and $0.0021$ respectively for type I and type II twins. These values seem to confirm that the satisfaction of the cofactor conditions in this material is much closer with some type II twins than with type I twins, in agreement with Proposition \ref{nonsipuo}. Values for Ti\textsubscript{74}Nb\textsubscript{23}Al\textsubscript{3} are  $0.0267$ and $0.0225$ respectively for type I and type II twins, confirming that triple junctions require high elastic energy in this material. 

In Figure \ref{Schema Fin} we summarise our algorithm to verify how closely the cofactor conditions are satisfied.
\begin {table}[h]
\begin{center}
\begin{tabular}{ |l|l|l|l|l| }
\hline
Twin\& Metric / Material & \Zn  & Ti\textsubscript{74}Nb\textsubscript{23}Al\textsubscript{3} \\ \hline
$\big|\lambda_2-1\big|$& $6.1\times 10^{-4}$ & $3.7\times 10^{-6}$
\\
$|\vc b\cdot\mt U\cof(\mt U^2-\mt 1)\vc m|$ for type I twins & $4.1\times 10^{-5}$ & $4.4\times 10^{-5}$\\
$|\vc b\cdot\mt U\cof(\mt U^2-\mt 1)\vc m|$ for type II twins & $3.8\times 10^{-5}$ & $3.8\times 10^{-5}$\\
$\big||\mt U^{-1}\hat{\vc e}|-1\big|$ for type I twins & $8.1\times 10^{-3}$ & $9.9\times 10^{-3}$\\
$\big||\mt U\hat{\vc e}|-1\big|$ for type II twins & $4.2\times 10^{-4}$ & $8.3\times 10^{-3}$\\
\eqref{NewMetric1} for type I twins & $1.7\times 10^{-2}$ & $2.7\times 10^{-2}$\\
\eqref{NewMetric2} for type II twins & $2.1\times 10^{-3}$ & $2.3\times 10^{-2}$ \\
\hline
\end{tabular}
\end{center}
\caption {
\label{Table Fin} {Comparison between different metrics to measure the cofactor conditions for type I and type II twins in \Zn\, and Ti\textsubscript{74}Nb\textsubscript{23}Al\textsubscript{3}. In the last two lines, we used the new metric introduced in \eqref{NewMetric1}--\eqref{NewMetric2}.
}}
\end{table}

\begin{figure}

\tikzstyle{abstract}=[rectangle, draw=black, 
text centered, anchor=north, text=black, text width=7.5cm]

\tikzstyle{myarrow}=[<-, >=open triangle 90, thick]
\tikzstyle{line}=[-, thick]
        
\begin{center}
\begin{tikzpicture}[node distance=2cm]

    \node (ItemA) [abstract]
        {Find the unique $\hat{\vc e}$ such that \\$\mt U = (2\hat{\vc e}\otimes\hat{\vc e} -\mt 1)\mt V(2\hat{\vc e}\otimes\hat{\vc e} -\mt 1)$};
    \node (BN) [text width=4cm, below=of ItemA] {};
\begin{scope}[node distance=2.5 and 4.5cm,on grid]   
 \node (BA) [abstract, left=of BN]
        {Construct $$\vc m = \vc e,\qquad\vc b=2\Big(\frac{\mt U^{-1}\hat{\vc e}}{|\mt U^{-1}\hat{\vc e}|^2}-\mt U\hat{\vc e}\Big)$$ };
    \node (BB) [abstract, right=of BN]
        {
        {Construct $$\vc m = 2\Big(\hat{\vc e}-\frac{\mt U^{2}\hat{\vc e}}{|\mt U\hat{\vc e}|^2}\Big),\qquad\vc b=\mt U\hat{\vc e}$$ }
                };
 \node (CA) [abstract, below=of BA]
        {Construct $$\mt E_* : =  \mt U^2 -  \mt 1  -\frac1{|\vc b|^2}\mt U\vc b\otimes \mt U\vc b + \frac{\mt U^{-1}\vc b}{|\mt U^{-1}\vc b|}\otimes\frac{\mt U^{-1}\vc b}{|\mt U^{-1}\vc b|}$$ 
 };

 \node (CB) [abstract, below=of BB]
        {Construct $$\mt C_* : =  \mt U^2 -  \mt 1  + (1+|\mt U\vc m|^2) \vc m\otimes \vc m - \mt U^2\vc m\odot\vc m $$
        where $ \mt U^2\vc m\odot\vc m = \mt U^2\vc m\otimes\vc m+ \vc m \otimes \mt U^2\vc m$
 };
 
  \node (DA) [abstract, below=of CA]
        {Compute the eigenvalues of $\mt E_*$, namely $\lambda_1^{\mt E_*},\lambda_2^{\mt E_*},\lambda_3^{\mt E_*}$ (one should be zero)
 };

  \node (DB) [abstract, below=of CB]
        {Compute the eigenvalues of $\mt C_*$, $\lambda_1^{\mt C_*},\lambda_2^{\mt C_*},\lambda_3^{\mt C_*}$ (one should be zero) 
 };
 
  \node (EA) [abstract, below=of DA]
        {Compute $$\max\{|\lambda_1^{\mt E_*}-\lambda_2^{\mt E_*}|,|\lambda_1^{\mt E_*}-\lambda_3^{\mt E_*}|,|\lambda_2^{\mt E_*}-\lambda_3^{\mt E_*}|\}	$$ a measure of how closely the cofactor conditions are satisfied by type I twins
        };

  \node (EB) [abstract, below=of DB]
        {Compute $$\max\{|\lambda_1^{\mt C_*}-\lambda_2^{\mt C_*}|,|\lambda_1^{\mt C_*}-\lambda_3^{\mt C_*}|,|\lambda_2^{\mt C_*}-\lambda_3^{\mt C_*}|\}	$$ a measure of how closely the cofactor conditions are satisfied by type II twins
 };

    \draw[myarrow] (BA.north) -- ++(0,0.8) -| (ItemA.south);
    \draw[myarrow] (BB.north) -- ++(0,0.8) -| (ItemA.south);
        \draw[myarrow] (CA.north) -- ++(0,0.4) -| (BA.south);
        \draw[myarrow] (CB.north) -- ++(0,0.2) -| (BB.south);
                \draw[myarrow] (DA.north) -- ++(0,0.8) -| (CA.south);
        \draw[myarrow] (DB.north) -- ++(0,0.7) -| (CB.south);
                \draw[myarrow] (EA.north) -- ++(0,0.3) -| (DA.south);
        \draw[myarrow] (EB.north) -- ++(0,0.3) -| (DB.south);
        \end{scope}
\end{tikzpicture}
\end{center}
\caption {\label{Schema Fin} {Algorithm to compute how closely a type I/II twin generated by the martensite variants $\mt U,\mt V$ satisfies the cofactor conditions. The proposed metric replaces the measurement of (CC1)--(CC3). A possible way to compute $\hat{\textbf{e}}$
satisfying \eqref{e hat} can be found in \cite[Appendix A]{JamesHyst} (see also Table \ref{Table 1} and Table \ref{Table 2} respectively for cubic to monoclinic and cubic to orthorhombic transformations). Our new algorithm can be applied also to measure how closely a compound twin can form triple junctions with austenite. In this case, however, there exists two different $\hat{\vc e}$ such that $\mt U = (2\hat{\vc e}\otimes\hat{\vc e} -\mt 1)\mt V(2\hat{\vc e}\otimes\hat{\vc e} -\mt 1)$. One thus has to run the algorithm twice, once with each $\hat{\vc e}$. For compound twins in cubic to monoclinic II transformations (see \eqref{cubictomono}), a simple way to check how closely (CC1)--(CC2) are satisfied is to compute $|d-1|$ (cf. Remark \ref{Solod1ecof}).
}}
\end{figure}
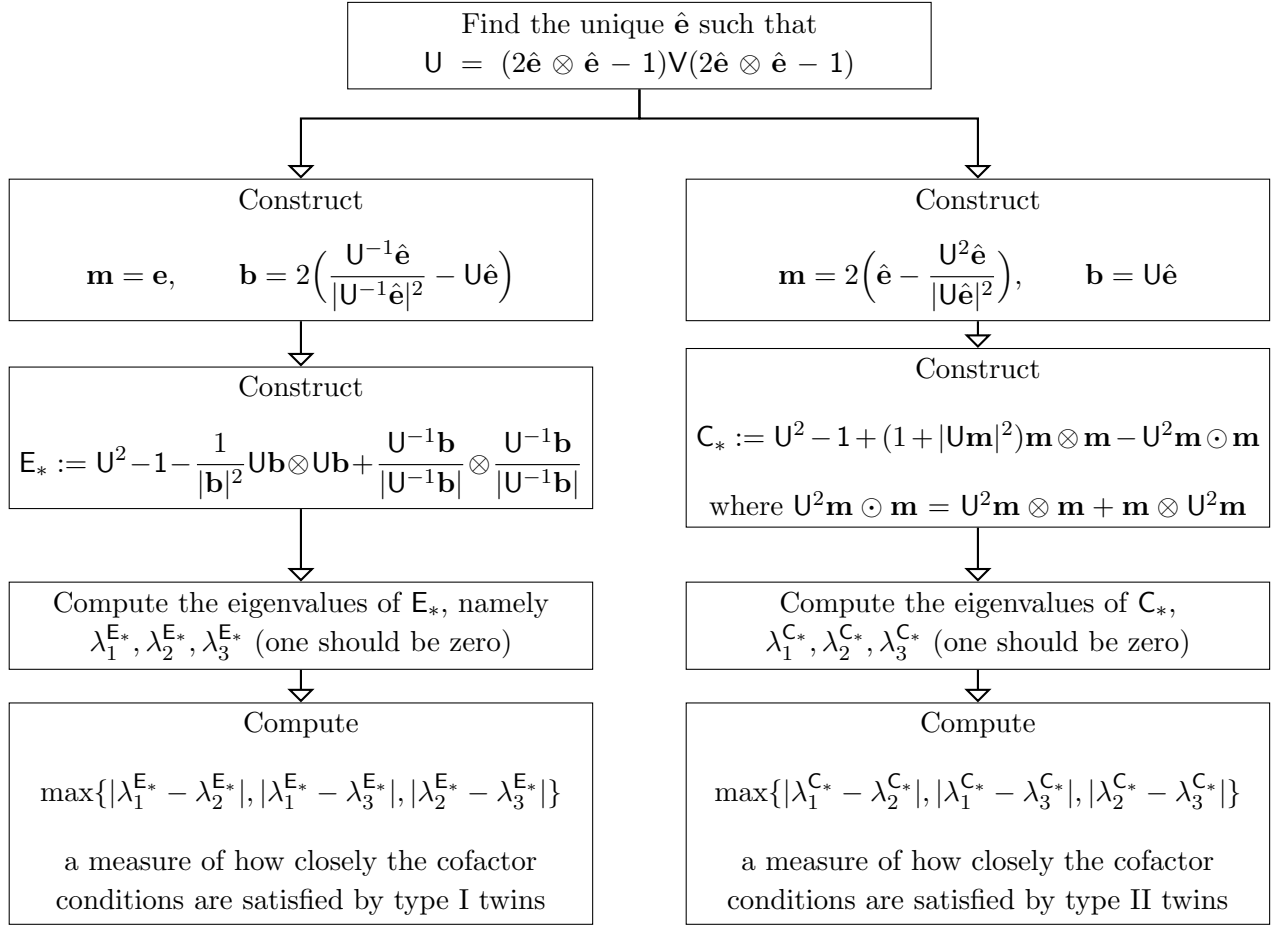

\footnotesize
\bibliographystyle{plain}
\bibliography{biblio}

\end{document}